\documentclass{LMCS}

\def\doi{9(1:12)2013}
\lmcsheading%
{\doi}
{1--56}
{}
{}
{Aug.~\phantom02, 2011}
{Mar.~20, 2013}
{}

\usepackage[textsize=footnotesize,color=blue!40]{todonotes}

\usepackage{amssymb}

\usepackage{xspace}

\usepackage[all]{xy}

\usepackage{verbatim} 

\usepackage{enumerate}
\usepackage{hyperref}

\hypersetup{
pdfauthor={Alexander Kartzow},
pdftitle={Collapsible Pushdown Graphs of Level 2 Are Tree-Automatic},
pdfkeywords={collapsible pushdown
  graphs, tree-automatic structures, first-order logic,
decidability, regular reachability}
}

\newcommand*{\N}[0]{\ensuremath{\mathbb{N}}}

\newcommand*{\domain}[0]{\ensuremath{\text{dom}}}

\newcommand*{\Push}[1]{\ensuremath{\mathrm{push}_{#1}}}
\newcommand*{\Pop}[1]{\ensuremath{{\mathrm{pop}_{#1}}}}
\newcommand*{\Clone}[1]{\ensuremath{{\mathrm{clone}_{#1}}}}
\newcommand*{\Collapse}[0]{\ensuremath{\mathrm{collapse}}}
\newcommand*{\TOP}[1]{\ensuremath{\mathrm{top}_{#1}}}
\newcommand*{\Sym}[0]{\ensuremath{\mathrm{Sym}}}
\newcommand*{\Lvl}[0]{\ensuremath{\mathrm{CLvl}}}
\newcommand*{\Lnk}[0]{\ensuremath{\mathrm{CLnk}}}
\newcommand*{\Op}[0]{\mathrm{OP}}
\newcommand*{\op}[0]{\mathrm{op}}

\newcommand*{\ColPop}[1]{\ensuremath{\mathrm{ColPop}_{#1}}}

\newcommand*{\Stacks}[0]{\mathrm{Stck}}

\newcommand*{\length}[0]{\mathrm{length}}
\newcommand*{\depth}[1]{\mathrm{dpt}({#1})}
\newcommand*{\height}[0]{\mathrm{hgt}}

\newcommand*{\Decode}[0]{\mathrm{Dec}}
\newcommand*{\Encode}[0]{\mathrm{Enc}}
\newcommand*{\bfEncode}[0]{\mathrm{\mathbf{Enc}}}

\newcommand*{\Trees}[1]{\ensuremath{\mathbb{T}_{#1}}}
\newcommand*{\EncTrees}[0]{\ensuremath{\mathbb{T}^{\mathrm{Enc}}}}

\newcommand*{\LeftTree}[1]{\ensuremath{LT({#1})}}

\newcommand*{\LeftStack}[0]{\ensuremath{\mathrm{LStck}}}
\newcommand*{\ExLeftStack}[0]{\ensuremath{\mathrm{exLStck}}}
\newcommand*{\InducedGenMilestone}[0]{\ensuremath{\mathrm{IgM}}}

\newcommand*{\inducedTreeof}[2]{\ensuremath{{#2}_{#1}}}

\newcommand*{\CPS}[0]{\ensuremath{\mathrm{CPS}}\xspace}
\newcommand*{\CPG}[0]{\ensuremath{\mathrm{CPG}}\xspace}

\newcommand*{\MSO}{\ensuremath{\mathrm{MSO}}\xspace}
\newcommand*{\FO}{\ensuremath{\mathrm{FO}}\xspace}

\newcommand*{\Reach}{\ensuremath{\mathrm{Reach}}}

\newcommand*{\coloneqq}[0]{\ensuremath{\mathrel{\mathop:}=}}

\newcommand*{\trans}[1]{\ensuremath{\mathrel{{\vdash^{#1}}}}}

\newcommand*{\Milestones}[0]{\ensuremath{\mathrm{MS}}}
\newcommand*{\genMilestones}[0]{\ensuremath{\mathrm{GMS}}
}

\newcommand*{\treeLR}[3]{\ensuremath{#1\left\langle{#2};{#3}\right\rangle}}
\newcommand*{\treeL}[2]{\ensuremath{#1\left\langle{#2};\emptyset\right\rangle}}
\newcommand*{\treeR}[2]{\ensuremath{#1\left\langle\emptyset;{#2}\right\rangle}}

\newcommand*{\prefixeq}[0]{\ensuremath{\mathop{\trianglelefteq}}}
\newcommand*{\notprefixeq}[0]{\ensuremath{\mathop{\not\trianglelefteq}}}

\newcommand*{\ExRet}[0]{\ensuremath{\mathrm{Rt}}}
\newcommand*{\ExLoop}[0]{\ensuremath{\mathrm{Lp}}}
\newcommand*{\ExHLoop}[0]{\ensuremath{\mathrm{hLp}}}
\newcommand*{\ExLLoop}[0]{\ensuremath{\mathrm{{\ell}Lp}}}
\newcommand*{\ExOneLoop}[0]{\ensuremath{1\mathrm{Lp}}}

\newcommand*{\Conf}[0]{\mathrm{Cnf}}

\newcommand*{\Runs}[0]{\ensuremath{\mathrm{Runs}}}

\newcommand*{\relA}[0]{\ensuremath{R^{\Leftarrow}}}
\newcommand*{\relB}[0]{\ensuremath{R^{\Downarrow}}}
\newcommand*{\relC}[0]{\ensuremath{R^{\Uparrow}}}
\newcommand*{\relD}[0]{\ensuremath{R^{\Rightarrow}}}

\begin{document}

\title[Tree-Automaticity of \texorpdfstring{$2$}{2}-CPG]{Collapsible
  Pushdown Graphs of Level 
  \texorpdfstring{$\mathbf{2}$ }{2} are Tree-Automatic\rsuper*}

\author[A.~Kartzow]{Alexander Kartzow}
\address{Universit{\"a}t Leipzig, Institut f\"ur Informatik,
Augustusplatz 10,  04103 Leipzig, Germany}
\email{kartzow@informatik.uni-leipzig.de}

\keywords{tree-automatic structures, collapsible pushdown graphs,
  collapsible pushdown systems, 
  first-order logic, decidability, reachability}

\ACMCCS{[{\bf  Theory of computation}]:  Formal languages and automata
  theory---Tree languages;  Logic---Higher order logic} 
\subjclass{F.4.1}
\titlecomment{{\lsuper*}A preliminary version of this paper has
  been presented at STACS'10
  \cite{Kartzow10}.}

\begin{abstract}
  \noindent  We show that graphs generated by collapsible pushdown
  systems of 
  level $2$ are tree-automatic. Even if we allow
  $\varepsilon$-contractions and reachability
  predicates (with regular constraints) for pairs of configurations,
  the structures remain 
  tree-automatic whence their first-order logic theories are decidable.
  As a corollary we obtain the tree-automaticity of the second level
  of the Caucal-hierarchy. 
\end{abstract}

\maketitle
\tableofcontents
\section{Introduction}

Higher-order pushdown systems were first introduced by Maslov
\cite{Maslov74,Maslov76} as accepting devices for word languages. 
Later, Knapik et al.\ \cite{KNU02} studied them as generators for
trees. They obtained an equi-expressivity result for higher-order
pushdown systems and for higher-order recursion schemes that satisfy
the constraint of \emph{safety}, which is a rather unnatural syntactic
condition.  
Hague et al.\ \cite{Hague2008} introduced collapsible pushdown systems
as extensions of higher-order pushdown systems and proved that these
have exactly the same power as higher-order recursion schemes as
methods for generating trees. 

Both  higher-order and collapsible pushdown systems  also form
interesting devices for generating graphs. 
Carayol and W\"ohrle \cite{cawo03} showed that the graphs generated by
higher-order pushdown systems
of level $l$  coincide with the graphs in 
the $l$-th level of the Caucal-hierarchy, a class of graphs
introduced by Caucal \cite{Caucal02}. Every level of this hierarchy is
obtained from the preceding level by applying graph unfoldings
and monadic second-order interpretations. Both operations preserve the
decidability of the monadic second-order theory whence the
Caucal-hierarchy forms a 
large class of graphs with decidable monadic second-order theories. 
If we use collapsible pushdown systems as generators for graphs we
obtain a different situation. Hague et al.\ showed that even the second
level of the hierarchy contains a graph with undecidable
monadic second-order theory. Furthermore,  they showed the
decidability of the 
modal $\mu$-calculus theories of
all graphs in the hierarchy. These results turn graphs generated by
collapsible pushdown systems into an interesting class. 
The author only  knows one further natural class of graphs which
shares these two properties,
viz.\ the class of nested pushdown trees (cf.~\cite{Alur06languagesof}). Moreover 
this class can
be seen as a subclass of that of collapsible pushdown graphs (cf.~\cite{KartzowPhd}). 

This paper is the long version of
\cite{Kartzow10} and studies the first-order model-checking problem on
collapsible pushdown graphs. We show that the
graphs in the second level of the collapsible pushdown hierarchy are
tree-automatic. 
Tree-automatic structures were introduced by Blumensath
\cite{Blumensath1999}. 
These structures enjoy decidable first-order theories due to the good
closure properties of finite automata. 
Since the translation from collapsible pushdown systems into
tree-automata presentations of the generated graphs is uniform, 
our result
implies that first-order model-checking on collapsible pushdown graphs
of level $2$
is decidable: given a pushdown system, first compute the
tree-automata representing its graph, then apply classical model-checking
for tree-automatic structures. 

Moreover, the result still holds if
regular reachability predicates are added to the graphs. 

\subsection{Main Result}

\begin{thm}\label{ThmCPGTreeAutomatic}
  Let $\mathcal{S}$ be  a collapsible pushdown system  of level $2$
  with configuration graph $\mathfrak{G}$. Let
  $\mathfrak{G}/\varepsilon$ be the $\varepsilon$-contraction of 
  $\mathfrak{G}$. 
  Any expansion of $\mathfrak{G}/\varepsilon$ by  regular reachability
  relations is tree-automatic.\footnote{
    Due to Broadbent et al.\ \cite{BroadbentCOS10}, the result still
    holds if we expand the graph by $L\mu$ definable predicates. 
  }
\end{thm}

A regular reachability relation is of the form $\Reach_L$ for some
regular language $L$. For nodes $a,b$ of some graph $\mathfrak{G}$
with labelled edges, $\mathfrak{G}\models\Reach_L(a,b)$ if there is a
path from $a$ to $b$ which is labelled by some word $w \in L$. 
The translation from collapsible pushdown systems to tree-automatic
presentations is uniform, i.e., there is a uniform way of computing, 
given a collapsible pushdown system
(and finite automata representing regular languages over the
edge-alphabet of the system), the tree-automata presentation of the
$\varepsilon$-contraction of the generated graph (expanded by the
regular reachability predicates). 
Once we have obtained tree-automata representing some graph, first-order
model-checking on this graph is decidable. Combining these results we
obtain the following corollary. 

\begin{cor}
  The following problem is decidable:\\
  \emph{Input:} a collapsible pushdown system $\mathcal{S}$ (of level $2$),
  finite automata $\mathcal{A}_1, \dots, \mathcal{A}_n$ representing
  regular languages $L_1, \dots, L_n$, and a formula $\varphi$ in
  first-order logic extended by the relations $\Reach_{L_1}, \dots,
  \Reach_{L_n}$\\ 
  \emph{Output:} $\mathfrak{G}/\varepsilon \models \varphi$? 
  ($\mathfrak{G}/\varepsilon$ denotes the $\varepsilon$-contraction of
  the graph generated by $\mathcal{S}$.)
\end{cor}

We also show that the decision procedure is necessarily nonelementary in the
size of the formula. 

\subsection{Outline of the Paper}
Sections \ref{Sec:Logics} and \ref{Sec:wordsAndTrees} introduce basic
notation. In Section  \ref{STACS:SecCPG} we introduce collapsible
pushdown graphs (of level $2$) and we show basic properties of these
graphs. We recall the notion of tree-automaticity in Section
\ref{sec:FiniteAutomata}. 
We present our encoding of configurations of collapsible pushdown graphs as
trees in Section \ref{sec:CPGTreeAut}. 
We also show
that, once we have proved that  regular reachability is tree-automatic
via this encoding,  collapsible pushdown graphs are tree-automatic. The final
part of that Section discusses the optimality of the first-order
model-checking  algorithm obtained from this tree-automata approach. 
Sections \ref{Sec:DecompositionOfRuns} and
\ref{sec:RegularReach} complete the proof by showing
that regular reachability is actually tree-automatic via our encoding.
In Section \ref{Sec:DecompositionOfRuns} we develop
the technical machinery for proving the regularity of the reachability
relation. We analyse how arbitrary runs of collapsible pushdown systems
decompose as sequences of simpler runs. 
Afterwards, in Section
\ref{sec:RegularReach}  we apply these results and
show that the  encoding  turns the regular reachability relations into
tree-automatic relations. 
Finally, Section \ref{sec:Conclusion} contains concluding remarks.

\section{Preliminaries and Basic Definitions}

For a function $f$, we denote by $\domain(f)$ its domain. 

\subsection{Logics}
\label{Sec:Logics}
We denote by $\FO{}$ first-order logic. 
Given some graph $\mathfrak{G}=(V, E_1, E_2, \dots, E_n)$ with
labelled edge relations $E_1, \dots, E_n\subseteq V\times V$ 
we denote by $\Reach$ the \emph{reachability predicate}, defined by
\begin{align*}
  \Reach:=\left\{(a,b):\text{ there is a }(\bigcup_{i=1}^n E_i)\text{-path
    from }a\text{ to }b\right\}.
\end{align*}
Given some regular language $L$, we denote by 
\begin{align*}
  \Reach_L:=\left\{(a,b):\text{ there is a }(\bigcup_{i=1}^n
  E_i)\text{-path from }a\text{ to }b\text{ labelled 
  by some word }w\in L\right\}
\end{align*}
the \emph{reachability predicate with respect to $L$}.

\subsection{Words and Trees}
\label{Sec:wordsAndTrees}

For words $w_1,w_2\in\Sigma^*$ we write $w_1 \leq w_2$ if $w_1$ is a
prefix of $w_2$. $w_1\sqcap w_2$ denotes the greatest common prefix of
$w_1$ and $w_2$. The concatenation of $w_1$ and $w_2$ is denoted by
$w_1w_2$.  

We call a finite set $D\subseteq \{0,1\}^*$ a \emph{tree domain}, if
$D$ is prefix closed. 
A \emph{$\Sigma$-labelled tree} is a mapping $T:D\rightarrow \Sigma$
for $D$ some tree domain. 
For $d\in D$ we denote the \emph{subtree rooted at $d$} by
$\inducedTreeof{d}{T}$. This is the tree defined by
$\inducedTreeof{d}{T}(e):= T(de)$. 
We will usually write $d\in T$ instead of $d\in\domain(T)$. 
We denote the \emph{depth} of the tree $T$ by 
\mbox{$\depth{T}\coloneqq\max\left\{\lvert t\rvert: t\in\domain(T)\right\}$}.


\label{HochPlusEinfuehrung}
For $T$ some tree with domain $D$,  let $D_+$ denote the set of minimal
elements of the complement of $D$, i.e., 
\begin{align*}
D_+=\{e\in\{0,1\}^*\setminus D:\text{ all proper ancestors of }e\text{
  are contained in }D\}.   
\end{align*}
We write $D^\oplus$ for $D\cup D_+$. 
Note that $D^\oplus$ is the extension of the tree domain $D$ by one
layer.  

Sometimes it is useful to define trees inductively by describing the
subtrees rooted at $0$ and $1$. For this purpose we fix the following
notation. 
Let $\hat T_0$ and $\hat T_1$ be $\Sigma$-labelled trees and
$\sigma\in\Sigma$. Then we write 
$T\coloneqq \treeLR{\sigma}{\hat T_0}{ \hat T_1}$
for the $\Sigma$-labelled tree $T$ with the following three
properties
\begin{align*}
  &1.\ T(\varepsilon) = \sigma,& &2.\ \inducedTreeof{0}{T} = \hat
  T_0 \text{, and }& &3.\ \inducedTreeof{1}{T} = \hat T_1.&
\end{align*}
We denote by $\Trees{\Sigma}$ the set of all
$\Sigma$-labelled trees. 

\subsection{Collapsible Pushdown Graphs}
\label{STACS:SecCPG}
Before we introduce collapsible pushdown graphs (CPG)
in
detail, we fix some notation. Then we informally explain collapsible
pushdown systems. Afterwards we formally introduce these
systems and the graphs generated by them. We conclude this section
with some basic results on runs of collapsible pushdown systems. 

We set $\Sigma^{+2}:={(\Sigma^+)}^+$ and $\Sigma^{*2}:={(\Sigma^*)}^*$. 
Each element of
$\Sigma^{*2}$ is called a $2$-word. Stacks of a collapsible
pushdown system are certain $2$-words from $\Sigma^{+2}$ over a special
alphabet. In analogy, we will call words also $1$-words. 

Let us fix a  $2$-word $s\in \Sigma^{*2}$.
$s$  consists of an ordered list 
\mbox{$w_1, w_2, \dots, w_m$} of words. 
If we want to state this list of words explicitly, we
write $s=w_1: w_2 : \dots : w_m$. 
Let $\lvert  s \rvert:=m $ denote the \emph{width} of $s$. 
The \emph{height} of $s$ is $\height(s):=\max\{\lvert
w_i \rvert: 1\leq i \leq m\}$ which is the length of the longest 
word occurring in $s$. 

Let $s'$ be another $2$-word with
$s'=w_1':w_2':\dots :w_l'$. We write 
$s:s'$ for the concatenation  
$w_1: w_2 : \dots :w_m:w_1':w_2':\dots:w_l'$. 

For some word $w$, let $[w]$ be the $2$-word that only consists of
$w$. We regularly omit the brackets if no confusion arises.  

A level $2$ stack $s$ over some alphabet $\Sigma$ is a $2$-word over
$\Sigma$ 
where each letter additionally carries a link to some $i$-word for
$1\leq i \leq 2$.
We call $i$ the level of the link. The idea is that the linked
$i$-word contains some information about what the stack looked like when the
letter was created.  

We first define the initial level $2$ stack; afterwards we
describe the stack operations that are used to generate all level $2$
stacks from the initial one. 

\begin{defi}
  Let $\Sigma$ be some finite alphabet with a distinguished
  bottom-of-stack symbol $\bot\in\Sigma$. 
  The \emph{initial stack} of level $1$ is the word
  $\bot_1\coloneqq \bot$.  The initial stack of level $2$ is  the 
  $2$-word
  $\bot_2:= [\bot_1]$. 
\end{defi}

We informally describe the stack operations that 
can be applied to a level $2$ stack. 
\begin{iteMize}{$\bullet$}
\item The  push operation of level $1$, denoted by $\Push{\sigma,k}$ for
  $\sigma\in\Sigma$ and $1\leq k \leq 2$, writes the symbol $\sigma$
  onto the topmost level $1$ stack and attaches a link of level
  $k$. This link points to a copy of the topmost 
  $k$-word of the resulting stack without the topmost $k-1$
  stack. For $k=2$ this means that the link points to the current
  stack where the  topmost word is removed. For $k=1$ the link points
  to the topmost word of the stack before the push operation was
  performed.  
\item The push operation of level $2$ is denoted
  by $\Clone{2}$. It duplicates the topmost word. 
  Since this also copies the values of the links stored in the topmost
  word, 
  the copy
  of each symbol in the newly created word still contains
  information what did the stack look like when the corresponding
  original symbol was pushed onto the stack. 
\item The level $i$ pop operation $\Pop{i}$ for $1\leq i\leq 2$ removes
  the topmost entry of the topmost $i$-word. 
\item The last operation is $\Collapse$. 
  The result of $\Collapse$ is determined by the link attached to the
  topmost letter of the stack. If we apply collapse to a stack $s$ where
  the link level of the topmost letter is $i$, then $\Collapse$
  replaces the topmost level $i$ stack of $s$ by the level $i$ stack
  to which the link points. 
  Due to how push operations create these links, 
  the application of a collapse is equivalent to the
  application of a sequence of  
  $\Pop{i}$ operations where the link of the topmost letter controls
  how long this sequence is.
\end{iteMize}
In the following, we formally introduce collapsible pushdown stacks
and the stack operations. 
We represent such a stack of letters with links as $2$-words over the
alphabet  
\mbox{$(\Sigma\cup (\Sigma\times\{ 2\}\times\N))^{+2}$}. We consider
elements from $\Sigma$ as elements 
with a link of level $1$ and elements $(\sigma,2,k)$ as letters with a
link of level $2$.
In the latter case, the third component specifies the width of the
substack to which the link  points. 
For letters with link of level $1$, the position of this letter within
the stack already determines the stack to which the link points. Thus,
we need not explicitly specify the link in this case.  
\begin{rem}
  Other equivalent definitions, for instance in \cite{Hague2008}, use
  a different 
  way of storing the links: they also store symbols $(\sigma,i,n)$ on
  the stack, but here $n$ denotes the number of $\Pop{i}$ transitions
  that are equivalent to performing the collapse operation at a stack
  with topmost element $(\sigma,i,n)$. The disadvantage of that
  approach is that the $\Clone{i}$ operation  cannot copy
  stacks. Instead, it can only copy the symbols stored
  in the topmost 
  stack and has to alter the links in the new copy. A clone of level
  $i$ must replace all links $(\sigma,i,n)$ by $(\sigma,i,n+1)$ in
  order to preserve the links stored in the stack. 
\end{rem}

We introduce some auxiliary functions which are useful for defining
the stack operations. 
\begin{defi}
  For $s=w_1:w_2:\dots: w_n\in(\Sigma\cup(\Sigma\times\{
  2\}\times\N))^{+2}$, and $w_n=a_1 a_2 \dots a_m$ we define the  
  following  auxiliary functions:
  \begin{iteMize}{$\bullet$}
  \item  $\TOP{2}(s):=w_n$ and
    $\TOP{1}(s):=a_m$.
  \item The \emph{topmost symbol} is
    \mbox{$\Sym(s)\coloneqq\sigma$} for $\TOP{1}(s) = \sigma\in\Sigma$
    or $\TOP{1}(s) = (\sigma,2,k)\in\Sigma\times\{2\}\times \N$.
  \item  The  \emph{collapse level of the
      topmost element} is
    $\Lvl(s)\coloneqq
    \begin{cases}
      1 &\text{if }\TOP{1}(s)\in\Sigma,\\
      2 & \text{otherwise.}      
    \end{cases}
    $
  \item Set
    $\Lnk(s)\coloneqq
    \begin{cases}
      j &\text{if }\TOP{1}(s)\in\Sigma\times\{2\}\times\{j\},\\
      m -1 &\text{if } \TOP{1}(s)\in\Sigma.
    \end{cases}$\\
    $\Lnk(s)$ is called the \emph{collapse link of the topmost element}.
  \item Set
    $\mathrm{p}_{\sigma,2,k}(w_m)\coloneqq w_m(\sigma,2,k)$ and
    $\mathrm{p}_{\sigma,2,k}(s)\coloneqq 
    w_1: w_2 : \dots : w_{n-1}: \mathrm{p}_{\sigma,2,k}(w_n)$ for all
    $k\in\N$.
  \end{iteMize}\vspace{3 pt}
\end{defi}

\noindent These auxiliary function are useful for the formalisation of the stack
operations. 
\begin{defi}
  For $s=w_1:w_2:\dots: w_n\in(\Sigma\cup(\Sigma\times\{
  2\}\times\N))^{+2}$ and $w_n=a_1 a_2\dots a_m$,
  for $\sigma\in\Sigma\setminus\{\bot\}$ and
  for $1\leq k \leq 2$, we define
  the stack operations 
  \begin{align*}
    \Clone{2}(s)\coloneqq
    &w_1: w_2 : \dots : w_{n-1}: w_n : w_n,\\
    \Push{\sigma,k}(s)\coloneqq
    &\begin{cases}
      w_1: w_2 : \dots : w_{n-1}: w_n\sigma  &\text{if } k=1,\\
      \mathrm{p}_{\sigma,k,n-1}(s) &\text{if } k = 2,
    \end{cases} \\
    \Pop{k}(s)\coloneqq
    &\begin{cases}
      w_1: w_2 : \dots : w_{n-1} & \text{if } k=2, n>1,\\
      w_1: w_2 : \dots : w_{n-1}: [a_1 a_2 \dots a_{m-1}] & \text{if }
      k=1, m >1,\\ 
      \text{undefined} & \text{otherwise},
    \end{cases} \\
    \Collapse{}(s)\coloneqq 
    &\begin{cases}
      w_1: w_2 : \dots : w_k & \text{if } \Lvl(s)=2, \Lnk(s)=k>0,\\
      \Pop{1}(s) & \text{if }
      \Lvl(s)=1,  m >1,\\
      \text{undefined}&\text{otherwise.}
    \end{cases}    
  \end{align*}
  The \emph{set of level $2$ operations} is 
  \begin{align*}
    \Op\coloneqq\{(\Push{\sigma,k})_{\sigma\in\Sigma,1\leq k\leq 2},
    \Clone{2}, (\Pop{k})_{1\leq k\leq 2}, \Collapse{}\}.
  \end{align*}
  The \emph{set of (level 2) stacks}  $\Stacks(\Sigma)$ is the
  smallest set 
  that contains $\bot_2$ and is closed under application of operations
  from $\Op$. 
\end{defi}

Note that we defined $\Collapse$ and $\Pop{k}$ in such a way that the 
the resulting stack is always nonempty and does not contain empty
words. This avoids the special treatment of 
empty stacks. Note that there is no $\Clone{1}$ operation. Thus, 
any $\Collapse$ that works on level $1$ is equivalent to one $\Pop{1}$
operation because level $1$ links always point to the preceding letter.
Every
$\Collapse$ that works on level $2$ is equivalent to a
sequence of $\Pop{2}$ operations. Next, we introduce the substack
relation.
\begin{defi}
  Let $s,s'\in\Stacks(\Sigma)$. We say that $s'$ is a \emph{substack} of
  $s$ (denoted as $s'\leq s$) if
  there are  $n,m\in\N$  such that
  $s' = \Pop{1}^{n}(\Pop{2}^{m}(s))$.
\end{defi}

Having concluded the definitions concerning stacks and stack
operations, it is time to introduce collapsible pushdown systems.

\begin{defi}
  A \emph{collapsible pushdown system} of level $2$ (\CPS) is
  a tuple 
  \begin{align*}
    \mathcal{S} = (Q,\Sigma, \Gamma, \Delta, q_0)    
  \end{align*}
  where $Q$ is a
  finite set of states, $\Sigma$  a
  finite stack  alphabet with a distinguished bottom-of-stack symbol
  $\bot\in\Sigma$, $\Gamma$  a finite 
  input alphabet,  $q_0\in Q$ the initial state, and 
  \begin{align*}
    \Delta\subseteq
    Q\times \Sigma \times\Gamma \times Q \times \Op
  \end{align*}
  the transition relation.

  Every  $(q,s)\in\Conf(Q, \Sigma):=Q\times\Stacks(\Sigma)$ is called 
  a \emph{configuration} and   $\Conf(Q,\Sigma)$ is called 
  the \emph{set of configurations}.\footnote{We
    write $\Conf$ instead of $\Conf(Q, \Sigma)$ if $Q$ and $\Sigma$
    are clear from the context.}
  We define  $\gamma$-labelled
  \emph{transitions} $\trans{\gamma}\subseteq \Conf\times\Conf$ as
  follows: 
  $ (q_1,s) \trans{\gamma} (q_2, t)$ 
  if there is
  a $(q_1, \sigma, \gamma, q_2, op)\in\Delta$ such that $\op(s)=t$ and
  \mbox{$\Sym(s) = \sigma$}. 

  We call $\trans{}:=\bigcup_{\gamma\in\Gamma} \trans{\gamma}$ the
  \emph{transition relation} of $\mathcal{S}$. 
  We set $C(\mathcal{S})$ to be the set of all
  configurations that are reachable from $(q_0,\bot_2)$ via
  $\trans{}$. These configurations are called \emph{reachable}.
  The 
  \emph{collapsible pushdown graph (\CPG) generated by $\mathcal{S}$} is 
  \begin{align*}
    \CPG(\mathcal{S})\coloneqq\big(C(\mathcal{S}), (C(\mathcal{S})^2\cap
    \trans{\gamma})_{\gamma\in\Gamma}\big).
  \end{align*}
\end{defi}

\begin{exa} \label{exa:CPG}
  The following example (Figure \ref{fig:CPGExample}) of a collapsible
  pushdown graph $\mathfrak{G}$ 
  of level $2$  is taken from 
  \cite{Hague2008}. Let $Q\coloneqq\{0,1,2\}, \Sigma\coloneqq
  \{\bot,a\}$, 
  $\Gamma:=\{\mathrm{Cl}, A, A', P, \mathrm{Co}\}$.
  $\Delta$ is given by $(0,-,\mathrm{Cl},1,\Clone{2})$,
  $(1,-,A,0,\Push{a,2})$, $(1,-,A',2,\Push{a,2})$,
  $(2,a,P,2,\Pop{1})$, and $(2,a,\mathrm{Co},0,\Collapse)$, where $-$
  denotes any 
  letter from $\Sigma$. \\
  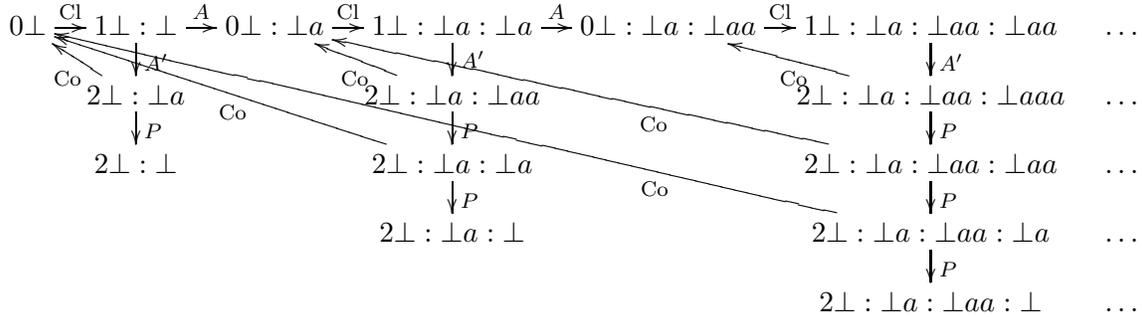
\begin{figure}[ht]
    \centering 
    $
    \begin{xy}
      \xymatrix@R=12pt@C=8.6pt{ 
        0 \bot \ar[r]^-{\mathrm{Cl}}& 
          1 \bot:\bot \ar[r]^-{A} \ar[d]^{ A'}& 
          0 \bot:\bot a \ar[r]^-{\mathrm{Cl}} &
          1 \bot:\bot a:\bot a  \ar[r]^-{A} \ar[d]^-{A'}& 
          0\bot:\bot a: \bot aa \ar[r]^-{\mathrm{Cl}}  
          & 1 \bot:\bot a : \bot aa:\bot aa \ar[d]^{A'}& \dots \\
          & 
          2 \bot:\bot a  \ar[d]^P \ar[ul]^{\mathrm{Co}}&  
          & 
          2 \bot:\bot a:\bot aa \ar[d]^P \ar[ul]^{\mathrm{Co}}&  
          & 
          2 \bot:\bot a : \bot aa:\bot aaa\ar[d]^P
            \ar[ul]^{\mathrm{Co}}
          & \dots \\
        & 
          2 \bot:\bot  &  
          &
          2 \bot:\bot a:\bot a  \ar[d]^P \ar[uulll]^{\mathrm{Co}}&  
          & 
          2 \bot:\bot a : \bot aa:\bot aa
          \ar[d]^P\ar[uulll]^(.415){\mathrm{Co}}
          & \dots \\
        &                &  
          & 
          2 \bot:\bot a:\bot  &  
          & 
          2 \bot:\bot a : \bot aa:\bot a \ar[d]^P\ar[uuulllll]^(.3){\mathrm{Co}}
          & \dots \\
        &              
          &  
          &
          &
          &
          2 \bot:\bot a : \bot aa:\bot  & \dots 
      }
    \end{xy}
    $
    \caption{Example of the  $2$-\CPG $\mathfrak{G}$ (the level
      $2$ links of the letters $a$  are omitted in the representation
      of the stacks).} 
    \label{fig:CPGExample}
  \end{figure}
\end{exa}

Hague et al.\ \cite{Hague2008} already noted that the previous example
has undecidable \MSO theory because the half grid $\{(n,m)\in\N^2:
n>m\}$ is \MSO interpretable in this graph (note that the collapse
edges of two vertices point to the same target if and only if the
vertices are on 
the same diagonal of the grid).

Next we define the $\varepsilon$-contraction of a given collapsible
pushdown graph. From now on,  we always assume that the input alphabet
$\Gamma$ contains the symbol $\varepsilon$. 

\begin{defi}
  Let $\Gamma$ be some alphabet. Let
  $L$ and $L_\gamma$  be the regular languages defined by the expressions
  $(\{\varepsilon\}^*(\Gamma\setminus\{\varepsilon\}))^*$ and
  ${\{\varepsilon\}^*\gamma}$, respectively.
  Given a collapsible pushdown graph $\mathfrak{G}$, the
  $\varepsilon$-contraction $\mathfrak{G}/\varepsilon$ of
  $\mathfrak{G}$ is the graph
  $(M,
  (\Reach_{L_\gamma})_{\gamma\in\Gamma\setminus\varepsilon})$ 
  where \mbox{$M:=\{g\in\mathfrak{G}: \mathfrak{G}\models
  \Reach_{L} ((q_0, \bot_2), g)\}$}.  
\end{defi}
\begin{rem}
  This is the usual definition of $\varepsilon$-contraction. An edge
  in the new graph consists of a sequence of $\varepsilon$-edges
  followed by one non-$\varepsilon$-edge. The set of configurations is
  then restricted to those configurations that are reachable via the new
  edges from the initial configuration. 
\end{rem}

Now we come to the notion of a run of a collapsible pushdown
system. 
\begin{defi}
  Let $\mathcal{S}$ be a collapsible pushdown system. 
  A \emph{run} $\rho$ of $\mathcal{S}$ is a sequence of configurations
  that are connected by transitions, i.e., a sequence
  \begin{align*}
    c_0 \trans{\gamma_1} c_1\trans{\gamma_2} c_2 \trans{\gamma_3} \cdots
    \trans{\gamma_n}c_n.    
  \end{align*}
  We denote by $\rho(i):=c_i$ the $i$-th configuration of $\rho$. 
  Moreover, we denote by $\length(\rho):=n$ the \emph{length} of
  $\rho$. 
  For $0\leq i \leq j \leq \length(\rho)$, 
  we write $\rho{\restriction}_{[i,j]}$ for the subrun 
  \begin{align*}
    c_i \trans{\gamma_{i+1}} c_{i+1} \trans{\gamma_{i+2}} \dots
    \trans{\gamma_j} c_j.    
  \end{align*}
  We write $\Runs(c,c')$ for the set of runs starting at $c$ and
  ending in $c'$. For $s,s'$ stacks, we also write
  $\Runs(s,s'):=\bigcup_{q, q'\in Q} \Runs( (q,s),(q',s'))$. 
\end{defi}

Consider some configuration $(q,s)$ of a \CPS. If $\lvert s \rvert=n$
then  
a $\Push{\sigma,2}$ transition applied to $(q,s)$ 
creates a letter with a link to the substack
of width $n-1$. 
Thus, links to the substack of width $n-1$ in some word above the
$n$-th one are always created by a $\Clone{2}$ operation. A
direct consequence of this fact 
is the following lemma.

\begin{lem} \label{Lem:Howtwolinksevolve}
 Let $s$ be some level $2$ stack with $\TOP{1}(s)=(\sigma,2,k)$. 
 Let $\rho \in \Runs(s,s)$ be a run that passes 
 $\Pop{1}(s)$. If $k<\lvert s \rvert -1$ then
 $\rho$  passes $\Pop{2}(s)$.  
\end{lem}
The proof is left to the reader. 
Later we often use the contraposition: 
if $\Lvl(s)=2$ and $\Lnk(s)<\lvert s \rvert-1$, then any run not
passing $\Pop{2}(s)$ does not pass $\Pop{1}(s)$.

Following the ideas of Blumensath \cite{Blumensath2008} for
higher-order pushdown systems, we introduce a prefix replacement for
collapsible pushdown systems. This replacement allows to copy runs
starting in one configuration into a run starting at another
configuration. 

\begin{defi}
  For $t\in\Stacks(\Sigma)$ and some substack $s\leq t$ we say that
  $s$
  is a \emph{prefix} 
  of $t$ and write 
  $s\prefixeq t$, if there are $n\leq m
  \in\N$ such that $s=w_1:w_2:\dots : w_{n-1}: w_n$ and 
  $t=w_1:w_2 \dots :w_{n-1} : v_n : v_{n+1} : \dots : v_m$ such that $w_n\leq
  v_j$ for all $n\leq j \leq m$. 
  
  For a configuration $c=(q,t)$, we write $s\prefixeq c$ as an
  abbreviation for $s\prefixeq t$.  
  For  some run $\rho$, we write $s\prefixeq \rho$ if
  $s\prefixeq\rho(i)$ for all $i\in\domain(\rho)$. 
\end{defi}

\begin{defi}
  Let $s,t,u$ be level $2$ stacks such that $s\prefixeq t$. Assume that 
  \begin{align*}
   &s=w_1:w_2:\dots : w_{n-1}: w_n,\\
   &t=w_1:w_2 \dots :w_{n-1} : v_n : v_{n+1} : \dots : v_m, \text{
     and}\\
   &u=x_1: x_2: \dots : x_{n-1}: x_n
  \end{align*}
  for numbers $n,m\in\N$ such that $n\leq m$. 
  For each $n\leq i \leq m$, let $\hat v_i$ be the unique word such that
  $v_i=w_n \hat v_i$. We define
  \begin{align*}
    t[s/u]:= x_1: x_2: \dots: x_{n-1} : (x_n \hat v_n) :
    (x_n \hat v_{n+1}): \dots : (x_n \hat v_m)
  \end{align*}
  and call $t[s/u]$ 
\emph{the stack obtained from $t$ by replacing the prefix $s$ by $u$}. 
\end{defi}
\begin{rem}
  Note that for $t$ some stack with level $2$ links, the resulting
  object $t[s/u]$ may be no stack. 
  Take for example the stacks 
  \begin{align*}
    &s=\bot (a,2,0) : \bot,\\ 
    &t=\bot (a,2,0)
    : \bot (a,2,0)\text{ and}\\ 
    &u=\bot:\bot.     
  \end{align*}
  Then $t[s/u]=\bot: \bot (a,2,0)$.
  This list of words
  cannot be created from the initial stack using the stack operation
  because an element $(a,2,0)$ in the second word has to be a clone of
  some element in the first one. But $(a,2,0)$ does not occur in the
  first word. Note that $t[s/u]$ is always a stack if $s=\Pop{2}^k(t)$
  for some $k\in\N$. 
\end{rem}

\begin{lem} \label{Lem:BlumensathLevel2}
  Let $\rho$ be a run of some collapsible pushdown system
  $\mathcal{S}$ of level $2$ and let $s$ and $u$ be stacks such that
  the following conditions are satisfied:
  \begin{enumerate}[\em(1)]
  \item $s\prefixeq \rho$,
  \item $\TOP{1}(u)=\TOP{1}(s)$, 
  \item $\lvert s \rvert = \vert u \rvert$, and
  \item for $\rho(0)=(q,t)$, $t[s/u]$ is a stack. 
  \end{enumerate}
  Under these conditions  the function 
  $\rho[s/u]$ defined by $\rho[s/u](i):=\rho(i)[s/u]$ is a run of
  $\mathcal{S}$.  
\end{lem}
\begin{proof}[Proof (sketch).]
  The proof is by induction on the length of $\rho$. 
  It is tedious but straightforward to prove that $\rho(i)[s/u]$ and
  $\rho(i)$ share the same topmost element. Thus, the transition
  $\delta$ 
  connecting $\rho(i)$ with $\rho(i+1)$ is also applicable to
  $\rho(i)[s/u]$. By case distinction on the stack operation one
  concludes that $\delta$ connects $\rho(i)[s/u]$  with
  $\rho(i+1)[s/u]$. 
\end{proof}

\subsection{Finite Automata and Automatic Structures}
\label{sec:FiniteAutomata}
In this section, we present
the basic theory of finite bottom-up tree-automata and tree-automatic
structures. For a more detailed introduction, we refer the reader to 
\cite{tata2007}. 

\begin{defi}
  A \emph{(finite tree-)automaton} is a tuple
  $\mathcal{A}=(Q,\Sigma,q_I,F,\Delta)$ 
  where $Q$ is a finite nonempty set of states, $\Sigma$ is a finite
  alphabet, $q_I\in Q$ is the initial  state, $F\subseteq Q$ is the
  set of final states,
  and $\Delta\subseteq Q\times  \Sigma \times Q\times Q$ is the
  transition relation.
\end{defi}

We next define the concept of a run of an  automaton on a tree.
\begin{defi}
  A \emph{run} of $\mathcal{A}$ on a binary $\Sigma$-labelled tree $t$
  is a map  
  $\rho:\domain(t)^\oplus \rightarrow Q$ such that  
  \begin{iteMize}{$\bullet$}
  \item 
    $\rho(d)=q_I $ for all $d\in\domain(t)_+$, and
  \item $\big(\rho(d), t(d), \rho(d0), \rho(d1)\big)\in \Delta$ for
    all $d\in\domain(t)$. 
  \end{iteMize}
  $\rho$ is \emph{accepting} if
  $\rho(\varepsilon)\in F$. 
  We say $t$ is accepted by $\mathcal{A}$ if there is an accepting run
  of $\mathcal{A}$ on $t$. 
  With each automaton $\mathcal{A}$, we associate the \emph{language} 
  \begin{align*}
    L(\mathcal{A}):=\{t: t\text{ is accepted by }\mathcal{A}\}    
  \end{align*}
  \emph{recognised} (or accepted)  by $\mathcal{A}$. 
  The class of languages accepted by automata is called the
  class of \emph{regular} languages. 
\end{defi}

Automata can be used to represent infinite structures. 
Such representations have good computational behaviour. 
In the following we recall the definitions and important results 
on tree-automatic structures. 

We first introduce the convolution of trees. This is a tool
for representing an $n$-tuple of $\Sigma$-trees as a single tree over the
alphabet $(\Sigma\cup\{\Box\})^n$ where $\Box$ is a padding symbol satisfying
$\Box\notin\Sigma$. 
\begin{defi}
  The \emph{convolution} of two
  $\Sigma$-labelled trees $t$ and $s$ is given by a function
  \begin{align*}
    t\otimes s : \domain(t)\cup\domain(s) \rightarrow
    (\Sigma \cup \{\Box\} )^2
  \end{align*}
  where $\Box$ is some new padding symbol, and
  \begin{align*}
    (t\otimes s)(d) \coloneqq
    \begin{cases}
      (t(d),s(d)) & \text{ if } d\in \domain(t)\cap
      \domain(s), \\ 
      (t(d), \Box) & \text{ if }d\in \domain(t)\setminus
      \domain(s), \\ 
      (\Box, s(d)) & \text{ if }d\in \domain(s) \setminus
      \domain(t).
    \end{cases}
  \end{align*}
  We also use the notation $\bigotimes(t_1, t_2, \dots, t_n)$ for 
  $t_1 \otimes t_2 \otimes \dots \otimes t_n$.
\end{defi}

Using convolutions of trees we can use a single automaton for
defining $n$-ary relations on a set of trees.
Thus, we can then use automata to represent a set and a tuple of
$n$-ary relations on this set. If we can represent the domain of some
structure and all its relations by automata, we call the structure
automatic.  
 
\begin{defi}
  We say a relation $R\subseteq {\Trees{\Sigma}}^n$ is automatic
  if there 
  is an automaton $\mathcal{A}$ such that 
  $L(\mathcal{A})=\{ \bigotimes(t_1 ,t_2, \dots,
  t_n)\in{\Trees{\Sigma}}^n: (t_1, t_2, \dots, t_n)\in R\}$. 

  A structure $\mathfrak{B}=(B,E_1, E_2, \dots, E_n)$ with relations
  $E_i$ is \emph{automatic} if there are automata $\mathcal{A}_B,
  \mathcal{A}_{E_1}, \mathcal{A}_{E_2}, 
  \dots, \mathcal{A}_{E_n}$ 
  and a bijection $f: L(\mathcal{A}_B)\rightarrow B$ such that
  for $c_1, c_2, \dots, c_n \in L(\mathcal{A}_B)$, the
  automaton $\mathcal{A}_{E_i}$
  accepts $\bigotimes(c_1, c_2, \dots, c_n)$ if and only if 
  \mbox{$(f(c_1), f(c_2), \dots, f(c_n))\in E_i$.}

  In other words, $f$ is a bijection between $L(\mathcal{A}_B)$ and $B$ and
  the automata $\mathcal{A}_{E_i}$ witness that the relations $E_i$
  are automatic via $f$.
  We call $f$ a tree presentation of $\mathfrak{B}$.
\end{defi}
Automatic structures form a nice class because automata theoretic
techniques may be used to decide first-order formulas on these
structures:

\begin{thm}[\cite{Blumensath1999}, \cite{Rubin2008}, \cite{KartzowPhd}]
  \label{Thm:FOTreeAutomaticDecidable} 
  If $\mathfrak{B}$ is automatic, then its 
  $\FO{}(\exists^{\mathrm{mod}}, \mathrm{Ram})$-theory is decidable.\footnote{
    $\FO{}(\exists^{\mathrm{mod}}, \mathrm{Ram})$ is the extension of
    $\FO{}$ by modulo counting quantifiers and by Ramsey-Quantifiers.}
\end{thm}

\section{Collapsible Pushdown Graphs are 
  Tree-Automatic} 
\label{sec:CPGTreeAut}
In Section \ref{STACS:SecEncoding} we present a bijection $\Encode$
between $\Conf$ and a regular set of trees. 
Moreover, $\Encode$ translates the  reachability predicates
$\Reach_L\subseteq \Conf\times\Conf$ for each regular language $L$ into
a tree-automatic relation. The proof of this claim which is the technical core
of this paper is developed in Sections \ref{Sec:DecompositionOfRuns} and
\ref{sec:RegularReach}. Before we present this proof, 
we show in Section \ref{subsec:ReachLAutomatic} how
this result can be used to prove our main theorem. Moreover, in
Section \ref{sec:LowerBounds} we discuss the optimality of the first-order
model-checking algorithm derived from this construction.

Regularity of the regular reachability predicates implies that
$\Encode{\restriction}_{\domain(\CPG(\mathcal{S})/\varepsilon)}$ is an
automatic presentation of $\CPG(\mathcal{S})/\varepsilon$ because its
domain and its transition relation can be defined as reachability relations
$\Reach_L$ for certain regular languages $L$. Note that for
the definition of the domain, we need
the encoding of the initial configuration as parameter. This parameter
can be hard-coded because its encoding is a fixed tree. 

\subsection{Encoding of Level 2 Stacks in Trees}
\label{STACS:SecEncoding}
 
In this section we present an encoding of level $2$ stacks in trees. 
The idea is to divide a stack into blocks and to
encode different blocks in different subtrees.
The crucial observation is that every stack is a list of words that
share the same first letter. A block is a maximal list of words occurring
in the stack which share the same two first letters. If we remove the first
letter of every word of such a block, the resulting  $2$-word
decomposes again as a list of blocks. Thus, we can inductively
carry on to decompose parts of a stack into blocks and encode every
block in a different subtree. The roots of these subtrees are
labelled with the first letter of the block. 
This results in a tree where every initial left-closed
path in the tree represents one word of the stack. 
A path of a tree is left-closed if its last element has no left
successor (i.e., no $0$-successor).

The following notation is useful for the formal definition of blocks. 
Let $w\in\Sigma^*$ be some word and \mbox{$s=w_1:w_2:\dots:w_n
\in\Sigma^{*2}$} some $2$-word. We
write $s'\coloneqq w\mathrel\backslash s$ 
for $s'=ww_1 : ww_2 : \dots: ww_n$. 
Note that $[w] \prefixeq (w\mathrel\backslash s)$, i.e., $[w]$ is a
prefix 
of $s'$. We say that $s'$
is 
\emph{$s$ prefixed by $w$}. 

\begin{figure}
  \centering
  $
  \begin{xy}
    \xymatrix@=0.2mm{
      & f \\
      & e & g & & i\\
      b& d & d & d & h & & &j & l\\
      a & c & c & c & c & c & c&c & k\\
      \bot & \bot & \bot & \bot & \bot & \bot & \bot & \bot&\bot
      \save "1,2"."4,4"*[F-]\frm{}
      \save "2,5"."4,5"*[F-]\frm{}
      \save "4,6"."4,6"*[F-]\frm{}
      \save "4,7"."4,7"*[F-]\frm{}
      \save "3,8"."4,8"*[F-]\frm{}
    }
  \end{xy}
  $
  \caption{A stack with blocks forming a $c$-blockline.}
  \label{STACSfig:Blocks}
\end{figure}
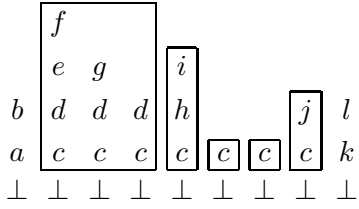

\begin{defi}
  Let $\sigma\in\Sigma$ and
  $b\in\Sigma^{+2}$. We call $b$  a \emph{$\sigma$-block} if
  $b=[\sigma]$ or $b=\sigma\tau\mathrel{\backslash} s'$ for some
  $\tau\in\Sigma$ and some $s'\in\Sigma^{*2}$.
  If $b_1, b_2,\dots, b_n$ are $\sigma$-blocks, then we call
  $b_1:b_2:\dots :b_n$ a \emph{$\sigma$-blockline}. 
\end{defi}

Note that every 
stack in $\Stacks(\Sigma)$ forms a $\bot$-blockline. 
Furthermore, every blockline $l$ decomposes uniquely as
\mbox{$l=b_1: b_2: \dots: b_n$} of maximal blocks $b_i$ in $l$. 

Another crucial observation is that a $\sigma$-block
$b\in\Sigma^{*2}\setminus \Sigma$
decomposes as $b=\sigma\mathrel{\backslash}l$ for some blockline $l$
and we call $l$  the  blockline \emph{induced} by $b$. For a block of
the 
form $[b]$ with $b\in\Sigma$, we define the
blockline induced by $[b]$ to be $\varepsilon$.

\begin{defi}
  Let  $l$ be a $\sigma$-blockline such that 
  $l = b_1:b_2:\dots:b_n$ is its decomposition into maximal blocks.  
  Let $i_1,i_2,\dots,i_m$ be those indices such that for all $1\leq j
  \leq n$ we have
  $b_{j}\neq  [\sigma]$ if and only if $j=i_k$ for some $1\leq k \leq
  m$.
  For $1\leq k \leq m$, let $b_{i_k}'$ be the $2$-word such that
  $b_{i_k}= \sigma \mathrel{\backslash} b'_{i_k}$. 
  We
  recursively define the \emph{blocks of $l$} to be the minimal set
  containing $b_1, b_2, \dots, b_n$ and 
  the blocks of each of  the $b_{i_k}$ ($1\leq k \leq m$) seen as
  $\tau$-blockline for some letter $\tau$. 
\end{defi}
See Figure \ref{STACSfig:Blocks} for an example of a stack with one
of its blocklines.

Recall that the symbols of a collapsible pushdown stack (of level $2$)
come from the set $\Sigma \cup (\Sigma \times\{2\}\times\N)$ where
$\Sigma$ is the stack alphabet. 
For $\tau\in\Sigma \cup (\Sigma \times\{2\}\times\N)$, we encode a
$\tau$-blockline $l$ in a tree as follows. The root of the tree is
labelled by $(\Sym(\tau), \Lvl(\tau))$. 
The
blockline induced by the first maximal block of 
$l$ is encoded in the left subtree  and  the
rest of $l$ is encoded in the right subtree. This means that we only
encode explicitly the symbol and the 
collapse level of each element of the stack, but not the
collapse link. We will later see how to decode the collapse links from the
encoding of a stack. 
When we
encode a part of a blockline in the right subtree, 
we do not repeat the label $(\Sym(\tau), \Lvl(\tau))$, but 
replace it by the empty word $\varepsilon$. 

\begin{defi}
  Let $\tau\in \Sigma\cup(\Sigma\times\{2\}\times\N)$. Furthermore,
  let 
  \begin{align*}
    s =
    w_1:w_2:\dots:w_n\in(\Sigma\cup(\Sigma\times\{2\}\times\N))^{+2}    
  \end{align*}
  be some  
  $\tau$-blockline. Let $w_i'$ be a word for each $1\leq i \leq n$ 
  such that 
  \mbox{$s= \tau \mathrel\backslash [w_1': w_2' :\dots
    :w_n']$} and set \mbox{$s'\coloneqq w_1': w_2' : \dots : w_n'$}. As an
  abbreviation we write  
  $_is_k\coloneqq w_i:w_{i+1}:\dots:w_k$.
  Let  
  $_1s_j$ be a maximal block of $s$. Note that $j>1$ implies 
  that there is some $\tau'\in
  \Sigma\cup(\Sigma\times\{2\}\times\N)$ and there are words $w''_{j'}$ for each
  $j'\leq j$ such that
  $w_{j'}= \tau\tau' w_{j'}''$. 

  For arbitrary
  $\sigma\in(\Sigma\times\{1,2\})\cup\{\varepsilon\}$, we define
  recursively the
  $(\Sigma\times\{1,2\})\cup\{\varepsilon\}$-labelled
  tree $\Encode(s,\sigma)$ via
  \begin{align*}
    \Encode(s,\sigma)\coloneqq
    \begin{cases}
      \sigma & \text{if } \lvert w_1\rvert=1, n=1\\
      \treeR{\sigma}{\Encode( _2s_n,\varepsilon)}
      &\text{if } \lvert w_1 \rvert =1, n>1\\
      \treeL{\sigma}{\Encode( _1s_n',(\Sym(\tau'), \Lvl(\tau')))}
      &\text{if }  \lvert w_1\rvert >1, j=n \\
      \treeLR{\sigma}{\Encode(
        _1s_j',(\Sym(\tau'),\Lvl(\tau')))}{\Encode( 
        _{j+1}s_n,\varepsilon)}      
      &\text{otherwise}
    \end{cases}
  \end{align*}
For every $s\in\Stacks(\Sigma)$, $\Encode(s)\coloneqq
\Encode(s,(\bot,1))$ is called the \emph{encoding of the stack} $s$.
\end{defi}
Figure \ref{STACSfig:Encoding} shows a configuration and its encoding. 
\begin{figure}[t]
  \centering
  $
  \begin{xy}
    \xymatrix@R=0pt@C=0pt{
      &       & (c,2,1) &         & e & \\
      &(b,2,0)& (b,2,0) &       c & (d,2,3) & \\
      &(a,2,0)& (a,2,0) & (a,2,2) & (a,2,2) & (a,2,2) \\
      & \bot  & \bot    &\bot     & \bot    & \bot
      }
  \end{xy}$
  \hskip 1cm
  $\begin{xy}
    \xymatrix@R=9pt@C=3pt{
           & c,2 &       & e,1      &  \\
      b,2 \ar[r]& \varepsilon\ar[u] & c,1 & d,2\ar[u] & \\
      a,2 \ar[u] &  & a,2\ar[r]\ar[u] & \varepsilon\ar[r]\ar[u]
      &\varepsilon \\
       \bot,1 \ar[rr]\ar[u] &  &\varepsilon\ar[u] &  & &      
      }
  \end{xy}
  $
  \caption{A stack $s$ and its Encoding $\Encode(s)$: right
    arrows lead to $1$-successors (right successors), upward arrows
    lead to $0$-successors (left successors).}
  \label{STACSfig:Encoding}
\end{figure}
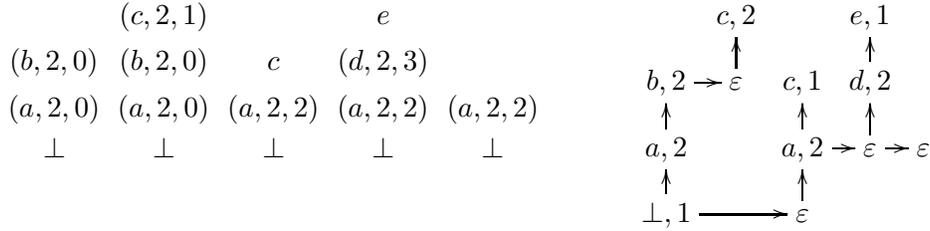

\begin{rem}
  Fix some stack $s$. 
  For $\sigma\in\Sigma$ and $k\in\N$, every $(\sigma,2,k)$-block of
  $s$ is
  encoded in a subtree whose root $d$ is labelled $(\sigma,2)$. 
  We can restore
  $k$ from the position of $d\in\{0,1\}^*0$ in the tree $\Encode(s)$
  as follows. 
  \begin{align*}
    k = \lvert \{d'\in \Encode(s) \cap \{0,1\}^*1:
    d'\leq_{\mathrm{lex}} d\}\rvert,     
  \end{align*}
  where $\leq_{\mathrm{lex}}$ is the lexicographic order. This
  is due to the fact that every right-successor corresponds to the
  separation of some block from
  some other. 

  This correspondence can be seen as a bijection.
  \label{milestonesInEncoding}
  Let \mbox{$s=w_1:w_2:\dots:w_n$} be some stack. We define the set 
  \mbox{$R\coloneqq \domain(\Encode(s)) \cap(\{\varepsilon\}\cup
    \{0,1\}^*1)$}. Then there is a bijection
  \mbox{$f: \{ 1,2, 3,   \dots, n\} \rightarrow R$} 
  such that $i$ is mapped to the $i$-th
  element of $R$ in lexicographic order. Each $1\leq i \leq n$
  represents the $i$-th word of $s$. 
  $f$ maps the
  first word of $s$ to the root of $\Encode(s)$ and every other word in
  $s$ to the element of $\Encode(s)$ that separates this word from its left
  neighbour in $s$. 

  If we interpret $\varepsilon$ as empty word, the word from the root
  to $f(i)$ in $\Encode(s)$ is the greatest common prefix of $w_{i-1}$
  and $w_i$. 
  More precisely,
  the word read along this path is the projection onto the letters and
  collapse levels of $w_{i-1}\sqcap w_i$.

  Furthermore, set $f'(i)\coloneqq f(i)0^m\in\Encode(s)$ 
  such that 
  $m$ is maximal with this property, i.e., $f'(i)$ is the leftmost
  descendant of $f(i)$. Then the path from $f(i)$ to $f'(i)$ is the
  suffix $w_i'$ such that $w_i=(w_{i-1}\sqcap w_i) w_i'$ (here
  we set $w_0\coloneqq \varepsilon$). More precisely, the word read along this
  path is the projection onto the symbols and collapse levels of $w_i'$.
\end{rem}

Having defined the encoding of a stack, we want to encode whole
configurations, i.e., a stack together with a state. To this end, we
just add the state as a new root of the tree and attach
the encoding of the stack as left subtree, i.e.,
for some configuration $(q,s)$ we set
\begin{align*}
  \Encode(q,s)\coloneqq \treeL{q}{\Encode(s)}.
\end{align*}

The image of this encoding function contains only trees of a very
specific type. We call this class $\EncTrees$. 
In the next definition we state the characterising properties of
$\EncTrees$. 
\begin{defi} \label{STACS:DefEncodingTrees}
  Let $\EncTrees$ be the class of trees $T$ that satisfy
  the following conditions.
  \begin{enumerate}[(1)]
  \item  The root of $T$ is labelled by some element of $Q$
    ($T(\varepsilon)\in Q$).
  \item $1\notin\domain(T)$,   $0\in\domain(T)$.
  \item $T(0)=(\bot,1)$. 
  \item Every element of the form $0\{0,1\}^*0$ is labelled by some
    $(\sigma,l)\in(\Sigma\setminus\{\bot\})\times\{1,2\}$, 
  \item Every element of the form $\{0,1\}^*1$ is labelled by
    $\varepsilon$. 
  \item \label{Cond:LevelOneBlockscoincide} 
    \label{STACS:fifthofDefEnc} 
    There is no $t\in T$ such that
    $T(t0) = (\sigma,1)$ and
    $T(t10) = (\sigma,1)$. 
  \end{enumerate}
\end{defi}
\begin{rem}
  Note that all trees in the image of $\Encode$ satisfy condition
  \ref{STACS:fifthofDefEnc} due to the following. 
  \mbox{$T(t0)=T(t10)=(\sigma,1)$} would imply that the 
  subtree rooted at $t$ encodes a blockline $l$ such that the first
  block $b_1$ of $l$ induces a $\sigma$-blockline and the second
  block $b_2$
  induces also a $\sigma$-blockline.  This
  contradicts the maximality of the blocks used in the encoding because
  all words of $b_1$ and $b_2$ have $\sigma$ as
  second letter whence $b_1:b_2$ forms a larger block.
  Note that for letters with links of level $2$ the analogous
  restriction does not hold. In Figure \ref{STACSfig:Encoding} one
  sees the encoding of a stack $s$ where $\Encode(s)(0) =
  \Encode(s)(10) = (a,2)$. Here, the label $(a,2)$ represents two
  different letters. $\Encode(s)(0)$ encodes the element $(a,2,0)$,
  while 
  $\Encode(s)(10)$ encodes the element $(a,2,2)$,
  i.e., the first element encodes a letter $a$ with undefined link and
  the second  encodes the  
  letter $a$ with a link to the substack of width $2$. 
\end{rem}

\begin{lem} \label{Lem:EncTreesAutomatic}
  There is a finite automaton $\mathcal{A}_{\EncTrees}$ with 
  $2+3\lvert\Sigma\rvert$ many states that recognises $\EncTrees$. 
\end{lem}
\begin{proof}
  Set
  $\mathcal{A}_{\EncTrees}:=
  (Q_{\mathcal{A}},
  Q\cup(\Sigma\times\{1,2\})\cup\{\varepsilon\}, \bot, \{q_I\},
  \Delta_{\mathcal{A}})$ where $Q_{\mathcal{A}}$
   and $\Delta_{\mathcal{A}}$ are defined as follows. 
   Let $Q_{\mathcal{A}} := \{ \bot, q_I\}\cup (\Sigma\times \{1,2\})
   \cup  \{P_\sigma: \sigma\in\Sigma\}$.
   The states of the form $(\sigma,i)$ are used to guess that a node
   of the tree is labelled by $(\sigma,i)$ while the states $P_\sigma$
   are used to prohibit that the left successor of a  node is
   labelled by $(\sigma,1)$ ($P_\bot$ is used if no restriction
   applies). The transitions ensure that whenever we
   guess that $d0$ is labelled by $(\sigma,1)$ then $d1$ is reached in
   state $P_\sigma$ ensuring that $d10$ cannot be labelled by
   $(\sigma, 1)$. 
   For the definition of $\Delta_{\mathcal{A}}$ we use the following
   conventions. $q$ ranges over $Q$, $i,j$ range over $\{1,2\}$,
   $\sigma$ over $\Sigma$, 
   $\tau$ over $\Sigma\setminus\{\bot\}$ and $\tau_{\not\sigma}$ over
   $\Sigma\setminus\{\bot, \sigma\}$ whenever $\sigma$ is fixed. 
   Set $\Delta_\mathcal{A}:=\{
   (q_I, q, (\bot, 1), \bot)$,
   $\left( (\sigma, i), (\sigma, i), (\tau, 1), P_\tau\right)$,
   $\left( (\sigma, i), (\sigma, i), (\tau, 2), P_\bot\right)$,
   $\left( (\sigma, i), (\sigma, i), (\tau, j), \bot\right)$,
   $\left((\sigma,i),(\sigma,i), \bot, P_\bot\right)$,
   $\left((\sigma,i),(\sigma,i), \bot, \bot\right)$,
   $\left(P_\sigma, \varepsilon, (\tau_{\not\sigma}, 1), P_\tau\right)$, 
   $\left(P_\sigma, \varepsilon, (\tau, 2), P_\bot\right)$,
   $\left(P_\sigma, \varepsilon, (\tau_{\not\sigma}, 1), \bot\right)$,
   $\left(P_\sigma, \varepsilon, (\tau, 2), \bot\right)$,
   $\left(P_\sigma, \varepsilon, \bot,  P_\bot\right)$,
   $\left(P_\sigma, \varepsilon, \bot, \bot\right)\}$.
 \end{proof}

\begin{lem} \label{STACS:Bijective}
  $\Encode:Q\times \Stacks(\Sigma) \rightarrow \EncTrees$ is a 
  bijection. We denote its inverse by $\Decode$.   
\end{lem}

The proof of this lemma is tedious. It can be found in Appendix
\ref{Appendix:BijectivityofEnc}. 

\subsection{Tree-automaticity of 
  Collapsible Pushdown Graphs} 

\label{subsec:ReachLAutomatic}

Our main technical contribution in this paper is stated in the next
proposition. It  concerns the regularity of the regular reachability
predicates $\Reach_L$ with respect to our encoding of configurations. 
We postpone the proof of this proposition to 
Section \ref{sec:RegularReach}.

\begin{prop} \label{Prop:ReachLregular}
  There are polynomials $p_1$ and $p_2$ such that the following holds.
  Let $\mathcal{S} = (Q, \Sigma,\Gamma, q_0, \Delta)$ be some
  collapsible pushdown system of level $2$ 
  and let $L$ be some regular language over $\Gamma$ recognised by
  some nondeterministic finite automaton with state set $P$.
  $\Reach_L$ is tree-automatic via $\Encode$ and there is a 
  nondeterministic finite tree-automaton with 
  $p_1(\lvert \Sigma\rvert) \cdot 
  \exp(p_2(\lvert Q \rvert \cdot \lvert P \rvert))$ many states
  recognising $\Reach_L$ in this encoding.
\end{prop}

\begin{rem}
  In the proposition, $\Reach_L$ has to be understood with respect to
  all possible configurations of a level $2$ collapsible pushdown
  system as opposed to those occurring in the configuration graph of
  $\mathcal{S}$, i.e., those reachable via the transitions of
  $\mathcal{S}$ from the initial configuration of $\mathcal{S}$. 
\end{rem}

We obtain the automaticity of the $\varepsilon$-contractions of all
level $2$ collapsible pushdown graphs as a direct
corollary of the previous result.

\begin{cor}
  There are polynomials $p$ and $q$ such that the following holds.
  Given a $\CPS$ \mbox{$\mathcal{S} = (Q, \Sigma,\Gamma, q_0,
    \Delta)$},  
  the $\varepsilon$-contraction
  $\CPG(\mathcal{S})/\varepsilon$ is regular via $\Encode$.
  Moreover, there is a presentation such that
  each automaton in the presentation of 
  $\CPG(\mathcal{S})/\varepsilon$ has at most
  $p(\lvert \Sigma\rvert) \cdot \exp(q(\lvert Q \rvert))$ many states. 
\end{cor}
\begin{proof}
  The domain of $\CPG(\mathcal{S})/\varepsilon$ is 
  $\left\{c: \CPG(\mathcal{S}) \models
    \Reach_{(\Gamma^*(\Gamma\setminus
      \{\varepsilon\}))^*}\left((q_0,\bot_2), c\right)\right\}$.    
  Note that $(\{\varepsilon\}^*(\Gamma\setminus \{\varepsilon\}))^*$ is
  accepted by an automaton with $2$
  states. Furthermore, hard-coding  
  $\Encode(q_0, \bot_2)$ as first argument to the automaton from 
  Proposition \ref{Prop:ReachLregular} increases the number of
  states by a at most a factor $3$ (because $\domain(\Encode(q_0,
  \bot_2))=\{\varepsilon, 0\}$). 
  Thus, the corresponding automaton has
  $3\cdot p_1(\lvert \Sigma\rvert) \cdot \exp(p_2(2\cdot \lvert Q
  \rvert))$ many states where $p_1$ and $p_2$ are the polynomials from
  Proposition \ref{Prop:ReachLregular}.

  Similarly, $\trans{\gamma}$ in $\CPG(\mathcal{S})/\varepsilon$ is
  exactly the same as $\Reach_{\{\varepsilon\}^*\gamma}$. Again $2$
  states suffice to recognise $\{\varepsilon\}^*\gamma$.  
\end{proof}

\subsection{Lower Bound for FO Model-Checking}
\label{sec:LowerBounds}

Since \CPG are  tree-automatic, their
\FO{} model-checking problem is decidable. 
The algorithm obtained this way has nonelementary complexity. In this
section we prove that we cannot do better: there is a fixed
collapsible pushdown graph of level $2$
whose
$\FO{}$ theory has nonelementary complexity. 
We present a reduction of \FO{} model-checking on the
full infinite binary tree to \FO{} model-checking on
this collapsible pushdown graph. 
Recall that \FO{} model-checking on the 
full  infinite binary tree \mbox{$\mathfrak{T}:=(T, \preceq, S_1, S_2)$}
with prefix order $\preceq$ and successor relations $S_1, S_2$
has a nonelementary lower bound 
(cf.~Example 8.3 in \cite{DBLP:journals/apal/ComptonH90}). 

\begin{thm}\label{thm:FOCPGnonelementary}
  The expression complexity of  any \FO{} model-checking algorithm for
  level $2$ collapsible pushdown graphs is nonelementary.
\end{thm}
\begin{rem}
  Note that this is a statement about plain collapsible pushdown
  graphs and not about the 
  $\varepsilon$-contractions.
  In contrast to the theorem, the first-order model-checking problem on
  non-$\varepsilon$-contracted level 1 pushdown graphs is 
  complete for alternating exponential time \cite{Volger83}.
\end{rem}
\begin{proof}
  We modify the \CPS of 
  Example \ref{exa:CPG}. 
  We add the transition $(2,a,\mathrm{P}_2,\Pop{2},0)$. 
  Note that the ordinal
  $(\omega, \preceq)$ is first order
  definable in this graph: restrict the domain to all elements with
  state $0$. The order $\preceq$ is then defined via
  $\varphi_\preceq(x,y):=\exists z\
  z\trans{\mathrm{P}_2} y \land z\trans{\mathrm{Co}} x$. 
  
  Now, we obtain the binary tree $(T, \preceq)$ by use of the stack
  alphabet $\{\bot, a, b\}$. For each occurrence of $a$ in a transition
  $\delta$, we make a copy of $\delta$ where we replace $a$ by $b$. 
  Then, each configuration $c$ with state $0$ is determined by
  $\TOP{2}(c)$ and these are in bijection to the set $\{a,b\}^*$. 
  Furthermore, $\varphi_\preceq$ defines the prefix relation on this
  set. 
  Thus, $(T, \preceq, S_1, S_2)$ is \FO{}-interpretable in this
  graph whence its \FO{} theory has nonelementary complexity.  
\end{proof}

\section{Decomposition of Runs}
\label{Sec:DecompositionOfRuns}

In this section we develop the technical background for the proof that
the regular reachability predicates are tree-automatic via $\Encode$. 
We  investigate the structure of runs of \CPS. 
We prove that any run is composed from subruns which can be classified
as 
\emph{returns}, \emph{loops}, or \emph{$1$-loops}. 
Forgetting about technical details, one can say that returns are runs
from some stack $s$ to $\Pop{2}(s)$, loops 
are runs that start and end in the same stack and $1$-loops are runs
from some stack $s$ to a stack $s'$ such that $s$ and $s'$ share the
same topmost word and $s$ is a substack of $s'$. 
Every run decomposes as a sequence of the form
$\lambda_0 \circ \rho_1 \circ \lambda_1 \circ \rho_2 \circ \dots \circ
\rho_n \circ \lambda_n$ where the $\rho_i$ only perform one operation
each 
and the 
$\lambda_i$ are returns, loops or $1$-loops of maximal length in a
certain sense. Let us explain this idea precisely in the case of a run
$\rho$ from some stack $s_1$ to a stack $s_2$ such that
$s_1=\Pop{2}^k(s_2)$. In this special case, the $\lambda_i$ are all
loops and the sequence of operations induced by $\rho_1, \dots,
\rho_n$ is a sequence of minimal length transforming $s_1$ into
$s_2$. This sequence of minimal length is in fact unique up to
replacement of $\Pop{1}$ and $\Collapse$ operations of level $1$. 
As a direct consequence, the loops $\lambda_i$ occurring in the
decomposition cover the largest possible part of $\rho$ in terms of
loops, returns and $1$-loops. 
This is also the key to understanding our decomposition result for
general runs: we identify maximal subruns of an arbitrary run which
are returns, loops and $1$-loops and we prove that the parts not
contained in one of these subruns form a short sequence of
operations.

Hence, understanding the existence of returns, loops and $1$-loops
allows to clarify whether runs between certain configurations exist. 
It turns out that our decomposition is very suitable for the analysis
with finite automata because such automata can be used to
decide whether returns, loops and $1$-loops starting in a given stack exist.

We next start with a general decomposition of any run into four
parts. Afterwards we prove decomposition results for each of the
parts where returns, loops and $1$-loops are the central pieces of the
decomposition. Finally, we show how finite automata acting on the
topmost word of a stack can be used to compute the existence of
returns, loops and $1$-loops starting at this stack.

\subsection{Decomposition of General Runs}
We introduce a decomposition of an arbitrary run $\rho$ into
four parts. The idea is that every run from a stack $s_1$ to a stack
$s_2$ passes a minimal common substack $t$ of $s_1$ and $s_2$. Any run
from $s_1$ to $t$  
decomposes into a first part from $s_1$ to a stack of the form
$t_1:=\Pop{2}^k(s_1)$ such that $\lvert t_1 \rvert = \lvert t \rvert$
and a second part from $t_1$ to $t$. Similarly, for the unique stack
$t_2:=\Pop{2}^j(s_2)$ such that $\lvert t_2 \rvert = \lvert t \rvert$
the run from $t$ to $s_2$ decomposes into a run from $t$ to $t_2$ and a
run from $t_2$ to $s_2$. In the following sections we prove that every
part of this decomposition again decomposes into returns, loops, and
$1$-loops. 

\begin{lem} \label{lem:GeneralRunDecomposition}
  Let $c_1=(q_1,s_1)$ and
  $c_2=(q_2,s_2)$ be configurations and  $\rho\in\Runs(c_1,c_2)$. Let
  $t$ be the minimal substack of $s_1$ such that $\rho$ visits $t$.  
  Furthermore, let
  $m_1:=\Pop{2}^{\lvert s_1 \rvert - \lvert t \rvert}(s_1)$ and
  $m_2:=\Pop{2}^{\lvert s_2\rvert - \lvert t \rvert }(s_2)$. 
  $\rho$ decomposes as
  $\rho=\rho_1\circ\rho_2\circ\rho_3\circ\rho_4$ where
  \begin{iteMize}{$\bullet$}
  \item $\rho_1\in\Runs(s_1,m_1)$ does not
    visit any substack of $m_1$ before its final configuration,
  \item   $\rho_2\in\Runs(m_1,t)$  does not visit any substack of
    $t$ before its final configuration, 
  \item $\rho_3\in\Runs(t, m_2)$ does not visit a substack of
    $\Pop{1}(t)$, and
  \item $\rho_4\in\Runs(m_2,s_2)$  does not visit any
    substack of $m_2$ after its initial configuration. 
  \end{iteMize}
\end{lem}
\begin{proof}
  Let $i_2\in\domain(\rho)$ be minimal such that $\rho(i_2)=t$. 
  If $t=m_1$ then set $i_1:=i_2$. 
  Otherwise there is some minimal $i_1<i_2$ such that
  $\rho(i_1)=m_1$: note that a stack operation alters either the width
  of the stack or the content of the topmost word. Thus, before
  reaching $t$, $\rho$ must visit some stack $\hat m$ of width at most 
  $\lvert t \rvert$ and of the form $\hat m=\Pop{2}^k(s_1)$ for some
  $k\in\N$.  Since $\hat m$ cannot be a substack of $t$, 
  $\hat m=m_1$. 
  Thus, we set $\rho_1:=\rho{\restriction}_{[0,i_1]}$ and
  $\rho_2:=\rho{\restriction}_{[i_1,i_2]}$. 

  For the definition of $\rho_3$ and $\rho_4$, note that $\lvert m_2
  \rvert = \lvert t \rvert$. Let $i_2<i_3 \in\domain(\rho)$ be maximal
  such that $\lvert \rho(i_3) \rvert = \lvert t \rvert$. 
  Since the first $\lvert t \rvert$ words of the stack are not changed
  by $\rho$ after $i_3$, $\rho(i_3)=m_2$ and $i_3$ is the last
  occurrence of $m_2$. Setting
  $\rho_3:=\rho{\restriction}_{[i_2,i_3]}$ and 
  $\rho_4:=\rho{\restriction}_{[i_3,\length(\rho)]}$ we are done.
\end{proof}

This decomposition motivates the following definition.
\begin{defi} \label{ABCD-Definition}
  Given a collapsible pushdown system $\mathcal{S}$, we define the
  following four relations on the configurations of $\mathcal{S}$:
  \begin{align*}
    &    \relA:=\{(c_1,c_2)\in\Conf^2:
    c_2=\Pop{2}^k(c_1)
    \text{ and }\exists\rho\in\Runs(c_1,c_2) \forall i<\length(\rho)\quad
    \rho(i)\not\leq c_2\}\\
    & \relB:=\{(c_1,c_2)\in\Conf^2: c_2=\Pop{1}^k(c_1)\text{ and }
    \exists \rho\in\Runs(c_1,c_2)
    \forall i<\length(\rho)\quad \rho(i)\not\leq c_2\} \\
    &  \relC:=\{(c_1,c_2)\in\Conf^2: c_1=\Pop{1}^k(c_2)\text{ and }
    \exists \rho\in\Runs(c_1,c_2)
    \forall i \leq\length(\rho)\quad \rho(i)\not< c_1\} \\
    &  \relD:=\{(c_1,c_2)\in\Conf^2: c_1=\Pop{2}^k(c_2) \text{ and }
    \exists \rho\in\Runs(c_1,c_2) \forall  i>0 \quad
    \rho(i)\not\leq c_1\}.  
  \end{align*}
\end{defi}
\begin{rem}\label{Rem:DecompositionOfReach}
  Since we allow runs of length $0$, the relations $\relA$, $\relB$, $\relC$ and
  $\relD$ are reflexive.
  Lemma \ref{lem:GeneralRunDecomposition} states that
  \begin{align*}
  (c_1,c_2)\in \Reach \Leftrightarrow  \exists d, e,
  f\quad  (c_1,d)\in \relA \land (d,e)\in \relB \land 
  (e, f) \in \relC \land (f,c_2)\in \relD.    
  \end{align*}
\end{rem}

In Section \ref{sec:Reachability_Tree-Automatic} we show that the
relations $\relA, \relB, \relC$ and $\relD$ are automatic whence
$\Reach$ is also automatic. 
In the next section, we prove a decomposition result that especially
applies to all runs in $\relD$. Afterwards, in Section
\ref{subsec:Return} we provide a corresponding decomposition result
for all runs in $\relA$. Finally, we provide (much simpler)
decompositions for $\relB$ and $\relC$ in Section \ref{subsec:DecBC}.

\subsection{Milestones, Loops and Increasing Runs}
\label{subsec:Milestones}

In this section we aim at a decomposition result for all runs in
$\relD$. For this purpose we first introduce the notion of
\emph{generalised milestones} of some stack $s$. The underlying idea
is as follows. Some stack $m$ is a generalised milestone of $s$ if any
run from the initial configuration to $s$ of any collapsible pushdown
system also passes $m$. 
Moreover, if some run ending in $s$ passes some 
generalised milestone $m$ of $s$, then it passes all generalised
milestones of $s$ that are not generalised milestones of $m$. 
From the definition of generalised milestones it will be obvious that
every run $\rho$ in $\relD$ starts at a generalised milestone of its final
stack. Thus, a run in $\relD$ can be decomposed into parts that
connect one generalised milestone of its final stack with the next
generalised milestone. After introducing the precise notion of a loop
we will see that each of these parts consists of such a loop plus one
further transition. 
At a first glance, the formal definition of a generalised milestone
has nothing to do with our informal description. The connection
between the intended meaning and the formal definition is that 
in order to create some stack $s$ from the initial stack $\bot_2$, we
have to create it word-by-word and each word letter-by-letter in the
following sense. If we want to create $s=w_1:w_2:\dots:w_k$, we have
to use push operations in order to create the first word, i.e., the
stack $[w_1]$. Then we have to apply $\Clone{2}$ and obtain
$w_1:w_1$. In order to generate $s$ from this stack, we first have to
generate $w_1:w_2$ and then we can proceed generating the other words
of $s$. But for this purpose, we first have to remove every letter
from the second copy of $w_1$ until we reach the greatest common
prefix of $w_1$ and $w_2$. This can only be done by iteratively
applying $\Pop{1}$ or $\Collapse$ operations of level $1$. Having
reached $w_1:(w_1\sqcap w_2)$, we again start to create $w_1:w_2$ by
using push operations of level $1$. 

This way of creating $s$ from $\bot_2$ is the shortest method
to create $s$ which is unique up to replacements of $\Pop{1}$
operations by $\Collapse$ of level 1 and vice versa. 
At the same time any other method contains this
pattern as a (scattered) subsequence (again up to replacement of
$\Pop{1}$ by $\Collapse$ of level 1 and vice versa). If we deviate
from the described 
way of creating $s$, then we just insert some loops where we first
create some different stack and then return to the position where we
started to deviate. 
At the end of this section, Corollary \ref{Cor:OrderEmbedding}
will show that our intuition is
correct. Let us now formally define generalised milestones. 

\begin{defi}
  Let $s=w_1:w_2:\dots :w_k$ be a stack and let $w_0 = \bot$. We call a
  stack $m$ a 
  \emph{generalised milestone of} $s$ if $m$ is  of the form 
  \begin{align*}
    &m=w_1:w_2:\dots: w_i:v_{i+1} \text{ where }
    0\leq i<k,\\
    &w_i \sqcap w_{i+1}\leq v_{i+1}\text{ and }\\
    &v_{i+1}\leq w_i\text{ or }v_{i+1}\leq w_{i+1}.
  \end{align*}
  We denote by $\genMilestones(s)$ the set of all generalised
  milestones of $s$.  
  
  For a generalised milestone $m$ of $s$, we call $m$ a
  \emph{milestone of }$s$ if  $m$ is a substack of $s$, i.e., if
  $v_{i+1}\leq w_{i+1}$ in the above definition. 
  We write $\Milestones(s)$ for the set of all milestones of
  $s$.
\end{defi}

We next define a partial order that 
turns out to be linear when restricted to 
the set $\genMilestones(s)$.

\begin{defi}
  We define a partial order $\ll$ on all stacks as follows.
  $\ll$ is the smallest reflexive and transitive relation that
  satisfies the following conditions.
  Let \mbox{$s=w_1:w_2:\dots :w_k$} and $t=v_1:v_2:\dots :v_l$ be stacks. 
  $s \ll t$ holds if
  \begin{enumerate}[(1)]
  \item  $\lvert s \rvert < \lvert t \rvert$, or
  \item  $l=k$, $w_i=v_i$ for $i<k$ and
    $v_k < w_k \leq w_{k-1}=v_{k-1}$, or
  \item 
    $l=k$, $w_i=v_i$ for $i<k$ and
    $w_k < v_k \not\leq w_{k-1}=v_{k-1}$.
  \end{enumerate}
\end{defi}

For each stack $s$, we now characterise $\ll$ restricted to
$\genMilestones(s)$. The straightforward proofs of the following
lemmas are left to the
reader. 

\begin{lem}
  Let $s=w_1:w_2:\dots :w_k$ be a stack and
  $m_1,m_2\in\genMilestones(s)$.
  Then  $m_1 \ll m_2$ if one of the following holds:
  \begin{enumerate}[\em(1)]
  \item $\lvert m_1 \rvert < \lvert m_2 \rvert$,
  \item  $\lvert m_1 \rvert = \lvert m_2 \rvert=i$,
    $w_{i-1}\sqcap w_i < \TOP{2}(m_1) \leq w_{i-1}$ and
    $w_{i-1}\sqcap w_i \leq \TOP{2}(m_2) \leq w_i$, 
  \item  $\lvert m_1 \rvert = \lvert m_2 \rvert=i$ and
    $w_{i-1}\sqcap w_i < \TOP{2}(m_2) \leq \TOP{2}(m_1) \leq
    w_{i-1}$, or
  \item  $\lvert m_1 \rvert = \lvert m_2 \rvert=i$ and
    $w_{i-1}\sqcap w_i \leq  \TOP{2}(m_1) \leq \TOP{2}(m_2) \leq
    w_{i}$.
  \end{enumerate}  
\end{lem}

\begin{lem}
  For each stack $s$, $\ll$ induces a finite linear order on
  $\genMilestones(s)$. 
  Moreover, 
  if $m\in\genMilestones(s)$ then $(\genMilestones(m),\ll)$ is the
  initial segment of $(\genMilestones(s),\ll)$ up to $m$. 
\end{lem}

We call some $m\in \genMilestones(s)$ the $i$-th
generalised milestone if it is the $i$-th element of
$\genMilestones(s)$ with respect to $\ll$. 
Later we will see that $\ll$ corresponds to the
order in which the generalised milestones appear in any run from the
initial configuration to $s$. Note that the restriction of $\ll$ to
$\Milestones(s)$ coincides with the substack relation $\leq$.

Now we introduce loops formally and 
characterise
runs connecting generalised milestones in terms of loops. 
Later we
will see that there is a close correspondence between the milestones
of some stack $s$ and the nodes of our encoding $\Encode(s)$. 
This correspondence is one of the key observation in proving that
$\Encode(s)$ yields a tree-automatic encoding of the relation
$\relD$. 

\begin{defi} \label{DefLoop}
  Let $s$ be a stack and $q,q'$ states.
  A \emph{loop} from $(q,s)$ to $(q',s)$ is a run
  $\lambda\in\Runs((q,s),(q',s))$ that does not pass a
  substack of  $\Pop{2}(s)$ and that may pass $\Pop{1}^k(s)$ only if
  the $k$  topmost elements of $\TOP{2}(s)$ are letters with links of
  level $1$. This means that for all $i\in\domain(\lambda)$, if
  $\lambda(i)= (q_i,\Pop{1}^k(s))$ then $\Lvl(\Pop{1}^{k'}(s))=1$ for
  all $0\leq k' < k$. 

  If $\lambda$ is a loop from $(q,s)$ to $(q',s)$ such that
  $\lambda(1)=\Pop{1}(s)=\lambda(\length(\lambda)-1)$%
  , then we call 
  $\lambda$ a \emph{low loop}. 
  If $\lambda$ is a loop from $(q,s)$ to  $(q',s)$ that never passes
  $\Pop{1}(s)$, then we call $\lambda$ a \emph{high loop}. 
\end{defi}

For $s$ some stack, we sometimes write $\lambda$ is a loop \emph{of}
$s$. By this we express that $\lambda$ is a loop and its initial (and
final) stack is $s$. 
We now characterise runs connecting milestones of some stack in terms
of loops.  

\begin{lem} \label{LemmaDecompositioninLoops}
  Let $\rho$ be a run ending in stack
  $s=w_1:w_2:\dots :w_k$. Furthermore, let 
  $m\in\genMilestones(s)\setminus\{s\}$ be such that $\rho$ visits $m$. Let
  $i\in\domain(\rho)$ be maximal such that $\rho$ visits $m$ at $i$,
  i.e., the stack at $\rho(i)$ is $m$. 
  Then $\rho$ also visits the $\ll$-minimal generalised milestone
  $m'\in \genMilestones(s)\setminus\genMilestones(m)$ and for
  $i'\in\domain(\rho)$ maximal such that $\rho(i')=m'$,
  $\rho{\restriction}_{[i+1,i']}$ is a loop of $m'$.
\end{lem}
\begin{proof}
  We distinguish the following cases.
  \begin{iteMize}{$\bullet$}
  \item Assume that $m' = \Clone{2}(m)$. In this case 
    \mbox{$m = w_1 : w_2 : \dots : w_{\lvert m \rvert}$}. Thus, at the last
    position $j\in\domain(\rho)$ where $\lvert \rho(j) \rvert =\lvert
    m \rvert$, the stack at 
    $\rho(j)$ is $m$ (because $\rho$ never changes the first $\lvert
    m\rvert$ many words after passing $\rho(j)$). Hence, $i=j$ by
    definition. Since 
    $\lvert s \rvert > \lvert m\rvert$, it follows directly that the
    operation at 
    $i$ is a $\Clone{2}$ leading to $m'$. Note that $\rho$ never
    passes  a stack of 
    width $\lvert m \rvert$ again. Thus, it follows from Lemma 
    \ref{Lem:Howtwolinksevolve} that for $i'$ maximal with
    $\rho(i')=m'$ the run
    $\rho{\restriction}_{[i+1,i']}$  never 
    visits $\Pop{1}^k(m')$ if
    $\Lvl(\Pop{1}^{k-1}(m'))=2$. Thus, we conclude that 
    $\rho{\restriction}_{[i+1,i']}$ is a loop. 
  \item Assume that $m' = \Pop{1}(m)$. In this case, 
    $m = w_1: w_2: \dots: w_{\lvert m\rvert -1}: w$ for some $w$ such that
    $w_{\lvert m\rvert-1} \sqcap w_{\lvert m\rvert} < w \leq
    w_{\lvert m \rvert -1}$. 
    Thus, $w\not\leq w_{\lvert m\rvert}$ and creating $w_{\lvert m\rvert}$
    as the $\lvert m\rvert$-th word 
    on the stack requires passing $w_1: w_2: \dots :w_{\lvert m \rvert-1}:
    w_{\lvert m\rvert-1} \sqcap  w_{\lvert m\rvert}$. This is only possible
    via applying $\Pop{1}$ or $\Collapse$ of level $1$ to $m$. Since
    we assumed $i$ 
    to be maximal, the operation at $i$ must be $\Pop{1}$ or
    $\Collapse$ of level $1$ and leads to $m'$. 

    We still have to show that $\rho{\restriction}_{[i+1,i']}$
    is a loop. 
    By definition of $i$, $\rho{\restriction}_{[i+1,i']}$
    starts and ends in $m'$. 
    Due to the maximality of $i$, $\rho{\restriction}_{[i+1,i']}$
    does not visit the stack \mbox{$\Pop{2}(m) = \Pop{2}(m')$.} 
    Furthermore, 
    $\TOP{1}(m')$ is a cloned element.
    Due to Lemma 
    \ref{Lem:Howtwolinksevolve}, 
    if
    $\rho{\restriction}_{[i+1,i']}$ visits
    $\Pop{1}^k(m')$ then
    $\Lvl(\Pop{1}^{k-1}(m'))=1$.
    Thus, $\rho{\restriction}_{[i+1,i']}$ is a loop.  
  \item The last case is $m' = \Push{\sigma,l}(m)$ for
    $(\sigma,l)\in\Sigma\times\{1,2\}$. In this
    case, 
    \begin{align*}
      m = w_1: w_2: \dots: w_{\lvert m\rvert-1}: w      
    \end{align*}
    for some $w$
    such that 
    $w_{\lvert m\rvert-1} \sqcap w_{\lvert m\rvert} \leq w <
    w_{\lvert m\rvert}$. 
    Creating $w_{\lvert m\rvert}$ on the stack requires pushing the
    missing symbols onto the stack as they cannot be obtained via
    clone operation  from the previous word. Since $i$ is maximal,
    the operation at $i$ is some $\Push{\sigma,l}$ leading to
    $m'$.
    $\rho{\restriction}_{[i,i']}$ is a high loop due
    to the maximality of $i$ (this part of $\rho$ never visits
    $m=\Pop{1}(m')$ or any other proper substack of $m'$).\qedhere
  \end{iteMize}\vspace{3 pt}
\end{proof}

\noindent As a  corollary of the lemma, we obtain that $\ll$ coincides with
the order in which the generalised milestones appear for the last time
in a given run starting in the initial configuration.

\begin{cor} \label{Cor:OrderEmbedding}
  Let $s$ be some stack and $m_1\in\Milestones(s)$.
  Some run $\rho\in \Runs(m_1,s)$ that does not visit
  substacks of $m_1$ (after the initial configuration)
  decomposes as 
  \begin{align*}
    \rho= \rho_1 \circ
    \lambda_1 \circ \dots \circ \rho_n \circ \lambda_n    
  \end{align*}
  where the
  $\lambda_i$ are loops and the $\rho_i$ are runs of length $1$ that
  connect one generalised milestone of $s$ with its $\ll$-successor
  (in $\genMilestones(s)$). 

  In particular, for $t=\Pop{2}^k(s)$ and configurations
  $c=(q,t)$ and $c'=(q',s)$ a run $\rho\in\Runs(c,c')$ witnesses
  $(c,c')\in \relD$ if and only if it decomposes as given above.
\end{cor}
For the direction from right to left of the last claim note 
that, if $t=\Pop{2}^k(s)$ and \mbox{$\rho\in\Runs(t,s)$} decomposes as
above, then $\rho_1$ performs a $\Clone{2}$. This implies that the run
\mbox{$\lambda_1\circ \rho_2 \circ \dots\circ \rho_n\circ\lambda_n$}
 cannot visit $t$ again. 

This corollary shows that it is sufficient to understand loops of
generalised milestones of a given stack $s$ in order to understand runs 
in $\relD$ ending in $s$. 
In Section \ref{sec:CompLoops} we show that the loops of a given stack $s$
can be computed by a finite automaton on input $\TOP{2}(s)$.

\subsection{Returns, 1-Loops and Decreasing Runs}
\label{subsec:Return}
Having analysed the form of runs in $\relD$, 
we now analyse runs in the converse direction: how can
we decompose a run in $\relA$? 
We need the notion of  \emph{returns}
and of
\emph{level-$1$-loops} in order to answer this question.

\begin{defi}  \label{DefReturn}
  Let $t=s:w$ be some stack with topmost word $w$. A \emph{return from
    $t$ to $s$} is 
  a run $\rho\in\Runs(t,s)$ such that $\rho$ never visits a
  substack of $s$ before $\length(\rho)$ and
  such that one of the following holds:
  \begin{enumerate}[(1)]
  \item the last operation in $\rho$ is $\Pop{2}$, or
  \item the last operation in $\rho$ is a $\Collapse$ and 
    $w < \TOP{2}(\rho(\length(\rho)-1))$, i.e., $\rho$ pushes at first some new
    letters onto $t$ and then performs a collapse of one of these new
    letters, or 
  \item  there is some
    $i\in\domain(\rho)$ such that
    $\rho{\restriction}_{[i,\length(\rho)]}$ is a return from
    $\Pop{1}(t)$ to $s$. 
  \end{enumerate}
\end{defi}
\begin{rem}
  The technical restrictions in the second condition have the
  following intention.
  A return from
  $t$ to $\Pop{2}(t)$ is 
  a run $\rho$ from $t$ to $\Pop{2}(t)$ that does
  not use the level $2$ links stored in $\TOP{2}(t)$
  (cf.~\cite{KartzowPhd} for a detailed discussion). 
\end{rem}
It is useful to note that any return from some stack $s$ that visits
$\Pop{1}(s)$ in fact satisfies the last condition in the definition of
return. 
\begin{lem} \label{Lem:ReturnVisitsPop1IsCase3}
  Let $\rho$ be some return from a stack $s$ to $\Pop{2}(s)$. If
  $0< i < \length(\rho)$ is the first position such that $\rho$ visits
  $\Pop{1}(s)$ at $i$, then $\rho{\restriction}_{[i,\length(\rho)]}$
  is a return from $\Pop{1}(s)$ to $\Pop{2}(s)$, i.e., $i$ witnesses
  that $\rho$ satisfies the third condition in the definition of a
  return. 
\end{lem}

In the case that the topmost word of the initial stack only contains
cloned elements, then there is an easy condition to verify that a run
starting at this stack is a return. 
\begin{lem}\label{lem:ReturnsOfClonedWords}
  Let $s$ and $t$ be stacks. Assume that $\lvert t \rvert <
  \lvert s \rvert$ and that $\TOP{2}(s) \leq \TOP{2}(t)$. 
  For $\rho$ some run of length $l$ starting at stack $s$, $\rho$ is a return if
  \begin{enumerate}[\em(1)]
  \item for all $0\leq i < l, \lvert \rho(i) \rvert \geq
    \lvert s \rvert$ and
  \item \label{Ass2-lem:ReturnsOfClonedWords}
    $\lvert t \rvert \leq \lvert \rho(l) \rvert < \lvert s \rvert$.
  \end{enumerate}
\end{lem}
\begin{proof}
  The proof is by induction on the length of $\TOP{2}(s)$. 
  Note that the operation at $\rho(l-1)$ is either $\Pop{2}$ or a
  $\Collapse$ of level 2. If it is a $\Pop{2}$, then it is immediate
  that $\rho$ is a return. If it is a $\Collapse$ of level $2$, then
  there is a prefix $w\leq \TOP{2}(s)$ and some word $v$ such that
  $\TOP{2}(\rho(l-1))=wv$. Since all level $2$ links in $\TOP{2}(s)$
  point to stacks of width smaller than $t$, $v$ must be nonempty by
  assumption \eqref{Ass2-lem:ReturnsOfClonedWords}. Thus, if $w=\TOP{2}(s)$ we conclude immediately that
  $\rho$ is a return. if $w<\TOP{2}(s)$ then the only way to create a
  word $wv$ on a stack of width at least $\lvert s \rvert$ with a
  link to a stack of smaller width requires to visit
  $\Pop{2}(s):w$. But this implies that $\rho$ visits $\Pop{1}(s)$ and
  we conclude by application of the induction hypothesis to the final
  part of $\rho$ starting at the first occurrence of $\Pop{1}(s)$. 
\end{proof}

\begin{defi}
  Let $s$ be some stack and $w$ some word. 
  A run $\lambda$ of length $n$ is called a \emph{level-$1$-loop} (or
  $1$-loop) of $s:w$ 
  if the following conditions are satisfied.
  \begin{enumerate}[(1)]
  \item $\lambda\in\Runs(s:w,s:s':w)$ for some $2$-word $s'$,
    \label{OneLoopCondTwo}
  \item for every $i\in\domain(\lambda)$, $\lvert \lambda(i)\rvert >
    \lvert s \rvert$,  \label{OneLoopCondThree} and
  \item for every $i\in\domain(\lambda)$ such that  
    $\TOP{2}(\lambda(i))=\Pop{1}(w)$, there is some $j>i$ such that
    $\lambda{\restriction}_{[i,j]}$ is a return. 
  \end{enumerate}
\end{defi}

Before we  analyse the form of runs in $\relA$, we prove an auxiliary
lemma.

\begin{lem}
  Let $\rho$ be a run of length $l$ starting in some stack $s$ with
  topmost word $w$ such that 
  \begin{enumerate}[\em(1)]
  \item $\rho$ does not visit a substack of
    $\Pop{2}(s)$ before its final configuration,
  \item $\TOP{2}(\rho(l-1))$ is a proper prefix of $w$, and
  \item the last stack operation in $\rho$ is a collapse of level
    $2$. 
  \end{enumerate}
  If $\rho$ is not a return, then there is some $0 < i < l$ such that 
  \begin{enumerate}[\em(1)]
  \item   $\TOP{2}(\rho(i))=\Pop{1}(w)$,
  \item  for all $i\leq j \leq l$ the
    subrun $\rho{\restriction}_{[i,j]}$ is not a return, and
  \item\label{ass3-LemNoReturnProperties} for all $0<j < i$ such that
    $\TOP{2}(\rho(j))=\Pop{1}(w)$ there is a $j<k<i$ such that
    $\rho{\restriction}_{[j,k]}$ is a return. 
  \end{enumerate}
\end{lem}
\begin{proof}
  Assume that $\rho$ is such a run and that it is not a return. 
  We first prove that there is some $0<i<l$ such that
  \begin{iteMize}{$\bullet$}
  \item   $\TOP{2}(\rho(i))=\Pop{1}(w)$ and
  \item  for all $i\leq j \leq l$ the
    subrun $\rho{\restriction}_{[i,j]}$ is not a return.
  \end{iteMize}
  Afterwards we show that the minimal such $i$ also satisfies
  claim \eqref{ass3-LemNoReturnProperties}.
  \begin{enumerate}[(1)]
  \item Assume that the stack at $\rho(l-1)$ decomposes as
    $w_1:w_2:\dots: w_k$ such that $w\not\leq w_{\lvert s\rvert}$. 
    In this case let $i\leq l-1$ be minimal such that 
    $\lvert \rho(i) \rvert = \lvert s \rvert$ and
    $w\not\leq\TOP{2}(\rho(i))$. 
    Since $\rho$ does not visit a substack of $\Pop{2}(s)$ before $i$,
    we conclude immediately that the stack at $i-1$ is $s$ and the
    stack at $i$ is $\Pop{1}(s)$. We show that $i$ witnesses the
    claim. Note that $\TOP{2}(\rho(i))=\Pop{1}(w)$. 
    Heading for a contradiction, assume that there is some $i\leq j
    \leq l$ such that $\rho':=\rho{\restriction}_{[i,j]}$ is a return. 
    This assumption implies directly that $\lvert \rho(j) \rvert <
    \lvert \rho(i) \rvert = \lvert s \rvert$ whence $\rho'$ visits a
    substack of $\Pop{2}(s)$. Since $\rho$ does not do so before $l$,
    we conclude that $l=j$. 
    But this implies that $\rho$ is a run that starts in $s$, passes
    $\Pop{1}(s)$ and continues with a return from $\Pop{1}(s)$. 
    By definition, this implies that $\rho$ is a return
    contradicting  our assumptions.
    Thus, we conclude that there is no $j\geq i$ such that 
    $\rho{\restriction}_{[i,j]}$ is a return.
  \item Otherwise, 
    the stack at $\rho(l-1)$ decomposes as
    $w_1:w_2:\dots: w_k$ for some $k>\lvert s \rvert$ and there is
    some  $n>\lvert s \rvert$ 
    such that 
    $w\leq w_i$ for all $\lvert s \rvert \leq i < n$ and
    $w\not\leq w_{n}$. 
    Let $i_0\leq l-1$ be maximal such that 
    $\lvert \rho(i_0-1) \rvert < n$. 
    Then the stack at $\rho(i_0-1)$ is $w_1:w_2:\dots:w_{n-1}$ and the
    operation at $i_0-1$ is $\Clone{2}$. 
    Thus, $w\leq \TOP{2}(\rho(i_0))$. 
    Let $i_0<i_1<l-1$ be minimal such that 
    $\lvert \rho(i_1) \rvert = n$ and $\TOP{2}(\rho(i_1))<w$. 
    By minimality of $i_1$, we conclude that $\lvert \rho(i_1-1)\rvert
    = n$ whence $\TOP{2}(\rho(i_1-1))=w$ and
    $\TOP{2}(\rho(i_1))=\Pop{1}(w)$. 
    In order to prove that $i_1$ witnesses the claim of the lemma, we
    have to prove that for all $i_1\leq j \leq l$,
    $\rho{\restriction}_{[i_1,j]}$ is not a return. 
    But note that $\lvert \rho(i_1)\rvert = n \leq \lvert \rho(j)
    \rvert$ for all $i_1\leq j < l$ whence
    $\rho{\restriction}_{[i_1,j]}$ is not a return for $i_1\leq j <
    l$. 
    Moreover, since $\rho$ ends in a collapse on a prefix of
    $w=\TOP{2}(\rho(0))$, we conclude that $\lvert \rho(l) \rvert <
    \lvert \rho(0) \rvert = \lvert s \rvert < n$. 
    Thus, $\lvert \rho(i_1) \rvert - \lvert (\rho(l) \rvert \geq 2$
    and $\rho{\restriction}_{[i_1,l]}$ is not a return because it
    does not end in $\Pop{2}(\rho(i_1))$.     
  \end{enumerate}
  This completes the first part of our proof. We still have to deal with
  claim \eqref{ass3-LemNoReturnProperties}. For this purpose
  let  $0<i<l$ be minimal such that
  \begin{iteMize}{$\bullet$}
  \item   $\TOP{2}(\rho(i))=\Pop{1}(w)$ and
  \item  for all $i\leq j \leq l$ the
    subrun $\rho{\restriction}_{[i,j]}$ is not a return.
  \end{iteMize}
  Heading for a contradiction, assume that there is some  $0<j <i$ with
  $\TOP{2}(\rho(j))=\Pop{1}(w)$ such that there is no $j<k\leq i$ such
  that $\rho{\restriction}_{[j,k]}$ is a return.
  
  By minimality of $i$ there is some
  $k > i$ such that $\rho_j:=\rho{\restriction}_{[j,k]}$ is a return. 
  Since $\rho$ is no return, we directly conclude that 
  $\rho(j) \neq \Pop{1}(s)$ whence
  $\lvert \rho(j) \rvert > \lvert s \rvert$. 
  We distinguish two cases.
  \begin{enumerate}[(1)]
  \item If $\lvert \rho(j) \rvert > \lvert \rho(i) \rvert$ then the
    minimal $k_0\geq j$ such that $\lvert \rho(k_0) \rvert < \lvert
    \rho(j) \rvert$ satisfies $k_0\leq i$. But by definition, $k_0$
    is the only candidate for $\rho{\restriction}_{[j,k_0]}$ being  
    a return. Thus, $k=k_0$ which contradicts the assumption that $k>i$.
  \item If $\lvert \rho(j) \rvert \leq \lvert \rho(i) \rvert$, we
    conclude that $\lvert s \rvert \leq \lvert \rho(k) \rvert <\lvert
    \rho(i) \rvert$. Thus, there is also a minimal $k_0\geq i$ such
    that $\lvert s \rvert \leq \lvert \rho(k_0)\rvert < \lvert \rho(i)
    \rvert$. Since $\TOP{2}(\rho(i)) \leq \TOP{2}(s) =w$, we conclude 
    with Lemma \ref{lem:ReturnsOfClonedWords} that
    $\rho{\restriction}_{[i,k_0]}$ is a return contradicting our
    choice of $i$. 
  \end{enumerate}
\end{proof}

With this lemma we are prepared to prove our decomposition result.

\begin{lem} \label{FormLemma}
  Let $s,s'\in\Stacks(\Sigma)$ such that
  $s'= \Pop{2}^k(s)$ for some $k\in\N$ and
  let $\rho$ be some run. 
  $\rho$ is a run in $\Runs(s,s')$
  that does not visit a substack of  $s'$ before its final
  configuration if and only if 
  $\rho\in\Runs(s,s')$ and $\rho$ decomposes as 
  $\rho=\rho_1\circ \rho_2 \circ \dots \circ \rho_n$ where each
  $\rho_i$ is of one of  the following forms.
  \begin{enumerate}[\em F1.]
  \item \label{FormReturn}$\rho_i$ is a return, 
  \item \label{1LoopCol} $\rho_i$ is a
    $1$-loop followed 
    by a $\Collapse$ of collapse level $2$, 
  \item \label{1LoopPop} $\rho_i$ is a 
    $1$-loop followed by a 
    $\Pop{1}$ (or a 
    $\Collapse$ of collapse level $1$), 
    there is a $j>i$ such that 
    $\rho_j$ is of the form F\ref{1LoopCol}
    and for all $i<k<j$ $\rho_k$ is 
    of the form F\ref{1LoopPop}. 
  \end{enumerate}
\end{lem}

\begin{proof}
  First of all, note that the case $\length(\rho)=0$ is solved by
  setting $n:=0$.  
  
  We proceed by induction on the length of $\rho$. 
  Assume that $\rho\in\Runs(s,s')$ and that it does not visit $s'$
  before the final configuration. 
  We write
  $(q_i,s_i)$ for the configuration $\rho(i)$.  
  Firstly, consider the case that there is some $m\in\domain(\rho)$
  such that  $\rho_1:=\rho{\restriction}_{[0,m]}$ is a return.
  Then $\rho_1$ is of the form
  F\ref{FormReturn}. By induction hypothesis,
  $\rho{\restriction}_{[m,\length(\rho)]}$ decomposes as desired. 

  Otherwise, 
  assume that there is no $m\in\domain(\rho)$ such that
  $\rho{\restriction}_{[0,m]}$ is a return. 
  
  Nevertheless, there is a minimal $m\in\domain(\rho)$ such that 
  $\lvert s_m\rvert < \lvert s \rvert$.  
  The last operation of $\hat\rho:=\rho{\restriction}_{[0,m]}$ is a
  $\Collapse$ such that \mbox{$\TOP{2}(s_{m-1}) \leq \TOP{2}(s)$} (otherwise
  $\hat\rho$ would be a return). 
  
  Writing $w:=\TOP{2}(s_{m-1})$, we distinguish two cases.
  \begin{enumerate}[(1)]
  \item 
  First consider the case that $w=\TOP{2}(s)$. 
  Note that this  implies $\Lvl(s)=2$ because the  last operation of
  $\hat\rho$ is a collapse of level $2$.
  
  Furthermore, we claim that  $\hat\rho$ does not
  visit $\Pop{1}(s)$. Heading for a contradiction, assume that 
  $\hat\rho(i) = \Pop{1}(s)$ for some $i\in\domain(\hat\rho)$.
  Since $\hat\rho$ does not visit $\Pop{2}(s)$ between $i$ and $m-1$,
  $\TOP{2}(\hat\rho(m-1))=w$ is only possible if $\Lnk(w) = \lvert s
  \rvert -1$ (cf.~Lemma \ref{Lem:Howtwolinksevolve}).
  But then $\hat\rho{\restriction}_{[i,m]}$ is a return of
  $\Pop{1}(s)$ whence by definition  $\hat\rho$ is a return of
  $s$. This contradicts our  assumption. 

  We claim that $\hat\rho$  
  is a $1$-loop plus a
  $\Collapse$ operation: 
  we have already seen that $\hat\rho$ does not visit any proper
  substack of $s$. 
  Thus, it suffices to show that $\hat\rho$ reaches a stack with
  topmost word $\Pop{1}(w)$ only at positions where a return
  starts. 

  Let $i$ be a position such that 
  $\TOP{2}(\hat\rho(i))=\Pop{1}(w)$. 
  Recall that \mbox{$\TOP{2}(s_{m-1})=w$}, $\Lvl(w)=2$ and 
  $\Lnk(w)\leq\lvert s \rvert -1$. Since
  $\hat\rho$ does not visit $\Pop{1}(s)$, 
  $\lvert \hat\rho(i) \rvert > \lvert s \rvert$ and
  we cannot restore $\TOP{1}(w)$ by a push operation. Thus, there is
  some minimal position $m >  j>i$ such that 
  $\lvert \hat\rho(j) \rvert < \lvert \hat\rho(i) \rvert$. 
  Lemma \ref{lem:ReturnsOfClonedWords} implies that
  $\hat\rho{\restriction}_{[i,j]}$ is a return. 
  
  Thus, $\rho_1:=\hat\rho$ is of the form F\ref{1LoopCol}.
\item 
  For the other case, assume that $w<\TOP{2}(s)$. 
  Since $\hat\rho$ is not a return, we may apply the previous lemma.
  We conclude that there is some
  \mbox{$i\in\domain(\hat\rho)$} such that
  $\rho_1:=\hat\rho{\restriction}_{[0,i]}$ is a $1$-loop
  followed by a $\Pop{1}$ or a collapse of level $1$ and such that
  there is no $j>i$ such that $\hat\rho{\restriction}_{[i,j]}$
  is a return. 
  In order to show that $\hat\rho$
  is of the form  F\ref{1LoopPop} we
  have to  check the side conditions on the segments
  following in the decomposition of $\rho$. 
  For this purpose set $\rho':=\rho{\restriction}_{[i,\length(\rho)]}$. 
  By induction hypothesis $\rho'$ decomposes as 
  $\rho'=\rho_2\circ\rho_3\circ\dots\circ\rho_n$ where
  the $\rho_i$ satisfy the claim of the lemma. 

  Now, by definition of $i$, $\rho'$ does not start with a
  return. Thus, $\rho_2$ is of one of the forms
  F\ref{1LoopCol} or F\ref{1LoopPop}. But these 
  forms require that there is some $j\geq 2$ such that $\rho_j$ is of
  form F\ref{1LoopCol} and for all $2\leq k < j$,
  $\rho_k$ is of the form 
  F\ref{1LoopPop}. 
  From this condition it follows directly that $\rho=\rho_1\circ\rho'
  = \rho_1\circ\rho_2\circ\rho_3\circ\dots\circ\rho_n$
  and $\rho_1$ is of the form F\ref{1LoopPop}.\qedhere\vspace{3 pt}
  \end{enumerate}

  \noindent For the other direction, assume that $\rho\in\Runs(s,s')$ decomposes
  as $\rho=\rho_1\circ \rho_2 \circ \dots \circ \rho_n$ where each
  $\rho_i$ is of one of the forms F1-F3. Let $i_1<i_2<\dots<i_k=n$ be
  the subsequence of subruns of the forms F1 or F2. 
  Let $s_{i_j}$ be the stack of the final configuration of $\rho_{i_j}$. 
  A straightforward
  induction shows that $\lvert s_{i_1} \rvert > \lvert s_{i_2} \rvert
  >\dots > \lvert s_{i_k} \rvert$ and that all stacks that occur after
  the final configuration of $\rho_{i_j}$ and before the final
  configuration of $\rho_{i_{j+1}}$ have width at least $\lvert
  s_{i_j} \rvert$. Analogously, all stacks occurring before the final
  configuration of $\rho_{i_1}$ have width at least $\lvert s
  \rvert$. Thus, we conclude immediately that substacks of $s'$ cannot
  be visited before the final configuration of $\rho$.   
\end{proof}

\subsection{Decompositions for Runs in 
  \texorpdfstring{$\relB$ or  $\relC$}
  {Relations R Down and R Up}}
\label{subsec:DecBC}
The decomposition of runs witnessing that $(c_1,c_2)\in \relA$
or $(c_1,c_2)\in \relD$, respectively, turns out to be very useful for
proving tree-automaticity of $\relA$ and $\relD$. We use similar
characterisations for runs 
witnessing that certain pairs of configurations are contained in
$\relB$ or $\relC$. The proofs of the following characterisations are
straightforward inductions.

\begin{lem} \label{Lem:CharacterisationBC}
  Let $c_1,c_2\in\Conf$ and
  $\rho$ some run.
  \begin{enumerate}[\em(1)]
  \item   $\rho$ witnesses $(c_1,c_2)\in \relB$ if and only if 
    $\rho\in\Runs(c_1,c_2)$ and
    $\rho$ decomposes as
    \[\rho= \lambda_1 \circ \rho_1 \circ \lambda_2 \circ \rho_2 \circ
    \dots \circ \lambda_n \circ \rho_n\] where each $\lambda_i$ is a high
    loop
    and each $\rho_i$ is a run performing exactly one transition which
    is $\Pop{1}$ or a collapse of level $1$. 
  \item 
    Analogously, $\rho$ witnesses $(c_1,c_2)\in \relC$ 
    if and only if $\rho\in\Runs(c_1,c_2)$ and $\rho$ decomposes as
    $\rho= \lambda_0 \circ \rho_1 \circ \lambda_1 \circ \rho_2 \circ
    \dots \circ \lambda_{n-1} \circ \rho_n \circ \lambda_n$ where the
    $\lambda_i$ are high loops and each $\rho_i$ performs exactly one push
    operation.  
  \end{enumerate}
\end{lem}

\subsection{Computing Returns} 
\label{sec:CompReturns}

In this section, we prove that the existence of returns starting at a
given stack $s$ inductively depend on the returns starting at
$\Pop{1}(s)$. Later we use this result in order to show that there is
a similar dependence of loops starting in $s$ from loops and returns
starting in $\Pop{1}(s)$. Let us fix a \CPS $\mathcal{S}$. 
For some word $w$ occurring in a level $2$ stack, let 
$w{\downarrow}_0$ denote the word where each level $2$ link is
replaced by $0$, i.e., each $(\sigma,2,i)$
is replaced by $(\sigma,2,0)$.

\begin{defi}
  Set
  $\ExRet(w):=\{(q,q'):$ there is a return 
  from $(q, w{\downarrow}_0:w{\downarrow}_0)$ to 
  $(q', w{\downarrow}_0)\}$. 
  We also  
  set $\ExRet(s):=\ExRet(\TOP{2}(s))$.
\end{defi}

The main goal of this section is
the proof of the following 
proposition.

\begin{prop} \label{Prop:ReturnsAutomatic}
  There is a finite automaton with $2^{\lvert Q\times
    Q\rvert}$ many states that
    computes $\ExRet(w)$ on input $w{\downarrow}_0$. 
\end{prop}
\begin{rem}
  In fact, the automaton can be effectively constructed from a given
  \CPS (cf.~\cite{KartzowPhd}). The same holds analogously for
  Proposition \ref{Prop:LoopsAutomatic} and Corollary
  \ref{Cor:OneLoopsAutomatic}. 
\end{rem}

This proposition relies on the observation that returns of a stack $s$ with
topmost word $w$ are composed by runs that are prefixed by $s$ and by
runs that are returns of stacks with topmost word
$\Pop{1}(w)$. Furthermore, it relies on the observation that stacks
with equal topmost word share the same returns. The reader who is not
interested in the proof details may safely skip these and continue
reading Section \ref{sec:CompLoops}. 

\begin{lem} \label{Lem:RetunDependsTopword}
  Let $s$ be a stack with $\lvert s\rvert\geq 2$. There is a return from
  $(q,s)$ to $(q',\Pop{2}(s))$ if and only if
  $(q,q')\in\ExRet(\TOP{2}(s))$.
\end{lem}
\begin{proof}
  Let $w:=\TOP{2}(s){\downarrow}_0$. 
  A tedious but straightforward induction on the length of the return
  provides a transition-by-transition copy 
  of a return starting at $(q,s)$ 
  to a return starting at $(q,w:w)$
 and vice versa. 
\end{proof}

The following auxiliary lemmas prepare the decomposition of returns
into subparts that are returns starting at stacks with smaller topmost
words and subparts that are prefixed by the first stack of the return.

\begin{lem}\label{LemmaReturnsContainSubreturns}
  Let $\rho$ be a return starting at some stack $s$ with
  $\TOP{2}(s)=w$. If $\rho$ visits $\Pop{1}(s)$, let $k$ be the
  first occurrence of $\Pop{1}(s)$ in $\rho$, otherwise let
  $k:=\length(\rho)$. 
  If $0<i < k$ is a position such that
  $s\prefixeq \rho(i-1)$ and $\TOP{2}(\rho(i))=\Pop{1}(w)$ then there
  is some $i<j<k$ such that $\rho{\restriction}_{[i,j]}$ is a return.
\end{lem}
\begin{proof}
  Since $i<k$, the stack at $\rho(i)$ is not $\Pop{1}(s)$. 
  Thus, the stack at $\rho(i)$ is of the form $s':\Pop{1}(w)$ with
  $s\prefixeq s'$, in particular 
  $\lvert s \rvert < \lvert \rho(i) \rvert$. 
  There is a minimal  $i < j \leq k$ such that 
  $\lvert \rho(j) \rvert < \lvert \rho(i) \rvert$. 

  If $j < \length(\rho)$, we conclude by application of Lemma
  \ref{lem:ReturnsOfClonedWords}. 
  Otherwise, $j=\length(\rho)$ and $\rho$ does not visit $\Pop{1}(s)$.
  Since $\rho$ is a return, the operation at $j-1$ is $\Pop{2}$ (whence 
  $\rho{\restriction}_{[i,j]}$ is a return) or there is a nonempty word
  $v$ such that $\TOP{2}(\rho(j-1))=wv$ and the operation at $j-1$ is
  a collapse of level $2$. Since $v$ was created between $i$ and $j-1$
  its topmost link points to a stack of width at least $\lvert \rho(i)
  \rvert -1$ and we conclude again with Lemma
  \ref{lem:ReturnsOfClonedWords} that  $\rho{\restriction}_{[i,j]}$ is
  a return.
\end{proof}

\begin{lem} \label{Lemma:SubReturnDecomposition}
  Let $\rho$ be some run, $s$ some stack with topmost word
  $w:=\TOP{2}(s)$ such that the following holds.
  \begin{enumerate}[\em(1)]
  \item $s\prefixeq \rho(0)$,
  \item 
    $\rho(i)\not< s$ for all $0\leq i \leq \length(\rho)$, and
  \item for all $0 < i \leq \length(\rho)$ such that
    $s\prefixeq \rho(i-1)$ and $w\not\leq \TOP{2}(\rho(i))$, there is
    some $i<j \leq \length(\rho)$ such that $\rho{\restriction}_{[i,j]}$ is
    a return.
  \end{enumerate}
  There is
  a well-defined sequence 
  \begin{align*}
    0:=j_0\leq i_1<j_1\leq i_2 <j_2\leq \dots \leq i_n <j_n \leq
    i_{n+1}:=\length(\rho)
  \end{align*}
  with the following properties.
  \begin{enumerate}[\em(1)]
  \item
    For $1\leq k \leq n+1$,  
    $s\prefixeq \rho{\restriction}_{[j_{k-1}, i_k]}$.
  \item 
    For each $1\leq k \leq n$, there is a stack $s_k$ with
    $\TOP{2}(s_k)=\Pop{1}(w)$ such that
    $\rho{\restriction}_{[i_{k}+1,j_k]}$ is a 
    return from $s_k$ to $\Pop{2}(s_k)$. 
  \item For all $1\leq k \leq n$, $\TOP{2}(\rho(i_k))=w$ and the
    operation at $i_k$ in $\rho$ is a $\Pop{1}$ or a collapse of level
    $1$.
  \end{enumerate}
\end{lem}
\begin{proof}
  The proof is by induction on $\length(\rho)$.
  If $s\prefixeq \rho$ (in particular, if $\length(\rho)=0$), set
  $n:=0$ and we are done.  
  Otherwise, 
  let $j_0:=0$ and 
  let
  $i_1\in\domain(\rho)$ be the minimal position such that 
  $s\prefixeq \rho(i_1)$ but $s\notprefixeq\rho(i_1+1)$. 
  Since $\rho(i_1+1)\not< s$, the stack at $\rho(i_1+1)$ must be of
  the form  $s':\Pop{1}(w)$ for some $s'$ such that $s\prefixeq s'$. 
  This requires that the stack at $\rho(i_1)$ is $s':w$ and the
  operation at $i_1$ is $\Pop{1}$ or $\Collapse$ of level $1$. 
  By assumption on $\rho$, there is some $i_1+1< j_1 \leq \length(\rho)$
  such that $\rho{\restriction}_{[i_1+1,j_1]}$ is a return. 
  Thus, the stack at $\rho(j_1)$ is $s'$ whence $s\prefixeq \rho(j_1)$. 
  Thus, we can apply the induction hypothesis to 
  $\rho{\restriction}_{[j_1,\length(\rho)]}$ which settles the claim. 
\end{proof}

The previous lemma allows to classify returns as follows.

\begin{cor}
  \label{Cor:Return_Decomposition_For_Simulation}
  Let $\rho$ be a run starting in some stack $s$ with topmost word
  $w=\TOP{2}(s)$. 
  $\rho$ is a return from $s$ to $\Pop{2}(s)$ that does not pass
  $\Pop{1}(s)$ if and only if  $\rho\in\Runs(s,\Pop{2}(s))$ and there is
  a uniquely defined sequence 
  \begin{align*}
    0:=j_0\leq i_1<j_1\leq i_2<j_2\leq \dots\leq i_n<j_n \leq
    i_{n+1}:=\length(\rho)-1
  \end{align*}
  with the following properties.
  \begin{enumerate}[\em(1)]
  \item \label{ass1-Cor:Return_Decomposition_For_Simulation}
    For $1\leq k \leq n+1$,  
    $s\prefixeq \rho{\restriction}_{[j_{k-1}, i_k]}$.
  \item \label{ass2-Cor:Return_Decomposition_For_Simulation}
    For each $1\leq k \leq n$, there is a stack $s_k$ with
    $\TOP{2}(s_k)=\Pop{1}(w)$ such that
    $\rho{\restriction}_{[i_{k}+1,j_k]}$ is a 
    return from $s_k$ to $\Pop{2}(s_k)$. 
  \item \label{ass3-Cor:Return_Decomposition_For_Simulation}
    For all $1\leq k \leq n$, $\TOP{2}(\rho(i_k))=w$ and the
    operation at $i_k$ in $\rho$ is a $\Pop{1}$ or a collapse of level
    $1$.
  \item \label{ass4-Cor:Return_Decomposition_For_Simulation}
    Either $w$ is a proper prefix of $\TOP{2}(\rho(i_{n+1}))$ and
    the  operation at $i_{n+1}$ is a collapse of level $2$ or
    $w$ is a prefix of $\TOP{2}(\rho(i_{n+1}))$ and the operation at
    $i_{n+1}$ is a $\Pop{2}$. 
  \end{enumerate}
\end{cor}
\begin{proof}
  First assume that $\rho$ is such a return. 
  Due to Lemma \ref{LemmaReturnsContainSubreturns}, we can apply Lemma
  \ref{Lemma:SubReturnDecomposition} to
  $\rho{\restriction}_{[0,\length(\rho)-1]}$. This gives immediately
  the first three items. The last item is a direct consequence of the
  definition of a return. 

  Now assume that $\rho$ is a run from $s$ to $\Pop{2}(s)$ that
  satisfies conditions
  \eqref{ass1-Cor:Return_Decomposition_For_Simulation}-%
  \eqref{ass4-Cor:Return_Decomposition_For_Simulation}. Heading for a
  contradiction assume that 
  $\rho$ 
  visits $\Pop{1}(s)$. Due to 
  \eqref{ass1-Cor:Return_Decomposition_For_Simulation} this happens at
  some position 
  $i_k+1 \leq j \leq j_k-1$.  Due to 
  \eqref{ass2-Cor:Return_Decomposition_For_Simulation} we conclude
  that the width of 
  the stack at $\rho(j_k)$ is smaller than the width at $j$. But this
  contradicts condition 1. because $s\prefixeq \rho(j_k)$.

  For similar reasons $\rho$ does not visit a substack of $\Pop{2}(s)$
  before the final configuration. Thus, condition 
  \eqref{ass4-Cor:Return_Decomposition_For_Simulation} implies that
  $\rho$ is a return.   
\end{proof}

\begin{cor}
  \label{Cor:Low-Return_Decomposition_For_Simulation}
  Let $\rho$ be a run starting in some stack $s$ with topmost word
  $w=\TOP{2}(s)$. 
  $\rho$ is a return from $s$ to $\Pop{2}(s)$ that
  passes $\Pop{1}(s)$ if and only if $\rho\in\Runs(s,\Pop{2}(s))$ and
  there is 
  a uniquely defined sequence 
  \begin{align*}
    0:=j_0\leq i_1<j_1\leq i_2<j_2\leq \dots \leq i_n<j_n=\length(\rho)
  \end{align*}
  with the following properties.
  \begin{enumerate}[\em(1)]
  \item \label{ass1-Cor:Low-Return_Decomposition_For_Simulation}
    For $1\leq k \leq n$,  
    $s\prefixeq \rho{\restriction}_{[j_{k-1}, i_k]}$.
  \item \label{ass2-Cor:Low-Return_Decomposition_For_Simulation}
    For each $1\leq k \leq n$, there is a stack $s_k$ with
    $\TOP{2}(s_k)=\Pop{1}(w)$ such that
    $\rho{\restriction}_{[i_{k}+1,j_k]}$ is a 
    return from $s_k$ to $\Pop{2}(s_k)$. 
  \item \label{ass3-Cor:Low-Return_Decomposition_For_Simulation}
    For all $1\leq k \leq n$, $\TOP{2}(\rho(i_k))=w$ and the
    operation at $i_k$ in $\rho$ is a $\Pop{1}$ or a collapse of level
    $1$.
  \end{enumerate}
\end{cor}
\begin{proof}
  First assume that $\rho$ is such a return. 
  Let $0<k<\length(\rho)$ be the first occurrence of $\Pop{1}(s)$ in
  $\rho$. Application of Lemma \ref{Lemma:SubReturnDecomposition} to
  $\rho{\restriction}_{[0, k-1]}$ yields a decomposition into
  $s$-prefixed parts and returns (ending with an $s$-prefixed
  part). Finally, due to Lemma 
  \ref{Lem:ReturnVisitsPop1IsCase3} 
  $\rho{\restriction}_{[k,\length(\rho)]}$ is also a return. 

  Now assume that $\rho$ is a run from $s$ to $\Pop{2}(s)$ that
  satisfies conditions
  \eqref{ass1-Cor:Low-Return_Decomposition_For_Simulation}-\eqref{ass3-Cor:Low-Return_Decomposition_For_Simulation}. As
  in the previous corollary, we conclude 
  that $\rho$ does not visit substacks of $\Pop{2}(s)$ before the
  final configuration. 
  Due to condition
  \eqref{ass2-Cor:Low-Return_Decomposition_For_Simulation},
  $\rho(i_n+1)=\Pop{1}(s)$ and  
  $\rho{\restriction}_{[i_{n}+1,\length(\rho)]}$ is a return whence
  $\rho$ is also a return.  
\end{proof}

In the following corollary, we assume that $\ExRet(\varepsilon)=\emptyset$.
\begin{cor} \label{Cor:ReturnsDetermined}
  For each stack $s$, $\ExRet(s)$ is determined by
  $\ExRet(\Pop{1}(s)), \Sym(s)$ and $\Lvl(s)$. 
\end{cor}
\begin{proof}
  Let $w$ and $w'$ be words such that
  \mbox{$\Sym(w)=\Sym(w')$}, $\Lvl(w)=\Lvl(w')$ and
  $\ExRet(\Pop{1}(w))=\ExRet(\Pop{1}(w'))$. 
  Fix a return $\rho$ starting in
  $(q_1, s)$ for $s:=w{\downarrow}_0:w{\downarrow}_0$ and ending in 
  $(q_2, w{\downarrow}_0)$.  
  We have to prove that there is a return $\rho'$ from 
  $(q_1, s')$ 
  to $(q_2, w'{\downarrow}_0)$ for 
  $s':=w'{\downarrow}_0:w'{\downarrow}_0$. 
  
  The proof is by induction on $\lvert w \rvert$.
  Assume that  
  $\rho$ does not visit $\Pop{1}(s)$ and let 
  \begin{align*}
    0:=j_0\leq i_1<j_1\leq i_2<j_2\leq \dots \leq i_n<j_n \leq
    i_{n+1}:=\length(\rho)-1
  \end{align*}
  be the sequence according to Corollary
  \ref{Cor:Return_Decomposition_For_Simulation}. If $\rho$ visits
  $\Pop{1}(s)$, use the sequence according to Corollary
  \ref{Cor:Low-Return_Decomposition_For_Simulation} and proceed analogously. 
  For all $k\leq n+1$,
  $s\prefixeq\rho{\restriction}_{[j_{k-1},i_k]}$.
  Set \mbox{$\rho_k:=\rho{\restriction}_{[j_{k-1},i_k]}[s/s']$}
  (cf.~Lemma \ref{Lem:BlumensathLevel2}). 
  This settles the claim if $n=0$. 
  
  For the case $n>0$, note that $\rho{\restriction}_{[i_k+1,j_k]}$ is
  a return starting at some stack with topmost word $\Pop{1}(w)$. 
  Thus,   $\ExRet(\Pop{1}(w))=\ExRet(\Pop{1}(w')) \neq \emptyset$
  whence $w$ and $w'$ are words of length at least $2$. 
  Let $\delta_k$ be the transition connecting $\rho(i_k)$ with
  $\rho(i_k+1)$. Note that $\delta_k$ is either a $\Pop{1}$ transition or
  a $\Collapse$ transition and $\Lvl(\rho(i_k))=1$. 
  Note that $\TOP{2}(\rho(i_k))=w{\downarrow}_0$ whence
  $\TOP{2}(\rho_k(\length(\rho_k))[s/s']=w'{\downarrow}_0$. Hence,
  $\delta_k$ is applicable to the last configuration of $\rho_k$ and
  leads to a configuration $c_k$ with topmost word
  $\Pop{1}(w'{\downarrow}_0)$. 

  Recall that $\rho{\restriction}_{[i_k+1,j_k]}$ is a return starting at
  some stack with topmost word $\Pop{1}(w{\downarrow}_0)$. 
  Since $\ExRet(\Pop{1}(w))=\ExRet(\Pop{1}(w'))$, Lemma
  \ref{Lem:RetunDependsTopword} provides  a return 
  $\rho'_k$ from $c_k$ to $\rho_{k+1}(0)$. 

  Finally, let $\gamma$ be the transition connecting $\rho(i_{n+1})$
  with $\rho(i_{n+1}+1)=\rho(\length(\rho))$. 
  $\gamma$ connects the last configuration of
  $\rho_{n+1}$ with $(q_2,w'{\downarrow}_0)=(q_2, \Pop{2}(s'))$: 
  either $\gamma$ is a $\Pop{2}$ transition and $\lvert \rho(i_{n+1})\rvert
  = 2 = \lvert \rho_{n+1}(\length(\rho_{n+1}))\rvert$ or
  $\gamma$ is a $\Collapse$ transition and $\TOP{1}(\rho(i_{n+1})) =
  (\sigma, 2, 1) =
  \TOP{1}(\rho(i_{n+1}))[s/s'] =
  \TOP{1}(\rho_{n+1}(\length(\rho_{n+1})))$. 
  Thus, 
  \begin{align*}
    \rho':= \rho_1 \circ \delta_1 \circ \rho_1' \circ \rho_2 \circ
    \dots \circ \rho_{n} \circ \delta_{n} \circ \rho_{n}' \circ
    \rho_{n+1} \circ \gamma    
  \end{align*}
  is a return from $(q_1, s')$ to $(q_2, \Pop{2}(s'))$. 
\end{proof}

Proposition \ref{Prop:ReturnsAutomatic}, which states that 
$\ExRet(w)$ can be computed by a finite automaton, 
 is a direct corollary of the
previous lemma: $\ExRet(w)$ is a subset of $Q\times Q$. An automaton
in state $\ExRet(\Pop{1}(w))$ can change to state $\ExRet(w)$ on input
$\Sym(w)$ and $\Lvl(w)$.

\subsection{Computing (1-) Loops}
\label{sec:CompLoops}

In analogy to the results of the previous section, we now investigate
the existence of loops. We follow exactly the same ideas except for
the fact that a loop of some stack $s$ depends on the loops \emph{and}
returns of $\Pop{1}(s)$. At the end of this section, we provide a
similar result for $1$-loops.  

\begin{defi}
  Set 
  $\ExLoop(w):=\{(q,q'):$ there is a loop
  from $(q, w{\downarrow}_0)$ to 
  $(q', w{\downarrow}_0)\}$. 
  Similarly,  let 
  $\ExHLoop(w)$ and
  $\ExLLoop(w)$ be the analogous sets for high loops and low loops,
  respectively.  
  Set $\ExOneLoop(w):=\{(q,q'):$ there is a stack $s$ and a 1-loop
  from $(q, w{\downarrow}_0)$ to 
  $(q', s:w{\downarrow}_0)\}$. 
  We also set $\ExLoop(s):=\ExLoop(\TOP{2}(s))$ and 
  analogous for  $\ExHLoop$, $\ExLLoop$ and $\ExOneLoop$.
\end{defi}

Extending the result of the previous section, our main goal is the
following automaticity result for $\ExLoop$ and $\ExOneLoop$.

\begin{prop} \label{Prop:LoopsAutomatic}
  There is a finite automaton $\mathcal{A}$ with 
  $2^{\lvert Q\times Q\rvert} \cdot 2^{\lvert Q\times Q \rvert} \cdot
  \lvert \Sigma \rvert^2 \cdot 2 $ many states that
  computes $\ExRet(w)$, $\ExLoop(w)$, $\ExHLoop(w)$, $\ExLLoop(w)$,
  $\ExOneLoop(w)$, 
  $\Sym(w)$ and 
  $\Lvl(w)$ on input $w{\downarrow}_0$ (where $w$ is a word occurring
  as topmost word of some stack $s$, i.e., for $w=\TOP{2}(s)$). 
\end{prop}

The reader who is not interested in the proof details, can safely skip
the rest of this section and continue with Section \ref{sec:RegularReach}.
In the following, we mainly use the same arguments as in the return
case, but we have to consider loops of the stack $\Pop{1}(s)$ because
those occur as subruns of low loops of $s$. We omit proofs whenever
they are analogous to the return case. 

\begin{lem} \label{Lem:LoopsDependsTopword}
  Let $s$ be some stack. There is a loop from
  $(q,s)$ to $(q',s)$ if and only if
  $(q_1,q_2)\in\ExLoop(\TOP{2}(s))$. The analogous statement holds for
  $\ExHLoop$, $\ExLLoop$, and
  $\ExOneLoop$. 
\end{lem}

The next step towards the proof of our main proposition is a
characterisation of $\ExLoop(w)$ in terms of
$\ExLoop(\Pop{1}(w))$ and $\ExRet(\Pop{1}(w))$
analogously to the result of Corollary
\ref{Cor:Return_Decomposition_For_Simulation} for returns. 
We do this in the following three lemmas. First, we present a
unique decomposition of loops into high and low loops. Afterwards, we
characterise low loops and high loops. 

\begin{lem} \label{Lem:LoopDecomposition}
  Let $\lambda$ be a loop from $(q,s)$ to $(q',s)$. 
  $\lambda$ is either a high loop or it has a unique decomposition as
  $\lambda = \lambda_0 \circ \lambda_1 \circ \lambda_2$ where
  $\lambda_0$ and $\lambda_2$ are
  high loops and $\lambda_1$ is a low loop.
\end{lem}
\begin{proof}
  If $\lambda$ is not a
  high loop, let 
  $i\in\domain(\lambda)$ be the minimal position just
  before the first occurrence of
  $\Pop{1}(s)$  and $j\in\domain(\lambda)$ be the position
  directly after the last occurrence of $\Pop{1}(s)$.
  By definition, $\lambda{\restriction}_{[0,i]}$ and
  $\lambda{\restriction}_{[j,\length(\lambda)]}$ are high loops and
  $\lambda{\restriction}_{[i,j]}$ is a low loop. 
\end{proof}
\begin{cor} \label{Cor:ExLoopDependence}
  The set $\ExLoop(s)$ can be computed from the sets $\ExHLoop(s)$ and
  $\ExLLoop(s)$ via 
  $\ExLoop(s):=\ExHLoop(s) \cup \{(q,q'): \exists q_1,q_2
    (q,q_1)\in\ExHLoop(s), (q_1,q_2)\in \ExLLoop(s),\text{ and }(q_2,q')\in
    \ExHLoop(s)\}$.
\end{cor}

In the following, we first explain how low loops depend on the loops
of smaller stacks, afterwards we explain how high loops depend on
returns of smaller stacks. 

\begin{lem} \label{Lem:LowLoopDecomposition}
  Let $\lambda$ be a low loop starting and ending in stack $s$. Then
  $\lambda{\restriction}_{[1,\length(\lambda)-1]}$ is a loop starting
  and ending in $\Pop{1}(s)$. The operation at $0$ is a $\Pop{1}$ or a
  $\Collapse$ of level $1$. The operation at $\length(\lambda)-1$ is
  a $\Push{\sigma}$ where $\TOP{1}(s)=\sigma\in\Sigma$. 
\end{lem}
\begin{proof}
  Note that each low loop $\lambda$ satisfies
$\lambda(1) = \Pop{1}(s)= \lambda(\length(\lambda)-1)$. 
  Since $\Pop{2}(\Pop{1}(s))=\Pop{2}(s)$, 
  it follows directly that the run in between satisfies the definition
  of a loop. 
\end{proof}

\begin{cor} \label{Cor:ExLLoopDependence}
  $\ExLLoop(s)$ depends on $\Sym(s), \Lvl(s), \Sym(\Pop{1}(s))$ and
  $\ExLoop(\Pop{1}(s))$. 
\end{cor}
Note that $\Sym(s)$ and $\Lvl(s)$ determine whether the first
transition of a low loop can be applied to the stack and that
$\Sym(\Pop{1}(s))$ determines whether the last transition of a low loop
can be applied. 

In analogy to the return case, we provide a decomposition of high
loops which shows that 
$\ExHLoop(s)$ is determined by the returns of $\Pop{1}(s)$ and by the
topmost symbol and link level of $s$. 

\begin{lem}[cf.~Corollary \ref{Cor:Return_Decomposition_For_Simulation}]
  \label{Lemma_Loop_Decomposition_For_Simulation}
  Let $\lambda$ be some run starting in
  some stack $s$ with topmost word
  $w=\TOP{2}(s)$. 
  $\lambda$ is a high loop from $s$ to $s$ if and only if
  $\lambda\in\Runs(s,s)$ and 
  there is a sequence $0=:j_0\leq i_1<j_1\leq  i_2 < j_2 \leq \dots
  \leq i_n < j_n \leq i_{n+1}:=\length(\lambda)$ such that
  \begin{enumerate}[\em(1)]
  \item for $1\leq k\leq n+1$,   $s\prefixeq
    \lambda{\restriction}_{[j_{k-1},i_k]}$ and 
  \item for each $1\leq k \leq n$, there is a stack $s_k$ with
    $\TOP{2}(s_k) = \Pop{1}(w)$ such that
    $\lambda{\restriction}_{[i_k+1,j_k]}$ is a return of $s_k$. 
  \end{enumerate}
\end{lem}

\begin{cor}[cf.~Corollary \ref{Cor:ReturnsDetermined}]
  \label{Cor:ExHLoopDependence}
  $\ExRet(\Pop{1}(s))$, $\Sym(s)$, and $\Lvl(s)$ determine $\ExHLoop(s)$.
\end{cor}

We conclude the proof of Proposition \ref{Prop:LoopsAutomatic} by
showing that $\ExOneLoop(s)$ is determined by
$\ExRet(s)$, $\Sym(s)$ and $\Lvl(s)$. 

\begin{lem}
  Let $\rho$ be some run. 
  $\rho$ is a $1$-loop from some stack $s$ to some stack $s'$ with
  $\TOP{2}(s)=\TOP{2}(s')$ if and
  only if
  $\rho$ is a run from some stack $s$ to some stack $s'$ with
  $\TOP{2}(s)=\TOP{2}(s')$ and $\rho$ decomposes as
  \begin{align*}
  \rho=\lambda_0\circ\rho_1\circ\lambda_1\circ\rho_2\circ\dots
       \circ\rho_n\circ\lambda_n    
  \end{align*}
  where $s\prefixeq \lambda_i$ and each $\rho_i$ is a
  return of a stack $s_i$ with $\TOP{2}(s_i)=\TOP{2}(s)$.
\end{lem}
\begin{proof}
  First assume that $\rho$ is a $1$-loop.
  If $s\prefixeq \rho$, set $n=0$. 
  Otherwise, let $j$ be minimal such that
  $s\notprefixeq\rho(j+1)$. 
  Set $\lambda_0:=\rho{\restriction}_{[0,j]}$. Note that $\lambda_0$
  is $s$ prefixed. 
  Since $\rho$ does not visit substacks of
  $\Pop{2}(s)$, $\TOP{2}(\rho(j+1))=\Pop{1}(\TOP{2}(s))$. 
  This implies that $\TOP{2}(\rho(j))=\TOP{2}(s)$. 
  By definition of a $1$-loop, there is some $k>j+1$ such that
  $\rho{\restriction}_{[j+1,k]}$ is a return. It follows directly that 
  $\rho_1:=\rho{\restriction}_{[j,k]}$ is a return of the stack
  $s_1:=\rho(j)$ and $\TOP{2}(s_1)=\TOP{2}(s)$. 
  
  The first direction of the lemma follows by iterating this
  construction. 

  Now assume that $\rho$ is a run from $s$ to $s'$ with
  $\TOP{2}(s')=\TOP{2}(s)$ that decomposes as specified above. 
  We show that $\rho$ is a $1$-loop. $\rho$ cannot visit a stack 
  $t$ with $\lvert t \rvert < \lvert s \rvert$ because  then it
  especially visits such a stack at $\lambda_j(0)$ for some
  $1\leq j \leq n$ which contradicts $s\prefixeq \lambda_j(0)$. 
  Moreover, if it visits some stack $t$ with
  $\TOP{2}(t)=\TOP{2}(\Pop{1}(s))$ then $t$ occurs  within some
  return $\rho_j$ before the final configuration of $\rho_j$.
  Assume that this position is $k$, i.e., the stack at $\rho_j(k)$ is
  $t$. Since $\rho_j$ is a return, there is a minimal 
  $k'$ such that the stack at $\rho_j(k')$ is narrower than $t$. 
  Since $\lvert s \rvert \leq \lvert \rho_j(k') \rvert < \lvert t
  \rvert$, we conclude by Lemma \ref{lem:ReturnsOfClonedWords} that
  $\rho{\restriction}_{[k,k']}$ is a return. 
\end{proof}

\begin{cor} \label{Cor:OneLoopsAutomatic}
  $\ExOneLoop(s)$ is determined by $\ExRet(s)$, $\Sym(s)$ and
  $\Lvl(s)$. 
\end{cor}
\begin{proof}
  By definition, it suffices to consider
  stacks of width $1$. Thus, 
  let $w$ and $w'$ be words with $\Sym(w)=\Sym(w')$,
  $\Lvl(w)=\Lvl(w')$ and $\ExRet(w)=\ExRet(w')$. 
  Set $s=[w]$ and $s'=[w']$.
  If $s\prefixeq\rho$, then 
  $s'\prefixeq \rho[s/s']$. Moreover if $\rho'$ is a return
  from $(q,\hat s:w)$  to $(q', \hat s)$ for some $\hat
  s\in\Stacks(\Sigma)$ then there is a
  return from $(q, \hat s':w')$ to $(q', \hat s')$ for all 
  $\hat s'\in\Stacks(\Sigma)$. Using the decomposition from the previous
  lemma, we can apply stack replacement and the existence of similar
  returns in order to show that $\ExOneLoop(w) = \ExOneLoop(w')$. 
\end{proof}

Proposition
\ref{Prop:LoopsAutomatic} now follows from
Corollaries  \ref{Cor:ExHLoopDependence},
\ref{Cor:ExLLoopDependence},
\ref{Cor:ExLoopDependence}, 
\ref{Cor:OneLoopsAutomatic} and from Proposition
\ref{Prop:ReturnsAutomatic}:
we can store $\ExRet(\Pop{1}(w))$,
$\ExLoop(\Pop{1}(w))$, $\Sym(\Pop{1}(w))$, $\Sym(w)$ and $\Lvl(w)$
in $2\cdot (2^{\lvert Q\times Q \rvert})^2 \cdot
\lvert \Sigma \rvert^2$ many states and update the information during a
transition reading the next letter of some word. Of course,
$\Sym(\Pop{1}(w))$ and $\Sym(w)$ are only defined for words of length
at least $2$. Note that we are only interested in words occurring as
topmost words of stacks. In such words, the combination $\Sym(w)=\bot$
and $\Sym(\Pop{1}(w))\in\Sigma$ does never occur because $\bot$ is the
bottom of stack symbol. Thus, some states from 
$2^{Q\times Q} \times 2^{Q\times Q} \times
\Sigma \times\{\bot\} \times\{1,2\}$ can be used to deal with the
cases of words of length at most $1$ separately.

\section{Regularity of the Reachability Predicate 
  via \texorpdfstring{$\Encode$}{Enc}} 
\label{sec:RegularReach}

Using the decomposition and automaticity results from the previous
Section, we show that for each
$\CPS$ $\mathcal{S}$ the encoding  $\Encode$
translates the relation $\Reach$ (on all possible configurations, not
only those occurring in the graph) into an automatic
relation. Using the closure of \CPS under products with finite
automata, we then extend this result to all reachability relations
$\Reach_L$ where $L$ is some regular language over the transition labels.

\subsection{Connection  between Milestones 
  and 
  \texorpdfstring{$\mathbf{\bfEncode{}}$ }{}}
\label{sec:MilestonesAndEncode}
We want to show the regularity of the regular
reachability relations on collapsible pushdown graphs. As a
preparation, we develop two correspondences between the nodes of
the encoding of some stack $s$ and the (generalised) milestones
$\Milestones(s)$ ($\genMilestones(s)$, respectively). 
Taking a node $d\in \Encode(s)$ to the stack encoded by
$\Encode(s){\restriction}_{\{e:e\leq_{\mathrm{lex}} d\}}$ is an order
isomorphism between $(\Encode(s), \leq_{\mathrm{lex}})$ and
$(\Milestones(s), \ll)$. We denote the image of a node $d$ under
this isomorphism by $\LeftStack(d, s)$. 
After the discussion of this isomorphism we develop  another
correspondence between nodes of $\Encode(s)$ and the generalised
milestones of $s$. For each $d\in \Encode(s)$, we define the
\emph{induced general milestone} 
$\InducedGenMilestone(d,s)\in\genMilestones(s)$. This is the $\ll$-maximal
$m\in\genMilestones(s)$ such that $\LeftStack(d,s) \prefixeq
m$. Apparently, if $\LeftStack(d,s)\prefixeq s$, then the 
induced general milestone is $s$. This occurs if and only if $d$ is in
the rightmost path of $\Encode(s)$. In all other cases
$\InducedGenMilestone(d,s)$ is the $\ll$-maximal generalised milestone of
$s$ whose topmost word is a copy of the topmost word of
$\LeftStack(d,s)$.  In this case $\TOP{2}(\LeftStack(d,s))$ is not a
prefix of $\TOP{2}(s)$ and $\InducedGenMilestone(d,s)$
is not a milestone. 

$\LeftStack$ and $\InducedGenMilestone$ are useful
concepts for the analysis of runs from some milestone $m$ of $s$ to $s$
due to the fact that any generalised milestone of $s$ occurs in the
image of $\LeftStack$ or in the image of $\InducedGenMilestone$. 
Furthermore, the generalised milestones associated to some node $d$
are closely connected to those associated to its successors. 
Assume that $d, d0, d1\in \Encode(s)$.  Then
$\LeftStack(d0,s)$ is the $\ll$-successor of $\LeftStack(d,s)$,
$\LeftStack(d1,s)$ is the $\ll$-successor of
$\InducedGenMilestone(d0,s)$ and
$\InducedGenMilestone(d,s)=\InducedGenMilestone(d1,s)$. 
If  $m\ll \LeftStack(d,s)$, Corollary \ref{Cor:OrderEmbedding}
implies that a run from $m$ to $s$ which does not visit substacks of $m$
visits $\LeftStack(d,s), 
\LeftStack(d0,s), \InducedGenMilestone(d0,s)$, etc. It also implies
that $\ll$-successors are connected by one operation followed by
some loop.  

We will later show that the combination of these observations is the
key to the regularity of the reachability relations. A finite
automaton may guess at each node $d$ the last states in which
$\LeftStack(d,s)$ and $\InducedGenMilestone(d,s)$ are visited by some
run from some milestone $m$ to $s$. Since the direct $\ll$-successors of these stacks are encoded in the successor or predecessor
of $d$, the automaton can check that these guesses are locally
consistent, i.e., that there is a single transitions followed by a
loop connecting $(q_1,s_1)$ to $(q_2,s_2)$ where $s_1$ and $s_2$ are
the generalised milestones represented by $d$ and its successor (or
predecessor). If the guess of the automaton is locally consistent at
all nodes, it witnesses the existence of a run from $m$ to $s$. 

\begin{defi}\label{STACS:DefLeftTree}
  Let $T\in\EncTrees$ be a tree and $d\in
  T\setminus\{\varepsilon\}$. Then the \emph{left and
  downward closed tree of $d$} is
  $\LeftTree{d,T} \coloneqq T{\restriction}_D$ where 
  \mbox{$D\coloneqq\{d'\in T: d'\leq_{\mathrm{lex}} d\}$}. 
  We denote by 
  \mbox{$\LeftStack(d,T)\coloneqq \pi_2(\Decode(\LeftTree{d,T}))$} the 
  \emph{ left stack induced by $d$}. $\pi_2$ denotes the projection to
  the stack of $\Decode(\LeftTree{d,T})$.  
  If $T$ is clear from the context, we omit it.
\end{defi}
\begin{rem}
  \label{TOP2Determined}
  We exclude the case $d=\varepsilon$ from the definition because the
  root encodes the state of the configuration and not a part of the
  stack.  In order to simplify notation, we use the
  following conventions. Let $c=(q,s)$ be a configuration. For
  arbitrary $d\in\{0,1\}^*$, we set
  $\LeftStack(d,s):=\LeftStack(0d,c):=\LeftStack(0d,\Encode(c))$. 

  Recall that
  \mbox{$w\coloneqq\TOP{2}(\LeftStack(d,s)){\downarrow_0}$} is
  $\TOP{2}(\LeftStack(d,s))$ where all level $2$ links are set to $0$.
  Due to the definition of the encoding, for every
  $d\in \Encode(s)$,
  $w$ is
  determined by the path 
  from the root to $d$: 
  interpreting $\varepsilon$ as empty word, the word along this
  path contains the pairs of      
  stack symbols and collapse levels of the letters of 
  $\TOP{2}(\LeftStack(d,s))$. Since all level $2$ links in $w$ are
  $0$, $w$ is determined by this path. 
  Thus, Proposition \ref{Prop:LoopsAutomatic} implies that there is
  an automaton that calculates at each position $d\in\Encode(q,s)$ the
  existence of loops of $\LeftStack(d, \Encode(q,s))$ with given
  initial and final state. 

  $\LeftStack(d,\Encode(q,s))$ is a substack of $s$ for all
  $d\in \Encode(q,s)$.
  This observation follows from Remark \ref{milestonesInEncoding}
  combined with the fact that the left stack is induced by a  
  lexicographically downward closed subset. 
\end{rem}

\begin{lem}
  Let  $s\in \Stacks(\Sigma)$. For each $d\in
  \Encode(s)$ we have
  $\LeftStack(d,s)\in\Milestones(s)$.
  Furthermore, for each
  $s'\in\Milestones(s)$ there is some
  \mbox{$d\in\Encode(s)$} such that
  \mbox{$s'=\LeftStack(d,s)$.}
\end{lem}

\begin{proof}
  For the first claim, let
  $d\in\Encode(s)$.  We 
  know that \mbox{$s_d:=\LeftStack(d,s)$} is a substack of $s$. 
  Recall that the path from the root to $s_d$ encodes
  $\TOP{2}(s_d)$. Furthermore, by definition of $\Encode$, $d$
  corresponds to some maximal block $b$ occurring in $s$ in the
  following sense:
  there are $2$-words $s_1, s_2$ and a word $w$ such that
  $s=s_1: (w\mathrel{\backslash} b) : s_2$ and such that the subtree
  rooted at $d$ encodes $b$. Moreover, $d$ encodes the first letter of
  $b$, i.e., if $b$ is a $\tau$-block, then the path from the
  root to $d$ encodes $w\tau$. 

  Note that by maximality of $b$, the greatest common prefix of the
  last word of $s_1$ and the first word of $w\mathrel{\backslash} b$
  is a prefix of $w\tau$. 
  Since the elements that are lexicographically smaller than $d$
  encode the 
  blocks to the left of $b$, one sees that $s_d=s_1:w\tau$. 
  Setting $k:=\lvert s_d \rvert$, we conclude that $s_d$ is a substack
  of $s$ such that the greatest common prefix of the $(k-1)$-st and
  the $k$-th word of $s$ is a prefix of $\TOP{2}(s_d)$. 
  Recall that this matches exactly the definition of a milestone of
  $s$. Thus, $s_d$ is a milestone of $s$ and  we completed the proof of
  the first claim.

  Now we turn to the second claim. 
  The fact that every milestone $s'\in\Milestones(s)$ is indeed represented by
  some node of $\Encode(s)$ can be seen by induction on the block
  structure of $s'$. 
  Assume that $s'\in\Milestones(s)$ and that $s'$  decomposes as 
  $s'=b_0: b_1 : \dots: b_{m-1}: b_m'$ into maximal blocks. We claim
  that $s$ then 
  decomposes as $s=  b_0: b_1: \dots:  b_{m-1}:b_m:\dots:b_n$ into
  maximal  blocks. 
  In order to verify this claim, we have to prove
  that $b_{m-1}$ cannot be the initial segment of a larger block
  $b_{m-1}:b_m$ in $s$. Note that if $b_{m}'$ only contains one letter, then 
  by definition of a milestone the last word of $b_{m-1}$ and the
  first word occurring in $s$ after $b_{m-1}$, which is the first word of
  $b_m$,  can only have a common prefix of length at most $1$. Hence,
  their composition does not form a block. 
  Otherwise, the first word of $b_m'$ contains two letters which do
  not coincide with the first two letters of the words in $b_{m-1}$. 
  Since this word is by definition a prefix of the first word in
  $b_m$, we can conclude again that $b_{m-1}:b_m$ does not form a block.
  

  Note that all words in the blocks $b_i$ for $1\leq i \leq n$ and in
  the block $b'_m$ share the same first letter which is encoded at the
  position $\varepsilon$ in $\Encode(s)$ and in $\Encode(s')$. 
  By the definition of $\Encode(s)$ the blockline induced by 
  $b_i$ is encoded in the subtree rooted at $1^i0$ in $\Encode(s)$. 
  For $i<m$ the same holds in $\Encode(s')$. We set $d:=1^m$. 
  Note that $\Encode(s')$ and $\Encode(s)$ coincide on all
  elements that are lexicographically smaller than $d$ (because these
  elements encode the blocks $b_1:b_2:\dots b_{m-1}$.  

  Now, we distinguish the following cases.
  \begin{enumerate}
  \item Assume that $b_m'= [\tau]$ for
    $\tau\in\Sigma\cup(\Sigma\times\{2\}\times\N)$. Then the block $b_m'$
    consists of only one letter. In this case
    $d$ is the lexicographically largest element of $\Encode(s')$
    whence
    $s'=\LeftStack(d,\Encode(s')) = \LeftStack(d,\Encode(s))$. 
  \item Otherwise, there is a
    $\tau\in\Sigma\cup(\Sigma\times\{2\}\times\N)$ such that
      \begin{align*}
        &b_m=  \tau \mathrel{\backslash}(c_0: c_1:\dots:
        c_{m'-1}:c_{m'}:\dots: c_{n'}) \text{ and}\\
        &b_m'= \tau\mathrel{\backslash}(c_0: c_1: \dots:
        c_{m'-1}:c'_{m'})        
      \end{align*}
     for some $m' \leq n'$ such that $c_0:c_1:\dots: c_{n'}$ are the
     maximal blocks of the blockline 
     induced by $b_m$  
     and \mbox{$c_0:c_1:\dots c_{m'-1}: c'_{m'}$} are the maximal
     blocks of the 
     blockline induced by $b'_{m}$.
     Now, \mbox{$c_1:c_2:\dots: c_{m'-1}$} are encoded in the subtrees
     rooted at $d01^i0$ for $0\leq i \leq m'-1$ in $\Encode(s)$ as
     well as in $\Encode(s')$.  $c_{m'+1}: c_{m'+2}:\dots :c_{n'}$ is
     encoded in the subtree rooted at $d01^{m'+1}$ in $\Encode(s)$
     and these elements are all lexicographically larger than
     $d01^{m'}0$. Hence, we can set $d':=d01^{m'}$ and repeat this
     case distinction on $d', c'_{m'}$ and $c_{m'}$ instead of $d,
     b'_m$ and $b_m$. 
  \end{enumerate}
  Since $s'$ is finite, by repeated application of the case distinction,
  we will eventually end up in the first case where we find a
  $d\in\Encode(s)$ such that $s'=\LeftStack(d,\Encode(s))$. 
\end{proof}

The next lemma states the tight connection between milestones of a stack
(with substack relation) and elements in the encoding of this stack 
(with lexicographic order).

\begin{lem} \label{LemmaOrderIso}
  $\LeftStack(\cdot,s)$ is an isomorphism between
  $(\domain(\Encode(s)), \leq_{\mathrm{lex}})$ and
  \mbox{$(\Milestones(s), \ll)$}. 
\end{lem}  
\begin{proof}
  If the successor of $d$ in lexicographic order is $d0$, then the
  left stack of the latter
  extends the former by just one letter. Otherwise, the
  left and downward closed tree of the successor of $d$
  contains more elements ending in 
  $1$, whence it encodes a stack of larger width. Since each left and
  downward closed tree
  induces a milestone, it follows that $g$ is an order isomorphism. 
\end{proof}

Recall that by Corollary \ref{Cor:OrderEmbedding}, each run to a
configuration $(q,s)$ visits the milestones of $s$ in the order given
by the substack relation. With the previous lemma, this translates
into the fact that the left stacks induced by the elements of
$\Encode(q,s)$ are visited by the run in lexicographical order of the
elements of $\Encode(q,s)$.

Beside the tight correspondence of milestones and nodes of the
encoding, there is another correspondence between generalised
milestones and nodes. Using both correspondences, each generalised
milestone is represented by a node of the encoding. The following
definition describes the second correspondence. Recall that $\leq$
denotes the prefix relation on trees. 

\begin{defi}
  Let $T\in\EncTrees$ be the encoding of a configuration. 
  Let $d\in T\setminus\{\varepsilon\}$. 
  Set $D:=\LeftTree{d,T} \cup \{d'\in T: d\leq d'\}$. 
  Let 
  $\ExLeftStack(d, T):= \pi_2(\Decode( T{\restriction}_{D}))$ 
   where $\pi_2(c)$ is the projection to the stack of the
   configuration $c$. 

  By case distinction on the rightmost branch of $T$, we define the
  \emph{generalised milestone induced by $d$} as follows.
  \begin{enumerate}[(1)]
  \item If $d$ is in the rightmost branch of $T$, then
    $\InducedGenMilestone(d, T):=\ExLeftStack(d,T) = \pi_2(\Decode(T))$,
  \item otherwise, set $\InducedGenMilestone(d,
    T):=\ExLeftStack(d,T):\TOP{2}(\LeftStack(d,T))$. 
  \end{enumerate}
\end{defi}
\begin{rem} \label{rem:BasicPropExLeftStack}
  As the name indicates, $\InducedGenMilestone(d,T)$ is a always a
  generalised milestone of $\Decode(T)$. For some $d$ in the
  rightmost branch of $T$ this holds trivially. 
  For all $d\in T$ that are not in the rightmost branch, note the
  following:  
  \begin{iteMize}{$\bullet$}
  \item If $d1\in T$, then $\InducedGenMilestone(d,T) =
    \InducedGenMilestone(d1,T)$.
  \item If $d$ is a leaf, then 
    $\InducedGenMilestone(d,T) = \Clone{2}(\LeftStack(d,T))$. 
    Since $\LeftStack(d,T)$ is the maximal milestone of 
    $\Decode(T)$ of width $ \lvert \LeftStack(d,T)\rvert$,
    $\InducedGenMilestone(d,T)$ is a generalised 
    milestone of $\Decode(T)$ (since $d$ is not in the rightmost
    branch, $\lvert \Decode(T)\rvert> \lvert\LeftStack(d,T)\rvert$).
  \item If $d,e\in T$ such that $e=d0$ and
    $d1\notin T$, 
    then $\InducedGenMilestone(d,T)=\Pop{1}(\InducedGenMilestone(e,T))$.
  \item If $d,e\in T$ such that $e=d0$ and
    $d1\in T$, 
    then $\LeftStack(d1,T)=\Pop{1}(\InducedGenMilestone(e,T))$.
  \end{iteMize}
  By induction from the leaves to the inner nodes one concludes that 
  $\InducedGenMilestone(d,T)$ is a generalised milestone of
  $\Decode(T)$ for each $d\in T$.  One also shows that, for each
  generalised milestone $m$ of $\Decode(T)$  that is not a milestone, 
  there is some 
  $d$ such that $m=\InducedGenMilestone(d,T)$ as follows. 
  For $s=w_1:w_2:\dots:w_n$, let $w_{i_k}$  be the $k$-th word
  occurring  in $s$ such that $w_{i_k}$ is not a prefix of $w_{i_k+1}$.
  Then $w_1:w_2:\dots:w_{i_k-1}:w_{i_k} = \LeftStack(d,T)$ for the $k$-th
  leaf $d$ and
  $\InducedGenMilestone(d,T) =
  w_1:w_2:\dots:w_{i_k-1}:w_{i_k}:w_{i_k}$ is the $k$-th generalised
  milestone of the form $w_1:w_2:\dots:w_{j-1}:w_j:w_j$ for some
  $j\leq n$ which is not a milestone. Now one can show the following. 
  If $\InducedGenMilestone(d,T)$ is not a milestone (but a generalised
  one) and
  $\Pop{1}(\InducedGenMilestone(d,T))$ is not a milestone, then 
  $\Pop{1}(\InducedGenMilestone(d,T))$ is a generalised milestone,
  there is a node $e$ such that $d=e01^k$ for some $k\in\N$, and
  $\Pop{1}(\InducedGenMilestone(d,T)) = \InducedGenMilestone(e,T)$. 
  Apparently, every generalised milestone that is not a milestone is
  of the form $\Pop{1}^k(w_1:w_2:\dots:w_{i-1}:w_i:w_i)$ for some
  $i\leq n$. Thus, this proves the claim. 
\end{rem}

\subsection{Tree-Automaticity of Reachability}  
\label{sec:Reachability_Tree-Automatic}

In this section we show that 
the reachability relation $\Reach$ is automatic via 
$\Encode$. In the next section we extend this result to the regular
reachability predicates.
Recall that due to Remark \ref{Rem:DecompositionOfReach}, 
a proof of the regularity of the relations  $\relA,
\relB, \relC$ and $\relD$ implies the regularity of $\Reach$.

\subsubsection{Regularity of the Relation 
  \texorpdfstring{$\relA$}{$\relA$}}
\label{subsec:RegA}
Recall that a pair of configurations $(c_1,c_2)$ is in $\relA$ if and
only if they are connected by a run $\rho$ that decomposes as
explained in Lemma \ref{FormLemma}. 
In the Appendix \ref{Appendix:AA} we construct an
automaton  $\mathcal{A}_{\relA}$ recognising $\relA$ based on the
following idea. $\mathcal{A}_{\relA}$ guesses the decomposition
according to Lemma \ref{FormLemma} and identifies a node representing
the stack reached after each part of the
decomposition. $\mathcal{A}_{\relA}$ labels this node 
by the initial and final state of the segment of the
decomposition starting at the corresponding stack and checks whether
the labelling of all the representatives fit together. 
The fact that $\mathcal{A}_{\relA}$ can check the correctness of its
guess relies heavily on the computability of the returns and $1$-loops
of the stack $\LeftStack(d, \Encode(q,s))$ along the path from the root
of $\Encode(q,s)$ to $d$. We next explain how $\mathcal{A}_{\relA}$
processes two configurations $c_1=(q_1,s_1)$ and $c_2=(q_2,s_2)$. Let
us assume that $s_2=\Pop{2}^k(s_1)$ and let 
$\rho$ be a run from $c_1$ to $c_2$ witnessing $(c_1,c_2)\in \relA$. 
Let us first assume that
$\rho$  is a sequence of returns
$\rho=\rho_1\circ \dots \circ \rho_k$. 
Due to the special form of $s_2$, $\Encode(s_1)$ and $\Encode(s_2)$
agree on the domain of $\Encode(s_2)$. Moreover $\Encode(s_1)$ extends
$\Encode(s_2)$ by $k$ paths to nodes $d_1, \dots, d_k$ such that each
$d_i$ does not have a $0$-successor. Moreover, all $d_i$ are
lexicographically larger than all nodes in $\Encode(s_2)$. 
There is a close correspondence between the $\rho_i$ and the $d_i$: 
$\LeftStack(d_i, \Encode(s_1))=\Pop{2}^{k-i}(s_1)$ and $\rho_i$ starts
in $\LeftStack(d_{k+1-i}, \Encode(s_1))$ and ends in
$\LeftStack(d_{k-i}, \Encode(s_1))$ (where we set $d_0$ to be the
rightmost leaf of $\Encode(s_2)$).
$\mathcal{A}_{\relA}$ guesses the existence of the run $\rho$ as
follows: at first, it checks that $\Encode(s_1)$ and $\Encode(s_2)$ agree on
the domain of $\Encode(s_2)$. 
Secondly, along the common prefix it propagates the initial and final
state of $\rho$, i.e., the states $q_1$ an $q_2$ and it computes at
each node $d$ the possible returns at the corresponding milestone. 
Now assume that at some node $e$ there starts a left and a right
branch such that the left branch is a prefix of $d_1, \dots, d_i$ and
the right branch is a prefix of $d_{i+1}, \dots, d_k$. At this
position the automaton guesses that 
$\LeftStack(d_{i+1}, \Encode(c_1))$ and 
$\LeftStack(d_{i}, \Encode(c_1))$ are connected via a return and
guesses the initial and final state $q_i, q_e$. Now the automaton
propagates along the left branch starting at $e$ the state information
$q_e$ and $q_2$ (trying to find a run from 
$\left(q_e,\LeftStack(d_i,\Encode(c_1)\right)$ to $c_2$) and along the
right branch the information 
$q_1, (R,  q_i, q_e)$ (trying to find a run from $c_1$ to
$c'=(q_i,\LeftStack(d_{i+1}, \Encode(c_1)))$ such that there is a return
starting in $c'$ and ending in state $q_e$). 
Doing the same at each splitting points of the prefixes of the 
$d_1, \dots, d_k$, the automaton finally reaches each node $d_i$ with
a tuple $(R,q_i, q_i')$ of states such that it has to check whether there
is a return starting in $(q_i, \LeftStack(d_i, \Encode(c_1)))$ and
ending in state $q_i'$. But since the returns of 
$\LeftStack(d_i, \Encode(c_1))$ are computable with an automaton
reading the path to $d_i$, this can be easily checked with a
tree-automaton. 

The case that $\rho$ decomposes as a sequence of returns is the
easiest one because all parts of the decomposition then start and end
in stacks that are milestones of $c_1$. Now assume that in the
decomposition of $\rho$ according to Lemma \ref{FormLemma} $1$-loops
followed by $\Pop{1}$ or $\Collapse$ operations occur. 
For simplicity of the explanation assume that 
$\rho=\lambda_1\circ \lambda_2 \circ \lambda_3$ such that $\lambda_1$
and $\lambda_3$ 
decompose as sequences of returns and 
$\lambda_2=\rho_1\circ\dots\circ \rho_n$ decomposes such that $\rho_i$ is a 
sequence of $1$-loops followed by $\Pop{1}$ and $\Collapse$ where only
$\rho_n$ ends in a collapse of level $2$. 
Using the notation from the previous case, we find $d_i$ and $d_j$
$(i<j)$ among $d_1, \dots, d_k$ such that $\lambda_2$ starts in stack
$\LeftStack(d_j, \Encode(c_1))$ and ends in $\LeftStack(d_i,
\Encode(c_1))$. 
We would like to treat this case similar to the return case, but
note that the final stacks of $\rho_1, \dots, \rho_{n-1}$ do not
necessarily appear as milestones of $s_1$: in general, these stack can
be wider than $s_1$! The key observation that allows to represent
these stack by certain milestones is the following: let
$t_0:=\LeftStack(d_j, \Encode(c_1))$ be the stack in which $\rho_1$
starts. By definition of a $1$-loop followed by a $\Pop{1}$ operation,
the final stack of $\rho_1$ is some stack $t_1$ such that
$\TOP{2}(t_1) = \TOP{2}(\Pop{1}(t_0))$. In particular, each level $2$
collapse link in $\TOP{2}(t_1)$ points to the same substack as the
corresponding element of $\TOP{2}(t_0)$. If $e_1$ is the
unique ancestor of $d_j$ such that $d_j=e_101^x$ for some $x\in\N$,
one sees that $t_1$ and $\LeftStack(e_1, \Encode(c_1))$ agree on their
topmost words including the targets of their level $2$ collapse
links. Note that $\lambda_2$ (modulo widening the stack during
$1$-loops) basically performs $\Pop{1} / \Collapse$
of level $1$ on $\TOP{2}(t_0)$ and finally a $\Collapse$ of level $2$
on a prefix of $\TOP{2}(t_0)$. Thus, with respect to the stack
operations induced by the run
$\rho_2\circ\dots\circ\rho_n$, $m_1:=\LeftStack(e_1, \Encode(c_1))$
and $t_1$ agree and $e_1$ can serve as representative of
$t_1$. Iterating this argument, we find nodes $e_2, \dots, e_{n-1}$
such that $m_x:=\LeftStack(e_x, \Encode(c_1))$ agrees with the final
stack $t_x$ of $\rho_x$ on the topmost word including the targets of
all level $2$ collapse links for all $x<n$. 
Moreover,
$\Collapse(t_{n-1})=\Collapse(m_{n-1})=\LeftStack(d_i,\Encode(c_1))$.
The automaton $\mathcal{A}_{\relA}$ uses this observation as follows. 
Processing the encodings of $c_1$ and $c_2$ from the root to the
leaves it arrives at the greatest common prefix $f$ of $d_i$ and $d_j$
with a guess $(q_i, q_e)$ of the initial and final state of the subrun
of $\rho$ which decomposes as
$\lambda_1'\circ\lambda_2\circ\lambda_3'$ where $\lambda_2$ is
defined as before and $\lambda_1'$ is a suffix of $\lambda_1$ and
$\lambda_3'$ a prefix of $\lambda_3$. 
Now the automaton guesses the initial and final state ($q_1,q_2$) of
the run $\lambda_2$ and propagates along the left branch the guess
that $\lambda_3'$ is a run from $q_2$ to $q_e$ and along the right
branch a guess of the form $(C, q_1, q_2)$ meaning that the last
part of $\lambda_1'\circ\lambda_2$ is a $1$-loop followed by collapse
(hence the ``$C$'')
from state $q_1$ to state $q_2$.
It then nondeterministically guesses the path to $d_j$ and updates
a guess $(X, q_1, q_2)$ with $X\in\{C, P\}$ at node $e_x$ as follows. 
From $e_x$ ($x\in\N$) on it propagates a guess $(P, q_1, q_2')$
towards $d_j$ meaning that the part of $\lambda_2$ connecting the
corresponding nodes of $e_{x-1}$ and $e_x$ is a $1$-loop followed by a
$\Pop{1} / \Collapse$ of level $1$ (``P'' for pop) from state $q_1$
to state $q_2'$. It verifies the compatibility of the guess 
$(X,q_1,q_2)$ at $e_x$ and the guess $(P, q_1, q_2')$ by checking that 
for any stack with the same topmost word as 
$m_x=\LeftStack(e_x,\Encode(c_1))$ there is a $1$-loop followed by an
operation induced by $X$ from state $q_2'$ to state $q_1$ (induced
operation means $\Collapse$ of level $2$ if $X=C$ and $\Pop{1}$ or
$\Collapse$ of level $1$ if $X=P$). This way the automaton reaches
$d_j$ with a guess $(X, q_1, q_2)$ and needs to verify that there 
is a 
$1$-loop followed by an operation induced by $X$ starting in 
$(q_1,t_0)$ and ending in $q_2$. Since an automaton can keep track of
the possible $1$-loops at each node of the tree this is possible. 
The automaton can guess states and successfully verify its assumptions
if  and only if 
there is a run from $\LeftStack(d_j, \Encode(c_1))$ to
$\LeftStack(d_i, \Encode(c_1))$ with initial and final state as
guessed at the node $f$ that decomposes into $1$-loops followed by one
$\Pop{1}$ or $\Collapse$ operation each. 
Combining this idea with the verification of guesses on parts of the
run $\rho$ that are returns, the tree-automaton accepts the encodings
of two configurations if and only if this pair of configurations is in
$\relA$. 

In Appendix \ref{Appendix:AA} we show that the described automaton
works correctly and can be implemented with exponentially many states
in the number of states of the collapsible pushdown system. Thus, we
obtain the following result.

\begin{lem}  \label{Lem:RegA}
  There are two polynomials $p_1, p_2$ such that the following holds.
  Let $\mathcal{S}$ be a level $2$ collapsible pushdown system with 
  stack alphabet $\Sigma$ and state space $Q$. 
  There is an
  automaton with $p_1(\lvert \Sigma\rvert) \cdot \exp(p_2(\lvert Q
  \rvert))$ many states  that accepts the convolution of two 
  trees if and only if this convolution is of the form 
  $\Encode(c_1) \otimes \Encode(c_2)$ for $c_1,c_2$
  configurations with  $(c_1,c_2)\in \relA$. 
\end{lem}

\subsubsection{Regularity of the Relation 
  \texorpdfstring{$\relB$}{$\relB$}}

Recall that the relation $\relB$ from Definition \ref{ABCD-Definition}
contains pairs $(c_1,c_2)$ if $c_2=\Pop{1}^m(c_1)$ and there is some
run from $c_1$ to $c_2$ not visiting any substack of $c_2$ before its
final configuration. 
A simple induction on the blocks in $c_1$ and $c_2$ 
yields the following characterisation of
$\Encode(c_1)\otimes\Encode(c_2)$.
For $d\in\{0,1\}^*$, $\lvert d \rvert_0$ denotes the number of $0$'s
in $d$. 

\begin{lem} \label{Lem:CharacterisationPop1Sequence}
  Let $s_1, s_2$ be stacks and
  $d$ the rightmost leaf of
  $\Encode(s_2)$.
  $s_2=\Pop{1}^m(s_1)$ for
  some $m\in\N$ if and only if $\Encode(s_1)\otimes\Encode(s_2)$ is of
  one of the following forms.
  \begin{enumerate}[\em(1)]
  \item  If $d\in\Encode(s_1)$, then
    $\domain(\Encode(s_1))=\domain(\Encode(s_2))\cup\{ d0^k:k\leq m\}$
    and $\Encode(s_1)$ and $\Encode(s_2)$ agree on 
    $\domain(\Encode(s_2))$.
  \item \label{ass2-Lem:CharacterisationPop1Sequence}
      If $d\notin\Encode(s_1)$, then $d\in\{0,1\}^*1$. Let
    $c\in\{0,1\}^*$ be 
    the predecessor of $d$. There is some $e\in\{0,1\}^*$ such that
    $ce$ is the 
    rightmost leaf of $\Encode(s_1)$. Then $\lvert e \rvert_0=m$, 
    \begin{align*}
      \domain(\Encode(s_1)) = \big(\domain(\Encode(s_2))\cup\{ x: c\leq x
        \leq e\}\big)\setminus\{d\}
    \end{align*}
    and $\Encode(s_1)$ and $\Encode(s_2)$ agree on
    $\domain(\Encode(s_2)) \setminus\{d\}$.
    Moreover, there exists some
    \mbox{$f\in\{0,1\}^*1$} such that 
    \begin{iteMize}{$\bullet$}
    \item $ce=f0^k$ for some
      $k\leq m$ and
    \item 
      \mbox{$\domain(\Encode(s_1))\setminus\domain(\Encode(s_2)) = 
        \{ f\leq x \leq ce\}$. }
    \end{iteMize}
  \end{enumerate}
\end{lem}
In Remark \ref{TOP2Determined} we pointed out that the path  from 
the root  to the rightmost leaf of $\Encode(s_1)$ encodes
$\TOP{2}(s_1)$. 
If the first case of the characterisation applies, then 
the $k$-th predecessor $x_k$ of the rightmost leaf of $\Encode(s_1)$
satisfies $\LeftStack(x_k, s_1) = \Pop{1}^{k}(s_1)$
for all $k<m$. 
If the second case applies,  
for all $c\leq x \leq e$ with $\lvert x \rvert_0 - \lvert c\rvert_0 =
k\leq m$,  
the path to $x$ encodes $w_k:=\TOP{2}(\Pop{1}^{m-k}(s_1))$ and
$\TOP{2}(\LeftStack(x, s_1))=w_k$. 
Thus, the elements on the path from $d$ (or $c$, respectively) to the
rightmost leaf of $\Encode(s_1)$ may serve as representatives of the
stacks that a run from $s_1$ to $s_2$ passes. 

Recall the decomposition into high loops of witnessing runs for
$(c_1,c_2)\in \relB$ from  Lemma \ref{Lem:CharacterisationBC}.
We  describe informally the automaton $\mathcal{A}_{\relB}$ 
that recognises the relation $\relB$. 
$\mathcal{A}_{\relB}$  guesses the path to the rightmost leaf of
$\Encode(c_2)$  and keeps track of $\ExHLoop(d)$ at each node $d$
on this path. 
Each node $d$ on the path from the rightmost leaf of $\Encode(c_2)$ to
the rightmost leaf of $\Encode(c_1)$ is labelled by
the state of $c_1$, by $\ExRet(\Pop{1}(\LeftStack(d,c_1)))$, by
$\Sym(\LeftStack(d,c_1))$, 
by $\Lvl(\LeftStack(d, c_1))$ and by a guess $q_d\in Q$ of a 
final state of
some run from $\Encode(c_1)$ to $\Pop{1}^{m-k}(s_1)$ for $k$
appropriate such that 
$\TOP{2}(\LeftStack(d,c_1))=\TOP{2}(\Pop{1}^{m-k}(c_1))$.  
Recall that such a state determines the set
$\ExHLoop(\LeftStack(d, c_1))$. 
The automaton can verify that the guesses of the $q_d$  are consistent in
the following sense. If it has labelled some node $d$ with a state
$q_d$ such 
that there is a run from $c_1$ to $(q_d, \Pop{1}^{m-k}(s_1))$, then
there is also a run from $c_1$ to $(q_c, \Pop{1}^{m-k-(1-i)}(s_1))$ for
$c$ the node such that $ci=d$. 
If $i=1$, $q_c=q_d$ and if $i=0$ then
the run to $s':=\Pop{1}^{m-k-(1-i)}(s_1)$ is extended by a high loop
of $s'$ followed by a $\Pop{1}$ or collapse of level $1$. 
Since the automaton ``knows'' the  possible high loops, the topmost
symbol and the link level of the stack, this check is trivial.  
Since it also stores  the state
of $c_1$, $\mathcal{A}_{\relB}$ can verify that its guess at the rightmost
leaf of $\Encode(c_1)$ is a state $q$ such that there is a high loop
starting in $c_1$ and ending in state $q$. 
We postpone the formal
definition of $\mathcal{A}_{\relB}$ and the proof of the following lemma to
Appendix \ref{Appendix:AB}.  
\begin{lem}  \label{Lem:RegB}
  There are polynomials $p_1, p_2$ such that the following holds.
  Let $\mathcal{S}$ be a level $2$ collapsible pushdown system with 
  stack alphabet $\Sigma$ and state space $Q$. 
  There is an
  automaton with 
  $p_1(\lvert \Sigma\rvert) \cdot \exp(p_2(\lvert Q \rvert))$ 
  many states that accepts the convolution of two 
  trees if and only if this convolution is of the form 
  $\Encode(c_1) \otimes \Encode(c_2)$ for $c_1,c_2$
  configurations with  $(c_1,c_2)\in \relB$. 
\end{lem}

\subsubsection{Regularity of the Relation 
  \texorpdfstring{$\relC$}{$\relC$}}

Recall that the relation $\relC$ is in some sense the backward version of 
the relation ${\relB}$. $(c_1,c_2)\in \relB$ holds if there is a 
sequence of high loops, 
$\Pop{1}$ and collapse of level $1$ connecting $c_1$ with $c_2$. 
Analogously, $(c_1,c_2)\in \relC$ holds if there is a sequence of high
loops and 
$\Push{\sigma,l}$ operations that generates $c_2$ from $c_1$. 
Since $(c_1,c_2)\in \relC$ implies that $c_1=\Pop{1}^k(c_2)$ for some
$k\in\N$, Lemma  \ref{Lem:CharacterisationPop1Sequence} applies
analogously. There is just one further condition: if $(c_1,c_2)\in \relC$
and $\TOP{2}(c_1) < \TOP{2}(c_2)\sqcap \TOP{2}(\Pop{2}(c_2))$, 
then there is some word $w$ such that
$\TOP{2}(c_2)\sqcap \TOP{2}(\Pop{2}(c_2)) = \TOP{2}(c_1) w$ and
$w$ only contains links of level $1$ (otherwise, the stack of
$c_2$ cannot be generated from the stack of $c_1$ without passing
$\Pop{2}(c_1)$). 

The automaton $\mathcal{A}_{\relC}$ recognising the relation $\relC$ via
$\Encode$ does the following.
Due to Lemma
\ref{Lem:CharacterisationPop1Sequence}, the path from the rightmost
leaf of $\Encode(c_1)$ to the rightmost leaf of $\Encode(c_2)$ has the
following form: 
For each $\lvert \TOP{2}(c_1)\rvert \leq m \leq \lvert \TOP{2}(c_2)
\rvert$ it contains nodes $d, d1, \dots, d1^{k_m}$ with $\lvert
d\rvert_0=m$ such that $\TOP{2}(\LeftStack(d, \Encode(c_2))$ is the
prefix of $\TOP{2}(c_2)$ of length $m$. 
Recall that $\mathcal{A}_{\relB}$ tries to label $d$ with a state $q_e$
and $e=d1^{k_m}0$ with a state $q_e'$ such that  
$q_e'$ and $q_e$ are connected by a high loop of $\LeftStack(e,
\Encode(c_1))$ plus a $\Pop{1}$ or collapse of level $1$. 
Since $\mathcal{A}_{\relC}$ works in the other direction, it labels $d$ with
a state $q_i$ and $d1^{k_m}0$ with a state $q_i'$ such that
$q_i$ and $q_i'$ are connected by a high loop of
$\LeftStack(d, \Encode(c_2))$ followed by a $\Push{\sigma,l}$. 
Furthermore, it checks that $l=1$ as long as the path to $d$ encodes a
proper prefix of 
$\TOP{2}(c_2)\sqcap \TOP{2}(\Pop{2}(c_2))$ which is not a prefix of
$\TOP{2}(c_1)$. 
With these remarks, the formal construction of $\mathcal{A}_{\relC}$
from $\mathcal{A}_{\relB}$ (cf.~Appendix \ref{Appendix:AB}) is left to the
reader. 

\begin{lem}  \label{Lem:RegC}
  There are two polynomials $p_1, p_2$ such that the following holds.
  Let $\mathcal{S}$ be a level $2$ collapsible pushdown system with 
  stack alphabet $\Sigma$ and state space $Q$. 
  There is an
  automaton with 
  $p_1(\lvert \Sigma\rvert) \cdot \exp(p_2(\lvert Q \rvert))$ 
  many states  that accepts the convolution of two 
  trees if and only if this convolution is of the form 
  $\Encode(c_1) \otimes \Encode(c_2)$ for $c_1,c_2$
  configurations with  $(c_1,c_2)\in \relC$. 
\end{lem}

\subsubsection{Regularity of the Relation 
  \texorpdfstring{$\relD$}{$\relD$}}
\label{subsec:RegD}

Given a \CPS $\mathcal{S}=(Q, \Sigma,\Gamma, \Delta_{\mathcal{S}},
q_0)$, we define an automaton $\mathcal{A}_{\relD}$ 
 that recognises the relation $\relD$
  in the following sense. 
Given configurations $c_1=(q_1, s_1)$ and $c_2=(q_2, s_2)$,
$\mathcal{A}_{\relD}$ accepts 
$\Encode(c_1)\otimes 
\Encode(c_2)$ if and
only if 
$s_1=\Pop{2}^k(s_2)$ for some $k\in\N$ and
there is a run $\rho$ of $\mathcal{S}$ from $c_1$ to
$c_2$ witnessing $(c_1,c_2)\in \relD$.

We informally explain how $\mathcal{A}_{\relD}$ processes the encoding 
\mbox{$\Encode(c_1)\otimes\Encode(c_2)$} of
two configurations $c_1,c_2$ in order to verify $(c_1,c_2)\in \relD$. 
First of all the automaton guarantees
that $s_1=\Pop{2}^k(s_2)$ for some $k\in\N$
(this is the case if and only if $\Encode(s_1)$ is a subtree of
$\Encode(s_2)$, the rightmost leaf $l$ of $\Encode(s_1)$ does not have
a $0$-successor in $\Encode(s_2)$ and
for each $d \leq_{\mathrm{lex}} l$, 
$d\in\Encode(s_2) \Leftrightarrow d\in \Encode(s_1)$).

Assume that $c_1=(q_1,s_1)$ and $c_2=(q_2,s_2)$ with
$s_1=\Pop{2}^k(s_2)$. If there is a run from $c_1$ to $c_2$ its form
is described in Corollary \ref{Cor:OrderEmbedding}: it is a sequence of
loops followed by one 
operation each that starts in some generalised milestone of $s_2$ and
leads to the next generalised milestone (with respect to $\ll$). 
Recall the following: each node in
$\Encode(s_2)$ which is not contained in $\Encode(s_1)$ corresponds to
a milestone in 
$\Milestones(s_2)\setminus\Milestones(s_1)$ via 
$\LeftStack(d, \Encode(s_2))$. Moreover, each node in the rightmost
branch of $\Encode(s_1)$ or in $\Encode(s_2)\setminus\Encode(s_1)$ 
corresponds to a generalised milestone in
$\genMilestones(s_2)\setminus\genMilestones(s_1)$ via
$\InducedGenMilestone$. 
The essence of Remark \ref{rem:BasicPropExLeftStack} is that for each
generalised milestones represented by a
node $d$ the node
representing the
$\ll$-successor of this generalised milestone can be found
locally around $d$. Since we can compute the possible loops of the
stack $\LeftStack(d, \Encode(c_2))$ and $\InducedGenMilestone(d,
\Encode(c_2))$ along the path from the root to $d$, a tree-automaton
may guess the initial and final states of each part of the
decomposition of a run according to Corollary \ref{Cor:OrderEmbedding} and
check the local compatibility of each of the guesses. 

The detailed definition of $\mathcal{A}_{\relD}$ as well as a proof of
the following lemma can be found in 
appendix \ref{Appendix:AD}. 
\begin{lem} \label{Lem:RegD}
  There are two polynomials $p_1, p_2$ such that the following holds.
  Let $\mathcal{S}$ be a level $2$ collapsible pushdown system with 
  stack alphabet $\Sigma$ and state space $Q$. 
  There is an
  automaton with 
  $p_1(\lvert \Sigma\rvert) \cdot \exp(p_2(\lvert Q \rvert))$ 
  many states  that accepts the convolution of two 
  trees if and only if this convolution is of the form 
  $\Encode(c_1) \otimes \Encode(c_2)$ for $c_1,c_2$
  configurations with  $(c_1,c_2)\in \relD$. 
\end{lem}

\subsection{Regularity of 
  \texorpdfstring{$\mathrm{Reach}_L$}
  {regular reachability}} 

In this part, we use the closure of collapsible pushdown systems
under products with finite automata in order to provide a proof of
the
automaticity of all regular reachability predicates (Proposition
\ref{Prop:ReachLregular}): we reduce regular
reachability to reachability in a product of the collapsible pushdown
system with the automaton for the regular language.

Recall that for $L\subseteq\Gamma^*$ some (string-) language, $\Reach_L$
is the binary relation that contains configurations $(c,\hat c)$ if and
only if there is a run $\rho$ from $c$ to $\hat c$ such that the labels of
the transitions used in $\rho$ form a word $w\in L$. 

Let $L$ be some regular language and $\mathcal{A}_L$ an automaton
recognising $L$. We construct the product $\mathcal{S} \times
\mathcal{A}_L$. 
$\Reach_L$ on $\CPG(\mathcal{S})$ is expressible via the relation
$\Reach$
 on $\CPG(\mathcal{S} \times
\mathcal{A}_L)$.  
As a corollary of this result, we obtain the tree-automaticity of
$\Reach_L$.

\begin{defi}
  Let $\mathcal{S}=(Q,\Sigma, \Gamma ,q_i,\Delta)$ be a
  $2$-\CPS and let
  $\mathcal{A}_L=(Q_L,\Gamma,i_0,F, \Delta_L)$ be a finite word-automaton. 
  We define the product of $\mathcal{S}$ and $\mathcal{A}_L$ to be the
  collapsible pushdown system 
  \begin{align*}
    &\mathcal{S}\times\mathcal{A}_L:=(
    Q\times Q_L, \Sigma, \Gamma,
    (q_i, i_0), \bar \Delta)    
    \text{ where}\\
  &\bar \Delta := \{((q,q_l),\sigma,\gamma, (q',q_l'), \op):
    (q,\sigma,\gamma,q',\op) \in \Delta 
    \text{ and } (q_l, \gamma, q_l')\in \Delta_L\}.      
  \end{align*}
\end{defi}
A straightforward induction shows that there is a run of
$\mathcal{S} \times \mathcal{A}_L$
from 
$((q,i_0), s)$ to  $((q',q_f), s')$ for $q,q'\in Q$, and $q_f\in F$ if
and only if there is a run of $\mathcal{S}$ from $(q,s)$ to $(q',s')$
such that the labels of the run form a word in $L$. 
Since $\Reach$ is a tree-automatic relation, we obtain the proof of
Proposition \ref{Prop:ReachLregular}.

\begin{proof}[Proof of Proposition \ref{Prop:ReachLregular}]
  \label{Proof:ReachLregular}
  Recall Remark \ref{Rem:DecompositionOfReach}. It says that 
  there is a positive existential first-order formula defining
  $\Reach$ in terms of $\relA, \relB, \relC$ and $\relD$.
  Due to Lemmas \ref{Lem:EncTreesAutomatic}, \ref{Lem:RegA}, \ref{Lem:RegB}, \ref{Lem:RegC} and
  \ref{Lem:RegD}, there are polynomials $p$ and $p'$ such that there
  is a (nondeterministic) tree-automaton
  $\mathcal{A}$ corresponding to $\Reach$ on 
  $\mathcal{S}\times\mathcal{A}_L$ with 
  $p(\lvert \Sigma \rvert) \cdot \exp(p'(\lvert Q \rvert \cdot \lvert
  P \rvert))$ many states.
  We obtain that for 
  states $q,q'$ and stacks $s,s'$ there is a final state $q_f\in F$ of
  $\mathcal{A}_L$ such that
  $\mathcal{A}$ accepts $(\Encode((q,i_0),s), \Encode((q',q_f),s'))$
  if and only if $((q,s),(q',s'))\in \Reach_L $ holds in $\mathcal{S}$.

  This almost completes the proof. We only have to modify
  $\mathcal{A}$ in such a way that it guesses $q_f$ and treats the
  configuration $(q,s)$ as if it was $((q,i_0),s)$. This can easily be
  done without increasing the number of states of the automaton
  because the states of the configurations are encoded in the roots of
  the trees. We explain in Appendix \ref{Appendix:ReachL} the detailed
  modification. 
\end{proof}

\section{Conclusion}
\label{sec:Conclusion}

We have shown that level $2$ collapsible pushdown graphs are uniformly
tree-automatic. Thus, their first-order theories are decidable with
nonelementary complexity. Moreover, even first-order extended by
regular reachability is decidable because of the automaticity of the
regular reachability relations.
Our result is sharp in several directions. First, we have also
shown a nonelementary lower bound for the complexity of the first-order
model-checking problem on collapsible pushdown graphs. Furthermore,
Broadbent \cite{Broadbent2012} showed that
the first-order theories of collapsible pushdown graphs are
undecidable from level $3$ on (which implies that they are not
tree-automatic).

\section*{Acknowledgement}
The content of this paper and its presentation developed during the
last 3 years and many people contributed to it with valuable
comments. I am afraid I cannot remember all people with whom I had
discussions on this topic but I surely have to give thanks to 
Dietrich Kuske, who initiated my interest in tree-automaticity. 
I also thank Martin Otto and Achim Blumensath. They surely had the
greatest influence on me during this time. 
I also want to thank the referees of the various versions that
appeared of this work, i.e., the referees of \cite{Kartzow10}, of
\cite{KartzowPhd} and of this paper itself, who had very valuable
comments on my work. 
Finally, I acknowledge funding from the DFG first via the project
'Model Constructions and Model-Theoretic Games
in Special Classes of Structures' 
and later via the project 'GELO'.

\bibliography{../../../Standard}

\begin{thebibliography}{10}

\bibitem{Alur06languagesof}
R.~Alur, S.~Chaudhuri, and P.~Madhusudan.
\newblock Languages of nested trees.
\newblock In {\em Proc. 18th International Conference on Computer-Aided
  Verification}, volume 4144 of {\em LNCS}, pages 329--342. Springer, 2006.

\bibitem{Blumensath1999}
A.~Blumensath.
\newblock Automatic structures. {Diploma} thesis, {RWTH A}achen, 1999.

\bibitem{Blumensath2008}
A.~Blumensath.
\newblock {On the structure of graphs in the Caucal hierarchy}.
\newblock {\em Theoretical Computer Science}, 400:19--45, 2008.

\bibitem{BroadbentCOS10}
C.~H. Broadbent, A.~Carayol, C.-H.~Luke Ong, and O.~Serre.
\newblock Recursion schemes and logical reflection.
\newblock In {\em LICS}, Proceedings of the 25th Annual IEEE Symposium on Logic
  in Computer Science, pages 120--129, 2010.

\bibitem{Broadbent2012}
Christopher~H. Broadbent.
\newblock The limits of decidability for first-order logic on cpda graphs.
\newblock In Christoph D{\"u}rr and Thomas Wilke, editors, {\em STACS},
  volume~14 of {\em LIPIcs}, pages 589--600. Schloss Dagstuhl - Leibniz-Zentrum
  fuer Informatik, 2012.

\bibitem{cawo03}
A.~Carayol and S.~W\"ohrle.
\newblock The {Caucal} hierarchy of infinite graphs in terms of logic and
  higher-order pushdown automata.
\newblock In {\em Proceedings of the 23rd Conference on Foundations of Software
  Technology and Theoretical Computer Science, FSTTCS 2003}, volume 2914 of
  {\em LNCS}, pages 112--123. Springer, 2003.

\bibitem{Caucal02}
D.~Caucal.
\newblock On infinite terms having a decidable monadic theory.
\newblock In {\em MFCS 02}, pages 165--176, 2002.

\bibitem{tata2007}
H.~Comon, M.~Dauchet, R.~Gilleron, C.~L\"oding, F.~Jacquemard, D.~Lugiez,
  S.~Tison, and M.~Tommasi.
\newblock Tree automata techniques and applications.
\newblock Available on: {http://www.grappa.univ-lille3.fr/tata}, 2007.
\newblock release October, 12th 2007.

\bibitem{DBLP:journals/apal/ComptonH90}
Kevin~J. Compton and C.~Ward Henson.
\newblock A uniform method for proving lower bounds on the computational
  complexity of logical theories.
\newblock {\em Ann. Pure Appl. Logic}, 48(1):1--79, 1990.

\bibitem{Hague2008}
M.~Hague, A.~S. Murawski, C-H.~L. Ong, and O.~Serre.
\newblock Collapsible pushdown automata and recursion schemes.
\newblock In {\em LICS '08: Proceedings of the 2008 23rd Annual IEEE Symposium
  on Logic in Computer Science}, pages 452--461, 2008.

\bibitem{KartzowPhd}
A.~Kartzow.
\newblock {\em First-Order Model Checking On Generalisations of Pushdown
  Graphs}.
\newblock PhD thesis, Technische Universit\"{a}t Darmstadt, Fachbereich
  Mathematik, 2011.

\bibitem{Kartzow10}
Alexander Kartzow.
\newblock Collapsible pushdown graphs of level 2 are tree-automatic.
\newblock In Jean-Yves Marion and Thomas Schwentick, editors, {\em STACS},
  volume~5 of {\em LIPIcs}, pages 501--512. Schloss Dagstuhl - Leibniz-Zentrum
  fuer Informatik, 2010.

\bibitem{KNU02}
T.~Knapik, D.~Niwinski, and P.~Urzyczyn.
\newblock Higher-order pushdown trees are easy.
\newblock In {\em Proceedings of FOSSACS'02}, volume 2303 of {\em LNCS}, pages
  205--222. Springer, 2002.

\bibitem{Maslov74}
A.~N. Maslov.
\newblock The hierarchy of indexed languages of an arbitrary level.
\newblock {\em Sov. Math., Dokl.}, 15:1170--1174, 1974.

\bibitem{Maslov76}
A.~N. Maslov.
\newblock Multilevel stack automata.
\newblock {\em Problems of Information Transmission}, 12:38--43, 1976.

\bibitem{Rubin2008}
S.~Rubin.
\newblock Automata presenting structures: A survey of the finite string case.
\newblock {\em Bulletin of Symbolic Logic}, 14(2):169--209, 2008.

\bibitem{Volger83}
Hugo Volger.
\newblock Turing machines with linear alternation, theories of bounded
  concatenation and the decision problem of first order theories.
\newblock {\em Theor. Comput. Sci.}, 23:333--337, 1983.

\end{thebibliography}
\bibliographystyle{plain}

\newpage
\appendix

\section{Proof of Bijectivity of 
  \texorpdfstring{$\Encode$}{Enc}}
\label{Appendix:BijectivityofEnc}
We start by explicitly constructing 
the inverse of $\Encode$. This inverse is called \mbox{$\Decode$}. Since
$\Encode$ removes the collapse links of 
the elements in a stack, we have to restore these now. In order to restore
the collapse links we use the following auxiliary 
functions for each $g\in \N$
\begin{align*}
  f_g: \{\varepsilon\}\cup(\Sigma\times\{1,2\}) \to
  \{\varepsilon\} \cup \Sigma\cup ( \Sigma\times\{2\} \times \N)  
\end{align*}
which map labels of trees to $1$-words of length up to
 $1$. We set 
\begin{align*}
  f_g(\tau)\coloneqq
  \begin{cases}
    \sigma &\text{if } \tau=(\sigma,1),\\
    (\sigma,2,g)  &\text{if } \tau=(\sigma,2),\\
    \varepsilon &\text{if } \tau=\varepsilon.
  \end{cases}
\end{align*}
In the next definition $g$ is the width of the stack decoded so far.
\begin{defi}
  Let \mbox{$\Gamma:=(\Sigma\times\{1,2\})\cup\{\varepsilon\}$}.
  Recall that encodings of stacks are trees in $\Trees{\Gamma}$. 
  We define the function 
 \mbox{ $\Decode: \Trees{\Gamma}\times \N \rightarrow
   (\Sigma\cup(\Sigma\times\{2\}\times\N))^{*2}$} 
 as
 follows. Let 
  \begin{align*}
      \Decode(T,g) = 
      \begin{cases}
        f_g(T(\varepsilon)) &
        \text{if }\domain(T)=\{\varepsilon\},\\
        f_g(T(\varepsilon)) \mathrel{\backslash}
        \Decode(\inducedTreeof{0}{T},g) &\text{if } 1\notin\domain(T),
        0\in\domain(T),\\
        f_g(T(\varepsilon)) \mathrel{\backslash}
        (\varepsilon:\Decode(\inducedTreeof{1}{T},g+1)) &
        \text{if } 0\notin\domain(T), 1\in\domain(T),\\
        f_g(T(\varepsilon)) \mathrel{\backslash}
        (\Decode(\inducedTreeof{0}{T},g):
        \Decode(\inducedTreeof{1}{T},g+G(\inducedTreeof{0}{T}))) &
        \text{otherwise,}
      \end{cases}
    \end{align*} where
    $G(\inducedTreeof{0}{T})\coloneqq
    \lvert\Decode(\inducedTreeof{0}{T},0)\rvert$ is 
    the width of the 
    stack encoded in $\inducedTreeof{0}{T}$. 
    For a tree $T\in\EncTrees$, the decoding of $T$ is
    \begin{align*}
      \Decode(T)\coloneqq(T(\varepsilon),
      \Decode(\inducedTreeof{0}{T},0))\in Q\times
      (\Sigma\cup(\Sigma\times\{2\}\times\N))^{+2}.       
    \end{align*}
\end{defi}
\begin{rem}
  Obviously, for each $T\in\EncTrees$, $\Decode(T)\in
  Q\times(\Sigma\cup (\Sigma\times\{2\}\times\N))^{+2}$. 
  In fact, the image of $\Decode$ is contained in $\Conf$,
  i.e., $\Decode(T)=(q,s)$ such that $s$ is a level $2$ stack. 
  The verification of this claim relies on two important
  observations. 

  Firstly, $T(0)=(\bot,1)$ due to condition 2 of Definition
  \ref{STACS:DefEncodingTrees}. Thus, all words in $s$ start with
  letter $\bot$. 
  $s$ is a stack if and only if the link structure of $s$ can be
  created using the push, clone and $\Pop{1}$ operations. 
  The proof of this claim can be done by a tedious but straightforward
  induction. We only sketch the most important observations for this
  fact. 

  Every letter $a$ of the form
  $(\sigma,2,g)$ occurring in $s$ is either a clone or can be created
  by the $\Push{\sigma,2}$ operation.   
  We call $a$ a clone if $a$
  occurs in $s$ in some word $waw'$ such that the word to the left of
  this word has $wa$ as prefix. Note that cloned elements are those
  that can 
  be created by use of the $\Clone{2}$ and $\Pop{1}$ operations from a
  certain substack of $s$. 
  
  If $a$ is not a clone in this sense, then $\Decode$ creates the
  letter $a$ because there is some $(\sigma,2)$-labelled node in $T$
  corresponding to $a$. Now, the important observation is that
  $\Decode$ defines
  $a=f_g((\sigma,2))$ where $g+1$ is the width of the stack decoded from
  the lexicographically smaller nodes. Hence, the letter $a$ occurs in
  the $(g+1)$-st word of $s$ and points to the $g$-th word. 
  Such a letter $a$ can clearly be created by a $\Push{\sigma,2}$
  operation. Thus, all $2$-words in the image of $\Decode$ can be
  generated by stack operations from the initial stack. A
  reformulation of this observation is that the image of $\Decode$
  only contains configurations.
\end{rem}

Now, we prove that 
$\Decode$ is injective on $\EncTrees$.  Afterwards, we show
that 
$\Decode\circ\Encode$ is 
the identity on the set of all configurations. This implies that
$\Decode$ is a surjective map from $\EncTrees$ to 
$\Conf$. Putting both facts together, we obtain that $\Decode$ is the
inverse of $\Encode$ whence, of course, $\Encode$ is bijective. 

\begin{lem} \label{DecInjective}
  $\Decode$ 
  is injective on $\EncTrees$.
\end{lem}
\begin{proof}
  Assume that there are trees $T', U'\in\EncTrees$ with
  $\Decode(T')=\Decode(U')=(q,s)$. Then by definition 
  $T'(\varepsilon) = U'(\varepsilon) = q$. Thus, we only have to compare
  the subtrees
  rooted at $0$, i.e., $T\coloneqq\inducedTreeof{0}{T'}$ and 
  $U\coloneqq\inducedTreeof{0}{U'}$. From our assumption it follows that
  \mbox{$\Decode(T,0)=\Decode(U,0)$.} 

  Note that the roots of $T$ and of $U$ are both labelled by 
  $(\bot, 1)$. 
  The lemma follows from the 
  following claim. 
  \begin{claim}
    Let $T$ and $U$ be trees such that there are $T', U'\in\EncTrees$
    and $d\in\domain(T')\setminus \{\varepsilon\}$, 
    $e\in\domain(U')\setminus \{\varepsilon\}$ such that
    $T=\inducedTreeof{d}{T'}$ and 
    $U=\inducedTreeof{e}{U'}$. 
    If $\Decode(T,  m) = \Decode(U, m)$ and either
    $T(\varepsilon) = U(\varepsilon) = \varepsilon$ or
    $T(\varepsilon)\in\Sigma\times\{1,2\}$ and
    $U(\varepsilon)\in\Sigma\times\{1,2\}$, then 
    $U=T$.\footnote{Since a node $d$ of a tree in $\EncTrees$ is
      labelled by $\varepsilon$ iff $d\in\{0,1\}^*1$, the pair of subtrees
      $\inducedTreeof{i}{T}$ and $\inducedTreeof{i}{U}$ inherit this
      condition for all $i\in\{0,1\}$.}
  \end{claim}

  The proof is by induction on the depth of the trees $U$ and
  $T$. 
  If \mbox{$\depth{U}=\depth{T}=0$}, $\Decode(U,m)$ and $\Decode(T,m)$
  are uniquely determined by
  the label of their roots. 
  A straightforward consequence of the definition of $\Decode$ is
  that $U(\varepsilon)=T(\varepsilon)$ whence $U=T$.

  Now, assume that the claim is true for all trees of depth at most
  $k$ for some fixed $k\in\N$. Let $U$ and $T$ be trees of
  depth at most $k+1$. 

  We proceed by a case distinction on whether the left or right
  subtree of $T$ and $U$ are defined. 
  In fact, we will later prove that
  $\Decode(T, m) = \Decode(U, m)$ implies that
  \begin{enumerate}[(1)]
  \item \label{ass1-DecInjective}
    $\inducedTreeof{0}{T}\neq \emptyset$ if and only if 
    $\inducedTreeof{0}{U}\neq \emptyset$ and
  \item \label{ass2-DecInjective}
    $\inducedTreeof{1}{T}\neq \emptyset$ if and only if
    $\inducedTreeof{1}{U}\neq \emptyset$. 
  \end{enumerate}
  We  first prove that $\Decode(T,  m)=\Decode(U,m)$ and
  conditions \eqref{ass1-DecInjective} and \eqref{ass2-DecInjective} imply that
  $U=T$. 
  Afterwards we show that all possible combinations that do not
  satisfy these conditions imply $\Decode(T,  m) \neq
  \Decode(U,m)$. 
  
  \begin{enumerate}[(1)]
  \item \label{Case1111}
    Assume that $\inducedTreeof{0}{U} = \inducedTreeof{1}{U} =
    \inducedTreeof{0}{T} = \inducedTreeof{1}{T} = \emptyset$. Then
    $\depth{T}= \depth{U} = 0$. For trees of depth $0$ we have already
    shown that  $\Decode(U, 0) = \Decode(T,0)$ implies $U=T$.
  \item \label{Case1010} 
    Assume that
    $\inducedTreeof{0}{U} = \emptyset$, 
    $\inducedTreeof{1}{U} \neq \emptyset$, 
    $\inducedTreeof{0}{T} = \emptyset$ and
    $\inducedTreeof{1}{T} \neq \emptyset$.
    In this case 
    \begin{align*}
      &\Decode(U, m) = f_m(U(\varepsilon)) \mathrel{\backslash}
      (\varepsilon: \Decode(\inducedTreeof{1}{U}, m+1))\text{
        and}\\
      &\Decode(T, m) = f_m(T(\varepsilon)) \mathrel{\backslash}
      (\varepsilon: \Decode(\inducedTreeof{1}{T}, m+1)). 
    \end{align*}
    Since $U(\varepsilon)=\varepsilon$ if and only if 
    $T(\varepsilon)=\varepsilon$, we can directly conclude that
    $U(\varepsilon) = T(\varepsilon)$. 
    But then $\Decode(T, m)  = \Decode(U, m)$ implies that 
    $\Decode(\inducedTreeof{1}{T}, m+1) =
    \Decode(\inducedTreeof{1}{U}, m+1)$.  
    Since $\depth{\inducedTreeof{1}{T}}\leq k$ and 
    $\depth{\inducedTreeof{1}{U}}\leq k$, the induction hypothesis
    implies that $\inducedTreeof{1}{T}=\inducedTreeof{1}{U}$. 
    We
    conclude that $T=U$.
  \item \label{Case0101}
    Assume that
    $\inducedTreeof{0}{U} \neq \emptyset$, 
    $\inducedTreeof{1}{U} = \emptyset$, 
    $\inducedTreeof{0}{T} \neq \emptyset$, and
    $\inducedTreeof{1}{T} =\emptyset$.
    In this case, 
    \begin{align*}
      &\Decode(U, m) = f_m(U(\varepsilon)) \mathrel{\backslash}
      \Decode(\inducedTreeof{0}{U}, m)\text{ and}\\
      &\Decode(T, m) = f_m(T(\varepsilon)) \mathrel{\backslash}
      \Decode(\inducedTreeof{0}{T}, m).      
    \end{align*}
    Since $U(\varepsilon)=\varepsilon$ if and only if
    $T(\varepsilon)=\varepsilon$, we conclude that
    $U(\varepsilon)=T(\varepsilon)$ and  
    $\Decode(\inducedTreeof{0}{U}, m) = 
    \Decode(\inducedTreeof{0}{T}, m)$. 
    Since the depths of $\inducedTreeof{0}{U}$ and of
    $\inducedTreeof{0}{T}$ are at most $k$, the induction hypothesis
    implies $\inducedTreeof{0}{U}=\inducedTreeof{0}{T}$ whence $U=T$.
  \item \label{Case0000} 
    Assume that
    $\inducedTreeof{0}{U}\neq \emptyset$, 
    $\inducedTreeof{1}{U}\neq \emptyset$, 
    $\inducedTreeof{0}{T}\neq \emptyset$,  and
    $\inducedTreeof{1}{T}\neq \emptyset$. 
    Then we have 
    \begin{align*}
      & 
      \Decode(U,m)= f_m(U(\varepsilon)) \mathrel{\backslash} 
      \left(\Decode(\inducedTreeof{0}{U},m):
        \Decode(\inducedTreeof{1}{U},m+m')\right) \text{ and }\\
      & 
      \Decode(T,m)= f_m(T(\varepsilon)) \mathrel{\backslash} 
      \left( \Decode(\inducedTreeof{0}{T},m):
        \Decode(\inducedTreeof{1}{T},m+m'')\right)
    \end{align*}
    for some natural numbers  $m',m''>0$. 
    
    Since $U(\varepsilon)=\varepsilon$ if and only if
    $T(\varepsilon)=\varepsilon$ this implies that the roots of $U$
    and $T$ coincide. 
    Hence,
    \begin{align*}
    \Decode(\inducedTreeof{0}{U},m):
    \Decode(\inducedTreeof{1}{U},m+m') 
    =
    \Decode(\inducedTreeof{0}{T},m):
    \Decode(\inducedTreeof{1}{T},m+m'')      
    \end{align*}
    If $\Decode(\inducedTreeof{0}{U}, m) = 
    \Decode(\inducedTreeof{0}{T}, m)$, then the induction
    hypothesis yields $\inducedTreeof{0}{U}=\inducedTreeof{0}{T}$. 
    Furthermore, this implies 
    $\Decode(\inducedTreeof{1}{U},m+m') =
    \Decode(\inducedTreeof{1}{T},m+m'')$ and $m'=m''$ whence
    by induction hypothesis 
    $\inducedTreeof{1}{U}=\inducedTreeof{1}{T}$.  
    In this case we conclude immediately that $T=U$. 
    
    The other case is that
    $\Decode(\inducedTreeof{0}{U},m)\neq
    \Decode(\inducedTreeof{0}{T},m)$. 
    We conclude immediately that
    the width of
    $\Decode(\inducedTreeof{0}{U},m)$ and the width of 
    $\Decode(\inducedTreeof{0}{T},m)$ do not coincide. 
    We prove that this case contradicts the assumption that
    $\Decode(U,m)=\Decode(T,m)$. 
    
    Let us assume that
    $\Decode(\inducedTreeof{0}{U},m) = 
    \Pop{2}^z\left(\Decode(\inducedTreeof{0}{T},m)\right)$ for
    some $z\in\N\setminus\{0\}$.  Note that this implies that the
    first word of $\Decode(\inducedTreeof{1}{U},m+m')$ is a word
    in $\Decode(\inducedTreeof{0}{T},m)$.
    
    Since $U(0)$ is a left successor in some tree belonging to $\EncTrees$, it
    is labelled by some \mbox{$(\sigma,l)\in \Sigma\times\{1,2\}$.} 
    We make a case distinction on $l$.
    \begin{enumerate}[(a)]
    \item Assume that $U(0) = (\sigma,2)$ for some
      $\sigma\in\Sigma$. Then 
      all words in $\Decode(\inducedTreeof{0}{T},m)$ start with
      the letter $(\sigma, 2, m)$. Thus, the first word of
      $\Decode(\inducedTreeof{1}{U},m+m')$ must also start with
      $(\sigma, 2, m)$. But all collapse links of level $2$ in
      $\Decode(\inducedTreeof{1}{U}, m+m')$ are at least
      $m+m'>m$. This is a contradiction.
    \item Otherwise,  $U(1) = (\sigma, 1)$ for some $\sigma\in\Sigma$. Thus, 
      all words in $\Decode(\inducedTreeof{0}{T},m)$ start with
      the letter $\sigma$. Thus, the first word of 
      $\Decode(\inducedTreeof{0}{U},m)$ and the first word of 
      $\Decode(\inducedTreeof{1}{U},m+m')$ have to start with
      $\sigma$. But this implies
      \mbox{$U(0)=U(10)=(\sigma, 1)$}. This contradicts the assumption that
      $U$ is a proper subtree of a tree from $\EncTrees$
      (cf.~condition \ref{Cond:LevelOneBlockscoincide} of Definition \ref{STACS:DefEncodingTrees}). 
    \end{enumerate}
    Both cases result in contradictions. Thus,  it is not the fact
    that there is  some $z\in\N\setminus\{0\}$ such that
    \begin{align*}
      \Decode(\inducedTreeof{0}{U},m) = 
      \Pop{2}^z\left(\Decode(\inducedTreeof{0}{T},m)\right)      
    \end{align*}
    By symmetry, we obtain that there is no $z\in\N\setminus\{0\}$
    such that
    \begin{align*}
      \Decode(\inducedTreeof{0}{T},m) = 
      \Pop{2}^z\left(\Decode(\inducedTreeof{0}{U},m)\right).
    \end{align*}
    Thus, we conclude that
    $\Decode(\inducedTreeof{0}{T},m) = 
    \Decode(\inducedTreeof{0}{U},m)$ whence $U=T$ as shown above. 
  \end{enumerate}
  If $\Decode(T,m)=\Decode(U,m)$, one of the previous cases applies
  because the following case distinction shows that 
  all other cases for 
  the defined or undefined subtrees of $T$ and $U$ imply
  $\Decode(T,m)\neq\Decode(U,m)$. 
  \begin{enumerate}[\phantom0(1)]
  \item \label{Case1110}
    Assume that $\inducedTreeof{0}{U}=\inducedTreeof{1}{U}=
    \inducedTreeof{0}{T}=\emptyset$ and 
    $\inducedTreeof{1}{T} \neq\emptyset$. In this case, 
    $\Decode(U,m)$ is $[\varepsilon]$ or
    $[\tau]$ for some
    $\tau\in \Sigma\cup(\Sigma\times\{2\}\times\N$). 
    Furthermore, 
    \begin{align*}
      \Decode(T, m) = f_m(T(\varepsilon)) \mathrel{\backslash} (\varepsilon:
      \Decode(\inducedTreeof{1}{T},m+1)).       
    \end{align*}
    It follows  that 
    $\lvert \Decode(T, m) \rvert \geq 2 > 
    \lvert \Decode(U, m) \rvert =  1$ whence
    \mbox{$\Decode(T,m)\neq\Decode(U,m)$.}
  \item \label{Case1101}
    Assume that 
    $\inducedTreeof{0}{U} = \inducedTreeof{1}{U} = \emptyset$, 
    $\inducedTreeof{0}{T} \neq \emptyset$, and 
    $\inducedTreeof{1}{T} = \emptyset$. 
    In this case,  $\Decode(U,m)$ is again $[\varepsilon]$ or
    $[\tau]$ for some
    $\tau\in\Sigma\cup(\Sigma \times\{2\}\times\N)$. 
    Since $U(\varepsilon)= \varepsilon$ if and only if
    $T(\varepsilon)=\varepsilon$,  we conclude that 
    $\lvert f_m(T(\varepsilon) \rvert = \lvert
    \Decode(U,m)\rvert$. Moreover, 
    \begin{align*}
      \Decode(T, m) = f_m(T(\varepsilon)) \mathrel{\backslash}
      f_m(T(0)) \mathrel{\backslash} s      
    \end{align*}
    for some
    $2$-word $s$. Since $T$ is a subtree of a tree in $\EncTrees$,
    $T(0)\in\Sigma\times\{1,2\}$. Thus,
    $f_m(T(0))
    \in\Sigma\cup(\Sigma\times\{1,2\}\times\N)$.
    We conclude that the length of the first word of $\Decode(T, m)$
    is greater than the length of the first word of 
    $\Decode(U, m)$. Thus,
    \mbox{$\Decode(T, m)  \neq \Decode(U, m)$.}
  \item \label{Case1100} Assume that 
    $\inducedTreeof{0}{U} = \inducedTreeof{1}{U} = \emptyset$, 
    $\inducedTreeof{0}{T} \neq \emptyset$, and 
    $\inducedTreeof{1}{T} \neq \emptyset$. 
    Completely analogous to case
    \ref{Case1110}, we conclude that 
    $\lvert \Decode(T,  m) \rvert \geq 2 > 
    \lvert \Decode(U, m) \rvert =  1$ whence
    \mbox{$\Decode(T,m) \neq \Decode(U,m)$.} 
  \item \label{Case1011}
    Assume that 
    $\inducedTreeof{0}{U} = \emptyset$, 
    $\inducedTreeof{1}{U} \neq \emptyset$, and
    $\inducedTreeof{0}{T} =
    \inducedTreeof{1}{T} = \emptyset$. Exchanging the roles of $U$ and
    $T$, this is exactly the same as case \ref{Case1110}.
  \item \label{Case1001}
    Assume that
    $\inducedTreeof{0}{U} = \emptyset$, 
    $\inducedTreeof{1}{U} \neq \emptyset$, 
    $\inducedTreeof{0}{T} \neq \emptyset$, and
    $\inducedTreeof{1}{T} = \emptyset$.
    Analogously to case \ref{Case1101}, we derive that the length of
    the first word of $\Decode(T, m)$ is greater than the length of
    the first word of $\Decode(U, m)$. Thus,
    $\Decode(T, m) \neq \Decode(U, m)$.
  \item \label{Case1000}
    Assume that
    $\inducedTreeof{0}{U} = \emptyset$, 
    $\inducedTreeof{1}{U} \neq \emptyset$, 
    $\inducedTreeof{0}{T} \neq \emptyset$, and
    $\inducedTreeof{1}{T} \neq \emptyset$.
    Analogously to case \ref{Case1101}, we derive that the length of
    the first word of $\Decode(T, m)$ is greater than the length of
    the first word of $\Decode(U, m)$. Thus,
    $\Decode(T, m) \neq \Decode(U, m)$.
  \item \label{Case0111}
    Assume that
    $\inducedTreeof{0}{U} \neq \emptyset$, and
    $\inducedTreeof{1}{U} = 
    \inducedTreeof{0}{T} =
    \inducedTreeof{1}{T} =\emptyset$.
    Exchanging the roles of $U$ and
    $T$, this is exactly the case \ref{Case1101}.
  \item \label{Case0110}
    Assume that
    $\inducedTreeof{0}{U} \neq \emptyset$, 
    $\inducedTreeof{1}{U} = 
    \inducedTreeof{0}{T} =\emptyset$, and
    $\inducedTreeof{1}{T} \neq\emptyset$.
    Exchanging the roles of $U$ and $T$, this is exactly the case
    \ref{Case1001}.
  \item \label{Case0100}
    Assume that
    $\inducedTreeof{0}{U} \neq \emptyset$, 
    $\inducedTreeof{1}{U} = \emptyset$, 
    $\inducedTreeof{0}{T} \neq \emptyset$, and
    $\inducedTreeof{1}{T} \neq \emptyset$.
    In this case, 
    \begin{align*}
      &\Decode(U,m) = f_m(U(\varepsilon))\mathrel{\backslash}
      \Decode(\inducedTreeof{0}{U},m)\\
      \text{and }     
      &\Decode(T,m) = f_m(T(\varepsilon))\mathrel{\backslash}
      \left(\Decode(\inducedTreeof{0}{T},m) :     
        \Decode(\inducedTreeof{1}{T},m+m')\right)      
    \end{align*}
    for some $m'\in\N\setminus\{0\}$. 
    Since $U(\varepsilon) = \varepsilon$ if and only if 
    $T(\varepsilon)= \varepsilon$, we conclude that  $U(\varepsilon) =
    T(\varepsilon)$.  
    Now, 
    \begin{align*}
      \Decode(\inducedTreeof{0}{U},  m) =
      \tau\mathrel{\backslash} u'
    \end{align*}
    for
    $\tau=f_m(U(0))
    \in\Sigma\cup(\Sigma\times\{2\}\times\{m\})$ and $u'$ some level
    $2$-word. We distinguish the following cases. 
    
    First assume that $\tau=(\sigma,2,m)$. 
    For all letters in $T'\coloneqq
    \Decode(\inducedTreeof{1}{T},m+m')$ of 
    collapse level $2$, the collapse link is greater or equal to $m+m'$. Hence,
    $T'$ does not contain a symbol $(\sigma,2,m)$ whence
    $\Decode(U,m) \neq \Decode(T,m)$. 
    
    Otherwise, $\tau\in\Sigma$. But then 
    $\Decode(U,m) = \Decode(T,m)$ would imply that
    \begin{align*}
      &\Decode(\inducedTreeof{0}{T},m) =
      \tau \mathrel\backslash T'\\
      \text{and }
      &\Decode(\inducedTreeof{1}{T},m+m') =
      \tau \mathrel\backslash T''      
    \end{align*}
    for certain nonempty level
    $2$-words 
    $T'$ and $T''$. Since $T(1)=\varepsilon$, it follows that
    \mbox{$T(0)=T(10)=(\tau,1)$ } 
    which contradicts the fact that $T$ is a subtree of some tree from
    $\EncTrees$. 
    
    Thus, we conclude that 
    $\Decode(T, m) \neq \Decode(U,m)$.
  \item[(10)]  \label{0011} 
    Assume that 
    $\inducedTreeof{0}{U} \neq \emptyset$, 
    $\inducedTreeof{1}{U} \neq \emptyset$, and
    $\inducedTreeof{0}{T} =
    \inducedTreeof{1}{T} = \emptyset$. 
    Exchanging the roles of $U$ and $T$, this is the same as case
    \ref{Case1100}.
  \item[(11)] \label{Case0010}
    Assume that
    $\inducedTreeof{0}{U} \neq \emptyset$, 
    $\inducedTreeof{1}{U} \neq \emptyset$,
    $\inducedTreeof{0}{T} = \emptyset$, and
    $\inducedTreeof{1}{T} \neq \emptyset$. 
    Exchanging the roles of $U$ and $T$, this is the same as case
    \ref{Case1000}.
  \item[(12)] \label{Case0001}
    Assume that
    $\inducedTreeof{0}{U} \neq \emptyset$, 
    $\inducedTreeof{1}{U} \neq \emptyset$,
    $\inducedTreeof{0}{T} \neq \emptyset$, and
    $\inducedTreeof{1}{T} = \emptyset$. 
    Exchanging the roles of $U$ and $T$, this is the same as case 
    \ref{Case0100}.
  \end{enumerate}
  Hence, we have seen that 
  $\Decode(T, m)= \Decode(U, m)$ implies that each of the
  subtrees of $T$ is defined if and only if the corresponding subtree
  of $U$ is defined. 
  Under this condition, we concluded that $U=T$.
  Thus, the claim holds
  and the lemma follows as indicated above. 
\end{proof}

Next, we prove that $\Decode$ is a surjective map from
$\EncTrees$ to $\Conf$. This is done by induction
on the size of blocklines used to encode a stack.  
In this proof we use the notion of \emph{left-maximal} blocks and
\emph{good} blocklines. Let 
\begin{align*}
  s:\left(w\mathrel\backslash (w':b)\right) : s'  
\end{align*}
be a stack  
where $s$ and $s'$ are  $2$-words, $w$, and $w'$ are words, and 
$b$ is a $\tau$-block.\footnote{In this definition, we explicitly
  allow the case
  $s= w' = \emptyset$, i.e., a stack of the form 
  $(w\mathrel\backslash b) : s'$.}
We call $b$  \emph{left maximal} in this
stack if either $b=[\tau]$ or 
$b=\tau\tau'\mathrel{\backslash} b'$ such that $w'$ does not
start with $\tau\tau'$ for some
$\tau'\in\Sigma\cup(\Sigma\times\{2\}\times\N)$.  
We call a blockline in some stack \emph{good}, if its first block is left
maximal. Furthermore, we call the blockline starting with the block
$b$ \emph{left maximal} if $w'$ does not start with $\tau$.
Recall that the encoding of stacks works on left maximal blocks
and good blocklines.

\begin{lem} \label{lem:DecEncIsId}
  $\Decode\circ\Encode$ is the identity,
  i.e., $\Decode(\Encode(c)) = c$, for all $c\in \Conf$.
\end{lem}
\begin{cor} \label{DecSurjective}
  $\Decode:\EncTrees\rightarrow \Conf$ is surjective.
\end{cor}
\begin{proof}[Proof of Lemma \ref{lem:DecEncIsId}]
  Let $c=(q,s)$ be a configuration. 
  Since $\Decode$ and $\Encode$ encode and decode the state of $c$ in the
  root of $\Encode(c)$, it suffices  to show that
  \begin{align*}
    \Decode(\Encode(s, (\bot,1)),0) = s    
  \end{align*}
  for all stacks $s\in \Stacks(\Sigma)$. 
  We proceed by induction on
  blocklines of the stack $s$. For this purpose we reformulate the lemma
  in the following claim.
  \begin{claim}
    Let $s'$ be some stack which decomposes as
    $s'=s'':(w\mathrel\backslash b):s'''$ such that
    \mbox{$b\in(\Sigma\cup(\Sigma\times\{2\}\times\N))^{+2}$} is a good  
    $\tau$-blockline for some
    $\tau\in\Sigma\cup(\Sigma\times\{2\}\times\N)$.   
    Then
    \begin{enumerate}[(1)]
    \item $\Decode(\Encode(b,\varepsilon),
      \lvert s''\rvert )= b'$ for the unique $2$-word $b'$ such that
      \mbox{$b=\tau\mathrel{\backslash}b'$} and 
    \item if $b$ is left maximal, then
      $\Decode(\Encode(b,(\sigma,l)),\lvert s''\rvert)=
      b$ where $\sigma=\Sym(\tau)$ and \mbox{$l=\Lvl(\tau)$.} 
    \end{enumerate}
  \end{claim}
  Note that the conditions in the second part require that 
  either $\tau\in\Sigma$ or $\tau=(\sigma,2,\lvert s'' \rvert)$ for
  some $\sigma\in\Sigma$. 

  The lemma follows from the second part of the claim because every stack
  is a left maximal $\bot$-blockline. 

  We prove both claims by parallel induction on the size of $b$. 
  As an abbreviation we set
  \mbox{$g\coloneqq\lvert s'' \rvert$}. We
  write $\overset{(1)}{=}$($\overset{(2)}{=}$, respectively)  when
  some equality is due to the induction hypothesis of the first claim
  (the second claim , respectively). 
  The arguments for the first claim are as follows.
  \begin{iteMize}{$\bullet$}
  \item If $b=[\tau]$ for
    $\tau\in\Sigma\cup(\Sigma\times\{2\}\times\N)$,  the claim  is 
    true because
    \begin{align*}
      \Decode(\Encode(b,\varepsilon), g)=
      \Decode(\varepsilon,  g) = \varepsilon.      
    \end{align*}
  \item  If there are $b_1,b_1'\in(\Sigma\cup(\Sigma\times\{2\}\times\N))^{*2}$
    such that 
    \begin{align*}
      b= [ \tau ] : b_1 = [\tau]
      :\left(\tau\mathrel{\backslash} b_1'\right)\text{ then}
    \end{align*}
    \begin{align*}
      &\Decode(\Encode(b,\varepsilon), g) = \Decode(
      \treeR{\varepsilon}{ \Encode(b_1,\varepsilon)},g)\\
      =   
      &f_g(\varepsilon) \mathrel{\backslash} \left(\varepsilon:
      \Decode(\Encode(b_1,\varepsilon),g+1)\right) \\
      \overset{(1)}{=}
      &\varepsilon \mathrel{\backslash}(\varepsilon : b_1') =
      \varepsilon:b_1'=b'.
    \end{align*}
  \item Assume that there is some
    $\tau'\in\Sigma\cup(\Sigma\times\{2\}\times\N)$ and some 
    $b_1\in (\Sigma\cup (\Sigma\times\{2\}\times\N))^{*2}$ such that
    \begin{align*}
      b=\tau\tau' \mathrel{\backslash} b_1.      
    \end{align*}
    The assumption that $b$ is good implies that  the blockline
    $\tau'\mathrel\backslash b_1$ is left maximal 
    whence
    \begin{align*}
      &\Decode(\Encode(b,\varepsilon),g) = 
      \Decode(    
      \treeL{\varepsilon}{
      \Encode( \tau'\mathrel{\backslash}
      b_1, (\Sym(\tau'), \Lvl(\tau')))}, g)\\
      = &f_g(\varepsilon) \mathrel{\backslash} 
      \Decode(\Encode( \tau'\mathrel{\backslash} b_1,
      (\Sym(\tau'), \Lvl(\tau')), g))\\
      \overset{(2)}{=}  &\tau' \mathrel{\backslash} b_1 =
      b'.
    \end{align*}
  \item The last case is that
    \begin{align*}
      b = \tau \mathrel{\backslash}
      \left((\tau' \mathrel{\backslash} b_1):b_2\right)
    \end{align*}
    for  $b_2$ a blockline of $s$ not starting with
    $\tau'$. By this we mean that
    $b_2\neq \tau' w':b_2'$ for any word $w'$ and any $2$-word
    $b_2'$. 
    Since $b$ is good,  
    \mbox{$\tau' \mathrel{\backslash} b_1$} is a left maximal
    blockline. Furthermore, $\tau \mathrel{\backslash} b_2$ is a
    good blockline.  
    Thus, 
    \begin{align*}
    &\Decode(\Encode(b,\varepsilon), g)\\
    =
    &\Decode\left(
    \treeLR{\varepsilon}{
      \Encode\left( \tau'\mathrel{\backslash}
      b_1, (\Sym(\tau'), \Lvl(\tau'))\right)}{
      \Encode(\tau\mathrel{\backslash}b_2,\varepsilon)},g\right)\\
    = 
    &f_g(\varepsilon)\mathrel{\backslash}
    \left(\Decode(\Encode(\tau'\mathrel{\backslash} b_1,
    ( \Sym(\tau'), \Lvl(\tau'))),g) : 
    \Decode(\Encode(\tau\mathrel{\backslash} b_2,
    \varepsilon),g + f)\right), 
    \end{align*}
    where
    \begin{align*}
      f = \lvert
    \Decode(\Encode(\tau'\mathrel{\backslash} b_1, 
    (\Sym(\tau'), \Lvl(\tau'))),g)  \rvert  \overset{(2)}{=}
    \lvert b_1\rvert.  
    \end{align*}
    From this, we obtain that
    \begin{align*}
      &\Decode(\Encode(b,\varepsilon),g) \\
      \overset{(2)}{=}
      &\varepsilon \mathrel{\backslash}\big(
      (\tau'\mathrel{\backslash} b_1) :
      \Decode(\Encode(\tau \mathrel{\backslash} b_2,
      \varepsilon), g+f)\big) \\
      \overset{(1)}{=}
      &(\tau'\mathrel{\backslash} b_1) : b_2 = b'.
    \end{align*}
  \end{iteMize} 
  For the proof
  of the second claim, note that the calculations are basically the
  same, but $f_g(\varepsilon)$ is replaced by $f_g(\sigma,l)$.
  Thus, if $l=1$ then $f_g(\sigma,l) = \sigma = \tau$. 
  For the case  $l=2$, recall that $g=\lvert s'' \rvert$ whence
  $f_g(\sigma,l) = (\sigma,2,\lvert s'' \rvert)$. 
  Note that $\Lnk(\tau)=\lvert s''\rvert$ due to the left maximality
  of $b$. 
  
  Thus, one proves the second case using
  the same calculations, but replacing
  $\varepsilon$ by $\tau$.
\end{proof}

From the previous lemmas, we directly obtain
Lemma \ref{STACS:Bijective}, i.e.,  we obtain that $\Encode$ is
bijective.

\section{Automaton for 
  Relation   \texorpdfstring{\relA}{R-Leftarrow}}
\label{Appendix:AA}
Given a \CPS $\mathcal{S}$,  there is an automaton
$\mathcal{A}_{\relA}$ that accepts the tree
$\Encode(c_1)\otimes\Encode(c_2)$ for arbitrary configurations $c_1$
and $c_2$ if and only if
$c_2=\Pop{2}^k(c_1)$ and there is a run $\rho$ from $c_1$ to $c_2$ such that
$\rho(j) \not\leq c_2$ for all $j<\length(\rho)$, i.e., if and only if
$(c_1,c_2)\in\relA$. 
The states of $\mathcal{A}_{\relA}$ come from the set
\begin{align*}
  &\{\bot, q_I, q_\emptyset, q_=\} \cup M\text{ with }  
  M:=\{S, R, P, C\} \times Q \times Q \times \Sigma \times
  \{1,2\} \times 2^{Q\times Q}.
\end{align*}

Before giving a definition of $\mathcal{A}_{\relA}$, we informally
describe how 
$\mathcal{A}_{\relA}$ processes some tree
$T:=\Encode(c_1)\otimes\Encode(c_2)$. 
An accepting run on $T$ labels $d\in\{0,1\}^*$
\begin{enumerate}[\phantom0 C1.]
\item \label{ConditionAC1}by $\bot$ if $d\in T_+$;
\item by $q_I$ if $d=\varepsilon$; it is the initial state in which
  the automaton reads the states of $c_1$ and $c_2$, before it 
  processes the encodings of the stacks,
\item by $q_=$ if $d\in\Encode(c_2)$ but not in the rightmost branch of
  $\Encode(c_2)$; this state is used to check equality of the parts
  of $\Encode(c_1)$ and $\Encode(c_2)$ that are left of the rightmost
  branch of $\Encode(c_2)$,
\item by some element from $\{S\} \times Q \times Q \times \Sigma \times
  \{1,2\} \times 2^{Q\times Q}$ if $d$ is in the
  rightmost branch of $\Encode(c_2)$; $S$ stands for searching the
  node encoding the final configuration of the run. 
\item by some element from the set
  $\{q_\emptyset\} \cup( \{ R, P, C\} \times Q \times Q \times
  \Sigma \times \{1,2\} \times 2^{Q\times Q}$) if
  $d\in\Encode(c_1)\setminus\Encode(c_2)$.
\end{enumerate}
Those labels that come from $M$ are used to check the existence of
some run from 
$c_1$ to $c_2$ as follows. For all $d\in\Encode(c_1)$, let $d^\nearrow$
be the rightmost leaf of the subtree induced by $d$ in
$\Encode(c_1)$. Set $s_d^\nearrow:=\LeftStack(d^\nearrow,
\Encode(c_1))$. Let $\bar q\in M$ be the label of some node $d$. By
$\pi_i(\bar q)$ we denote the projection of $\bar q$ to the $i$-th
component. 
Depending on  $\pi_1(\bar q)$ we define a stack $s_d$ as
follows. 
\begin{enumerate}[(1)]
\item If $\pi_1(\bar q)=S$, set $s_d$ to be the stack of $c_2$.
\item If $\pi_1(\bar q) = R$, set $s_d:=\Pop{2}(\LeftStack(d,
  \Encode(c_1)))$. 
\item If $\pi_1(\bar q) = C$, set
  $s_d:=\Collapse(\LeftStack(d,\Encode(c_1)))$.
\item If $\pi_1(\bar q)=P$, set $s_d:=\Pop{1}(\LeftStack(d,\Encode(c_1)))$.
\end{enumerate}
An accepting run $\rho$ of $\mathcal{A}_{\relA}$ will label some node $d$ by
$\bar q\in M$ such that
\begin{enumerate}[\phantom0 C1.]
  \setcounter{enumi}{5}
\item $\pi_4(\bar q) = \Sym(\LeftStack(d, \Encode(c_1)))$,
\item $\pi_5(\bar q) = \Lvl(\LeftStack(d, \Encode(c_1)))$, and
\item \label{ConditionACLast}$\pi_6(\bar q) = \ExRet(\LeftStack(d,
  \Encode(c_1)))$.
\item \label{ConditionACNOTP} Moreover, if $\pi_1(\bar q)\neq P$ 
then there is a run $\rho$ from $(\pi_2(\bar q), s_d^\nearrow)$
to $(\pi_3(\bar q), s_d)$ 
(which is an infix of some run witnessing $(c_1,c_2)\in\relA$). 
The meaning of the labels $R$ and $C$ is as follows.
If $\pi_1(\bar q)=R$ then $\rho$ ends in 
$\Pop{2}(\LeftStack(d,\Encode(c_1)))$. 
If $\pi_1(\bar q)=C$ then $\rho$ ends in 
$\Collapse(\LeftStack(d,\Encode(c_1)))$ and the collapse level is $2$
(moreover, $\rho$ actually performs as the last operation a collapse
on a copy of the topmost element of $\LeftStack(d, c_1)$). 
Thus, in both cases the run will end in the stack $s_d$. To be more
precise, the run ends in $s_d$ and does not visit any substack of
$s_d$ before its final configuration. 
\item \label{ConditionACP}
If $\pi_1(\bar q) = P$ then there is 
some stack $s'$ with $s_d\leq s'$ and
$\TOP{2}(s_d)=\TOP{2}(s')$ such that there is a run from 
$(\pi_2(\bar q), s_d^\nearrow)$
to $(\pi_3(\bar q), s')$ (which is again an infix of some run
witnessing $(c_1,c_2)\in\relA$).
\item 
$\rho$  will label $d$ by $q_\emptyset$ if 
there is a run from $c_1$ to 
$c_2$ not passing $s_e^\nearrow$ for all $d\leq e$.\vspace{3 pt}
\end{enumerate}

Let us fix some notation. In  this section, $\gamma$ ranges over
$\Gamma$, 
$y$ ranges over \mbox{$(\Sigma\times\{1,2\})\cup\{\varepsilon\}$} and $w$ ranges
over all words of the form $w=\TOP{2}(s){\downarrow}_0$. 
Whenever $w$ is fixed, we write $\sigma:=\Sym(w)$ and
$l=\Lvl(w)$. Furthermore, $x$ ranges over 
$\{(\sigma,l), \varepsilon\}$.  
The variables $q_1, q_2, q_1', q_2', q, q'$ range over 
$Q$. $\tau$ ranges over $\Sigma\setminus\{\bot\}$ and $k$ over
$\{1,2\}$. 
We use the abbreviation
\mbox{``$(q, \sigma, q', \ColPop{k})\in \Delta$''} for ``$\exists
\gamma$ such that
\begin{enumerate}[(1)]
\item $(q, \sigma, \gamma, q', \Pop{1})\in \Delta$ or
\item $(q, \sigma, \gamma, q', \Collapse)\in \Delta$ and $k=1$''. 
\end{enumerate}
If $w, \tau$ and $k$ are fixed, we write $w\tau_k$ for
the word 
$w\theta$ where $\theta=
\begin{cases}
  (\tau,2,0)&\text{if }k=2,\\
  \tau &\text{if }k=1.
\end{cases}$

\begin{defi}\label{Def:AA}
  Fix some \CPS $\mathcal{S}=(Q,\Sigma, \Gamma, \Delta, q_0)$. 
  Define   
  $\mathcal{A}_{\relA}:=(Q_A, \Sigma_A, \bot, F, \Delta_A)$ with
  $Q_A:= \{\bot, q_I, q_\emptyset, q_=\} \cup M$,
  $\Sigma_A=\left(\{\varepsilon, \Box\}\cup\{\Sigma\times
    \{1,2\}\}\cup Q\right)^2$,  
  $F=\{q_I\}$. $\Delta_A$ contains the following transitions.
  \begin{enumerate}[\phantom0 T1.]
  \item $\left(q_I, (q_1,q_2), 
    (S, q_1, q_2, \bot, 1, \ExRet(\bot_2)),   \bot\right)$,
  \item $(q_\emptyset, (y, \Box), Y, Z)$ for 
    $Y, Z\in\{\bot, q_\emptyset\}$, and
  \item $(q_=, (y, y), Y, Z)$ for $Y,  Z\in\{\bot, q_=\}$,
  \end{enumerate}
  Fix some 
  $\bar q:=(S, q_1, q_2, \sigma, l, \ExRet(w))$.
  We add the following transitions to $\Delta_A$:
  \begin{enumerate}[\phantom0 T1.]
    \setcounter{enumi}{3}
  \item  $(\bar q, (x,x), \bot, \bot)$ if  $q_1=q_2$; \label{AA:T4} 
  \item  $(\bar q, (x,x), \bot, \bar q_1)$ for  
    $\bar q_1=(R, q_1, q_2, \sigma, l, \ExRet(w))$ ;
  \item  $(\bar q, (x,x), X, \bar q)$ for $X\in\{\bot, q_=\}$;
  \item  $(\bar q, (x,x), \bar q_0, \bot)$ for 
    $\bar q_0=(S, q_1, q_2, \tau, k,\ExRet(w\tau_k))$;
  \item 
    $(\bar q, (x,x), \bar q_0, \bar q_1)$ for 
    $\bar q_1=(R, q_1, q_2', \sigma, l, \ExRet(w))$ and
    $\bar q_0=(S,q_2',q_2,\tau,k,\ExRet(w\tau_k))$.
  \end{enumerate}
  Fix some $\bar q:=(R, q_1, q_2, \sigma, l, \ExRet(w))$. We add the
  following transitions to $\Delta_A$:  
  \begin{enumerate}[\phantom0 T1.]
    \setcounter{enumi}{8}
  \item $(\bar q, (x, \Box), \bot, \bot)$  \label{AA:T10}
    if $(q_1, q_2)\in \ExRet(w)$;
  \item $(\bar q, (x, \Box), \bot, \bar q_1)$ for  
    $\bar q_1=(R, q_1, q_2', \sigma, l, \ExRet(w))$
    such that
    $(q_2', q_2)\in\ExRet(w)$;
  \item 
    $(\bar q, (x, \Box), \bar q_0, \bot)$ for 
    $\bar q_0=(R, q_1, q_2, \tau, i, \ExRet(w\tau_i))$ for $i\in\{1,2\}$;
  \item 
    $(\bar q, (x, \Box), \bar q_0, \bar q_1)$ for      
    \begin{align*}
      &\bar q_1=(R, q_1, q_2', \sigma, l, \ExRet(w)),\\
      &\bar q_0=(R, q_2', q_2, \tau, i, \ExRet(w\tau_i))
      \text{ and}\\
      &i\in\{1,2\};      
    \end{align*}
  \item 
    $(\bar q, (x, \Box), \bar q_0, \bot)$ for 
    $\bar q_0=(C, q_1, q_2, \tau, 2, \ExRet(w\tau_2))$;
  \item 
    $(\bar q, (x, \Box), \bar q_0, \bar q_1)$ for      
    $\bar q_1=(R, q_1, q_2', \sigma, l, \ExRet(w))$ and
    $\bar q_0=(C, q_2', q_2, \tau, 2, \ExRet(w\tau_2))$.
  \end{enumerate}
  Fix some $\bar q:=(P, q_1, q_2, \sigma, l, \ExRet(w))$. We add the
  following transitions to $\Delta_A$:
  \begin{enumerate}[\phantom0 T1.]
    \setcounter{enumi}{14}
  \item \label{AA:T25}
    $(\bar q, (x, \Box), \bot, \bot)$ if there
    is a $q$ with
    $(q_1, q)\in \ExOneLoop(w)$  and 
    $(q, \sigma, q_2, \ColPop{l})\in \Delta$;
  \item 
    $(\bar q, (x, \Box), X, \bar q)$ for $X\in \{\bot, q_\emptyset\}$;
  \item 
     $(\bar q, (x, \Box), \bar q_0, \bot)$ for 
     \begin{align*}
       &\bar q_0=(P, q_1, q_2', \tau, k, \ExRet(w\tau_k)),\\ 
       &(q_2', q)\in \ExOneLoop(w)\text{ and}\\
       &(q, \sigma, q_2, \ColPop{l})\in \Delta;       
     \end{align*}
  \item 
    $(\bar q, (x, \Box), \bar q_0, \bar q_1)$ for 
    \begin{align*}
      &\bar q_1=(R, q_1, q_1', \sigma, l, \ExRet(w)),\\
      &\bar q_0=(P, q_1', q_2', \tau, k, \ExRet(w\tau_k)),\\
      &(q_2', q)\in   \ExOneLoop(w)\text{ and}\\
      &(q, \sigma, q_2, \ColPop{l})\in\Delta. 
    \end{align*}
  \end{enumerate}
  Fix some $\bar q:=(C, q_1, q_2, \sigma, l, \ExRet(w))$ with $l=2$
  (whence  $x$ ranges here over $\{(\sigma,2),\varepsilon\}$). We add the 
  following transitions to $\Delta_A$:  
  \begin{enumerate}[\phantom0 T1.]
    \setcounter{enumi}{18}
  \item \label{AA:T30}
    $(\bar q, (x, \Box), \bot, \bot)$ if there is a $q$ 
    with $(q_1, q)\in  \ExOneLoop(w)$  and 
    $(q, \sigma, \gamma, q_2, \Collapse)\in \Delta$;
  \item $(\bar q, (x, \Box), X, \bar q)$ for $X\in\{\bot, q_\emptyset\}$;
  \item 
    $(\bar q, (x, \Box), \bot, \bar q_1)$  for 
    \begin{align*}
      &\bar q_1 = (R, q_1, q_2', \sigma, 2, \ExRet(w)),\\
      &(q_2', q) \in \ExOneLoop(w) \text{ and}\\
      &(q, \sigma, \gamma, q_2, \Collapse)\in \Delta;      
    \end{align*}
  \item $(\bar q, (x, \Box), \bar q_0, \bot)$  for 
    \begin{align*}
      &\bar q_0 = (P, q_1, q_2', \tau, k, \ExRet(w\tau_k)),\\
      &(q_2', q)\in \ExOneLoop(w) \text{ and}\\
      &(q, \sigma, \gamma, q_2, \Collapse) \in \Delta;      
    \end{align*}
  \item $(\bar q, (x, \Box), \bar q_0, \bar q_1)$ for 
    \begin{align*}
      &\bar q_1 = (R, q_1, q_1', \sigma, 2, \ExRet(w)),\\
      &\bar q_0 = (P, q_1', q_2', \tau, k, \ExRet(w\tau_k)),\\
      &(q_2', q)\in \ExOneLoop(w)\text{ and}\\
      &(q,\sigma, \gamma, q_2,  \Collapse)\in\Delta. 
    \end{align*}
  \end{enumerate}  
\end{defi}

\begin{lem}
  If $\mathcal{A}_{\relA}$ accepts a tree $\Encode(c_1)\otimes \Encode(c_2)$
  for configurations $c_1, c_2$, then \mbox{$c_2 = \Pop{2}^k(c_1)$} and there is
  some run from $c_1$ to $c_2$ that does not reach a substack of $c_2$
  before the final configuration, i.e., $(c_1,c_2)\in \relA$. 
\end{lem}
\begin{proof}
  Assume that $\rho_{\relA}$ is an accepting run of $\mathcal{A}_{\relA}$ on 
  $T:=\Encode(c_1) \otimes \Encode(c_2)$. 
  A straightforward induction from the root to the leaves shows that
  $c_1=\Pop{2}^k(c_2)$, that 
  Conditions C\ref{ConditionAC1} -- C\ref{ConditionACLast} hold and
  that the rightmost leaf of $T$ is labelled by some element of $M$. 
  Furthermore, $\rho_{\relA}(0)\in M$ and $(\pi_1(\rho_{\relA}(0)),
  \pi_2(\rho_{\relA}(0)),
  \pi_3(\rho_{\relA}(0)))=(S,q_1, q_2)$ for $c_i=(q_i,s_i)$. 
  Moreover, note that  
  $\rho_{\relA}(d)\in M$ and $\pi_1(\rho_{\relA}(d)) = C$ implies 
  $\Lvl(\LeftStack(d, c_1))=2$: if a transition at some node $d$
  labels the $0$-successor $d0$ by $C$, then we always have
  $\pi_5(\rho_{\relA}(d0))=2$. By construction, the
  transition of $\rho_{\relA}$ applied at $d0$ enforces that
  $\pi_5(\rho_{\relA}(d0))$ is the link level encoded in the tree at
  $d0$. Thus, the claim holds for all $0$ successors. 
  Moreover, a $1$-successor is labelled by $C$ only if its predecessor
  is also labelled by $C$. Thus, by induction on the distance to the
  first ancestor which is a $0$-successor the claim holds also for
  $1$-successors. 
  
  By induction from the leaves to the root, we show that 
  Conditions C\ref{ConditionACNOTP} and C\ref{ConditionACP} hold. 
  This completes the proof, because $\rho_{\relA}(0)$ then witnesses that
  there is a run from $c_1$ to $c_2$ not passing a substack of $c_2$
  before its final configuration. 
  For the base case, assume that $d\in T$ is a leaf labelled by some
  $\bar q\in M$. Depending on $\pi_1(\bar q)$ we have the following cases.
  \begin{iteMize}{$\bullet$}
  \item If $\pi_1(\bar q)=S$, then $d$ is the rightmost leaf of
    $\Encode(c_2)$. Thus, $s_d^\nearrow=s_d=s_2$. 
    Since $\rho$ is an accepting run, it uses a transition of the form
    T\ref{AA:T4}. Thus, $\pi_2(\bar q)=\pi_3(\bar q)$ and Condition
    C\ref{ConditionACNOTP} is trivially satisfied.
  \item If $\pi_1(\bar q)=R$,  $\rho$ applies a transition of the
    form T\ref{AA:T10}. 
    Recall that $d$ is a leaf whence \mbox{$s_d=\Pop{2}(\LeftStack(d,
      c_1))$},  $s_d^\nearrow=\LeftStack(d, c_1)$, and
    \mbox{$\pi_6(\bar q) = \ExRet(\LeftStack(d, c_1))$}. Thus, the
    condition in T\ref{AA:T10} ensures that there is a run from 
    $(\pi_2(\bar q), s_d^\nearrow)$ to $(\pi_3(\bar q), s_d)$, i.e., 
    Condition C\ref{ConditionACNOTP} holds.
  \item If $\pi_1(\bar q)=C$, then $\rho$ applies a transition of the
    form T\ref{AA:T30}. Since the collapse level of $s_d^\nearrow$ is
    $2$, we conclude
    analogously to the previous case that
    Condition C\ref{ConditionACNOTP} holds.
  \item If $\pi_1(\bar q)=P$, the transition of $\rho_{\relA}$ at $d$ is of
    the form T\ref{AA:T25}.
    Since $s_d^\nearrow=\LeftStack(d, c_1)$ the conditions
    of T\ref{AA:T25} ensure there exists some stack $s'$ with
    $s_d^\nearrow  \leq s'$ and 
    \mbox{$\TOP{2}(s') = \TOP{2}(s_d^\nearrow)$} such that there
    is a $1$-loop from $s_d^\nearrow$ to $s'$ followed by a $\Pop{1}$
    operation or a collapse of level $1$. Note that $s_d\leq
    \Pop{1}(s')$ and $\TOP{2}(s_d) = \TOP{2}(\Pop{1}(s'))$. Thus, Condition
    C\ref{ConditionACP} is satisfied. 
  \end{iteMize}
  A tedious but easy case distinction shows that Conditions
  C\ref{ConditionACNOTP} and C\ref{ConditionACP} carry over to all nodes
  of $T$. Instead of giving the full case distinction, we mention
  briefly the underlying ideas. 
  \begin{enumerate}[(1)]
  \item If $d0\in \Encode(c_1)\setminus\Encode(c_2)$, $d1\in
    \Encode(c_1)\setminus\Encode(c_2)$ and $\rho_{\relA}(d0)\in M$,
    then also 
    \mbox{$\rho_{\relA}(d1)\in M$} and
    $\pi_1\left(\rho_{\relA}(d1)\right)=R$ whence 
    $s_{d1}=s_{d0}^\nearrow$. Thus, we can compose the run associated
    to $d1$ with the run associated to $d0$ and obtain a run
    associated to $d$.  
  \item If $i\in\{0,1\}$ minimal such that $di\in\Encode(c_1)$, then 
    either $s_{di}=s_d$ such that the final part of the run associated
    to $s_{di}$ can serve as final part of the run associated to $s_d$
    or
    $s_{di}= \LeftStack(d, c_1)$ and the conditions on the
    transition at $d$ ensure that this run can be extended to a run to
    $s_d$ if $\pi_1\left(\rho_{\relA}(d)\right)\neq P$.
    If $\pi_1\left(\rho_{\relA}(d)\right) = P$, this run can be
    extended to some stack $s'$ with $s_d\leq s'$ and
    $\TOP{2}(s_d)=\TOP{2}(s')$. 
  \item If $i\in\{0,1\}$ is maximal such that $di\in\Encode(c_1)$ then
    $s_d^\nearrow=s_{di}^\nearrow$ and 
    the run associated to $di$ may serve as initial part of the run
    associated to $d$. \qedhere
  \end{enumerate}
\end{proof}

\begin{lem}
  Let $c_1=(q,s_1), c_2=(q',s_2)$ be configurations such that
  $s_2=\Pop{2}^k(s_1)$ and there is a run from $c_1$ to $c_2$ that
  passes a substack of $s_2$ only in its final configuration. Then
  there is an accepting run of $\mathcal{A}_{\relA}$ on
  $\Encode(c_1)\otimes\Encode(c_2)$. 
\end{lem}
\begin{proof}
  Let $\rho$ be some run from $c_1$ to $c_2$. 
  Recall the decomposition
  $\rho=\rho_1\circ\rho_2\circ\dots\circ\rho_n$
  provided by Lemma \ref{FormLemma}. 
  Let $\rho_0:=\rho{\restriction}_{[0,0]}$ and
  let $q_i$ denote the final state of $\rho_i$ for all $0\leq i \leq
  n$.
  In the following, we will use the notation 
  $\hat d:=\LeftStack(d, c_1)$ for all $d\in\{0,1\}^*$. 
  Let $d$ be the rightmost leaf of
  $\Encode(c_1)\otimes\Encode(c_2)$. Note that  
  \begin{align*}
    \hat d=s_1=\rho(0)=\rho_0(0)=\rho_0(\length(\rho_0))    
  \end{align*}
  whence $\rho_0$ ends in $(q_0, \hat d)$. 
  We define an accepting run $\rho_{\relA}$
  of $\mathcal{A}_{\relA}$
  on \mbox{$T:=\Encode(c_1)\otimes \Encode(c_2)$}  by induction 
  as follows. Let $d\in T$ be the
  lexicographically maximal node of
  $\Encode(c_1)$  such that $\rho_{\relA}$ has not been defined yet at $d$.
  Assume that
  there is some maximal $i\geq 1$ such that $\rho_i(0)=(q_{i-1},
  s)$ for some stack $s$ satisfying 
  $\Pop{2}( \hat d)\leq \Pop{2}(s)$ and 
  $\TOP{2}(\hat d) = \TOP{2}(s)$. Furthermore, 
  assume that $s=\hat d$ if $\rho_{i-1}$ is not
  of the form F\ref{1LoopPop}. 
  Depending on the form of
  $\rho_i$, we proceed as follows. 
  \begin{enumerate}[(1)]
  \item If $\rho_i$ is of the Form F\ref{FormReturn}, 
    let $d'$ be the minimal element such that
    $d=d'0^m$ for some $m\in\N$. 
    We set 
    $\rho_{\relA}(d):= (R, q_{i-1}, q_i, \Sym(\hat d),
             \Lvl(\hat d), \ExRet( \hat d))$ and
    for  $d'\leq e < d$ we set
    $\rho_{\relA}(e):= (R, q_e, q_i, \Sym(\hat e),
    \Lvl(\hat e), \ExRet( \hat e))$
    where
    $q_e= \pi_2\left(\rho_{\relA}(ej)\right)$ for 
    $j=\max\{i\in\{0,1\}:  di\in \Encode(c_1)\}$.
  \item If $\rho_i$ is of the Form F\ref{1LoopCol}, 
    let $e'$ be minimal such that $d=e'0^{m_0}1^{m_1}$ for some
    \mbox{$m_0,m_1\in\N$}. 
    For each $e$ satisfying $e'0^{m_0} \leq e \leq d$ and
    for all
    $e0\leq f \in\Encode(c_1)$
    we set
    $\rho_{\relA}(e):=(C,q_{i-1},q_i,
    \Sym(\hat e),
    \Lvl(\hat e),
    \ExRet(\hat e))$ and
    $\rho_{\relA}(f):=q_\emptyset$. 

    For all $e$ with $e'\leq e < e'0^{m_0}$ we define
    $\rho_{\relA}(e):=
    (R, q_e, q_i, \Sym(\hat e),
    \Lvl(\hat e),
    \ExRet(\hat e))$ where $q_e$ is
    defined as in the previous case.    
  \item If $\rho_i$ is of the Form F\ref{1LoopPop}, 
    then we proceed as follows.
    \begin{iteMize}{$\bullet$}
    \item If $d$ is a leaf of
    $\Encode(c_1)$, set
    $\rho_{\relA}(d):=(P, q_{i-1}, q_i, 
      \Sym(\hat d),
      \Lvl(\hat d),
      \ExRet(\hat d)
    )$;
  \item otherwise, let $j\in\{0,1\}$ be maximal such that
    $dj\in\Encode(c_1)$. Then we set
    $\rho_{\relA}(d):=\left(P, \pi_2\left(\rho_{\relA}(ej)\right), q_i, 
      \Sym(\hat d),
    \Lvl(\hat d),
    \ExRet(\hat d)\right)$.
    \end{iteMize}
    In case that there is some $e\in\{0,1\}^*$ such that
    $d=e1$, then define $\rho_{\relA}(e0f):=q_\emptyset$ for all
    $f\in\{0,1\}^*$ such that $e0f\in\Encode(c_1)$.  
  \end{enumerate}
  These rules define $\rho_{\relA}$ on $\Encode(c_1)\setminus\Encode(c_2)$. 
  Let $d$ be the rightmost leaf of $\Encode(c_2)$, and
  set $\rho_{\relA}(d):=\left(S,q_n,q_n, \Sym(\hat d), 
  \Lvl(\hat d),
  \ExRet(\hat d)\right)$.
  Let $0\leq d$ be the maximal element in the rightmost branch of
  $\Encode(c_2)$ such that $\rho_{\relA}(d)$ is undefined. 
  Let $j\in\{0,1\}$ be maximal such that $dj\in\Encode(c_1)$. 
  We set $\rho_{\relA}(d):=\left(S, \pi_2(\rho_{\relA}(dj)), q_n, 
  \Sym(\hat d),  
  \Lvl(\hat d),
  \ExRet(\hat d)\right)$.
  We complete the definition by $\rho_{\relA}(\varepsilon)=q_I$ and
  $\rho_{\relA}(d):=q_=$ for all $d\in\Encode(c_2)$ that are not in the
  rightmost branch of $\Encode(c_2)$. 
  A tedious, but straightforward induction shows that 
  $\rho_{\relA}$ is an accepting run of $\mathcal{A}_{\relA}$ on
  $\Encode(c_1)\otimes\Encode(c_2)$. 
\end{proof}

\section{Automaton for 
  Relation   \texorpdfstring{\relB}{R-Downarrow}}
\label{Appendix:AB}
In the following definition, $w$ ranges over words, $\tau$ over
letters from $\Sigma\setminus\{\bot\}$, $k$ over $\{1,2\}$,
$q_i, q_e, q_e'$ over $Q$ and $z,z'$ over $\{q_=, \bot\}$. Whenever we
have fixed a word $w$, then 
$\sigma:=\Sym(w)$, $l:=\Lvl(w)$ and $x$ ranges over $\{(\sigma,l),
\varepsilon\}$.  

\begin{defi}
  $\mathcal{A}:=\left(Q_{\mathcal{A}},\Sigma_{\mathcal{A}},\bot,\{q_I\},
  \Delta_{\mathcal{A}}\right)$
  where 
  \begin{iteMize}{$\bullet$}
  \item $\Sigma_{\mathcal{A}}:=( Q\cup
    (\Sigma\times\{1,2\})\cup\{\varepsilon, \Box\})^2$, 
  \item  $Q_{\mathcal{A}}:= 
    \{q_I, \bot, q_=, (\Box, \varepsilon) \} \cup (Q\times Q \times
    2^{Q\times Q} \times \Sigma \times
    \{1,2\}\times \{S, P_1, P_2\})$, and 
  \item $\Delta_\mathcal{A}$ contains the following transitions:
    \begin{enumerate}
    \item $(q_I, (q_1, q_2), (q_1, q_2, \emptyset, \bot, 1, S), \bot)$;
    \item $(q_=, (y,y), z, z')$ for all       
      $y\in (\Sigma\times\{1,2\}) \cup \{\varepsilon\}$;
    \item \label{TransBoxVarepsilon} $((\Box, \varepsilon), (\Box,
      \varepsilon), \bot, \bot)$; 
    \end{enumerate}
    now fix an arbitrary
    $\bar q=(q_i, q_e, \ExRet(\Pop{1}(w)), \Sym(w), \Lvl(w), S)$. 
    $\Delta_{\mathcal{A}}$ contains
    \begin{enumerate}
      \setcounter{enumi}{3}
    \item $(\bar q, (x,x), \bar q_0, \bot)$  for each
      $\bar q_0=(q_i, q_e, \ExRet(w), \tau, k, S)$;
    \item 
      $(\bar q, (x,x), z, \bar q)$;
    \item $(\bar q, (x,x), \bot, \bot)$ if $q_i=q_e$;
    \item \label{RelBT7}$(\bar q, (x,x), \bar q_0^2, \bot)$ and
      $(\bar q, (x,x), \bar q_0^1, (\Box, \varepsilon))$
      for
      $\bar q_0^j=(q_i, q_e', \ExRet(w), \tau,
      k, P_j)$  such that there is some $q\in Q$ with 
      $(q_e', q)\in\ExHLoop(w\tau_k)$ and 
      $(q, \tau, q_e, \ColPop{k})\in\Delta$;
    \end{enumerate}
    now fix an arbitrary
    $\bar q=(q_i, q_e, \ExRet(\Pop{1}(w)), \Sym(w), \Lvl(w), P_2)$. 
    $\Delta_\mathcal{A}$ contains
    \begin{enumerate}
      \setcounter{enumi}{7}
    \item\label{RelBT8} $(\bar q, (x,\Box), \bot, \bot) $ if $q_i=q_e$;
    \item\label{RelBT9} $(\bar q, (x,\Box), \bar q_0, \bot) $  for 
      $\bar q_0=(q_i, q_e', \ExRet(w), \tau,
      k, P_2)$ such that there is a $q\in Q$ with 
      $(q_e', q)\in \ExHLoop(w\tau_k)$ and 
      $(q, \tau, q_e, \ColPop{k})\in\Delta$;
    \end{enumerate}
    now fix an arbitrary
    $\bar q=(q_i, q_e, \ExRet(\Pop{1}(w)), \Sym(w), \Lvl(w), P_1)$. 
    $\Delta_\mathcal{A}$ contains
    \begin{enumerate}
      \setcounter{enumi}{9}
    \item $(\bar q, (x,x), \bar q_0, \bot)$ for
      $\bar q_0=(q_i, q_e', \ExRet(w), \tau,
      k, P_1)$ such that there is a $q\in Q$ with 
      $(q_e', q)\in \ExHLoop(w\tau_k)$ and $(q, \tau, q_e,
      \ColPop{k})\in\Delta$;
    \item $(\bar q, (x,x), z, \bar q)$;
    \item $(\bar q, (x,x), z, \bar q_1)$ for 
      $\bar q_1=(q_i, q_e, \ExRet(\Pop{1}(w)), \sigma, l, P_2)$. 
    \end{enumerate}
  \end{iteMize}
\end{defi}

\noindent Let us explain the use of the flags $S$ ('{\bf S}earching the rightmost leaf
of the second input'), $P_1$ and $P_2$ ('{\bf P}op sequence'). 
Let $c=(q,s)$ and $c'=(p,t)$ be configurations such that
$t=\Pop{1}^k(s)$. Then we can always define nodes $d_1,d_2,d_3$ (and
an auxiliary node $e_3$)  in the
convolution of $\Encode(c) \otimes \Encode(c')$ as follows. Let $d_3$
be the rightmost leaf  of $\Encode(c)$,
let $e_3$ be the rightmost leaf of $\Encode(c')$, 
let $d_2$ be the minimal node of the rightmost path of $\Encode(c)$ which
is not in $\Encode(c')\setminus\{e_3\}$ and let $d_1$ be the
maximal node of the 
rightmost path of $\Encode(c)$ which is  on the rightmost path of
$\Encode(c')$. See Figure \ref{fig:P1P2Andd1d2d3Example} for an
example. By definition one concludes that $d_1\leq d_2 \leq d_3$. An
accepting run of $\mathcal{A}$ labels all nodes up to $d_1$ with flag
$S$, the nodes strictly between $d_1$ and $d_2$ with $P_1$ and the
nodes between $d_2$ and $d_3$ by $P_2$. Using these flags the
automaton guarantees that $t=\Pop{1}^k(s)$ for some $s$. 
Furthermore the transitions used at the nodes labelled by $P_1$ or
$P_2$ guarantee that there is a  sequence of loops and pop
operations connecting the two configurations. 

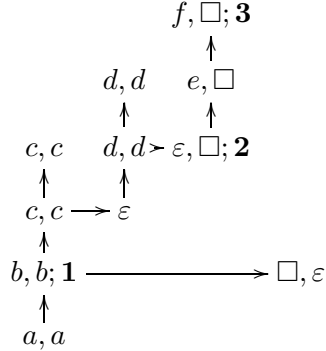
\begin{figure}
  \centering
  $\begin{xy}
    \xymatrix@R=9pt@C=3pt{
       &     & f,\Box;{\bf 3}\\
       & d,d & e,\Box \ar[u]\\
      c,c & d,d \ar[u] \ar[r] & \varepsilon,\Box;{\bf 2} \ar[u]\\
      c,c \ar[u]\ar[r]& \varepsilon\ar[u]  \\
      b,b;{\bf 1} \ar[u]\ar[rrr] & &  & \Box,\varepsilon \\
       a,a \ar[u]       
      }
  \end{xy}
  $
  \caption{Nodes $d_1,d_2,d_3$ in case of 
    $s=abcc:abcdd:abcdef$ and $t=abcc:abcdd:ab$ are marked by boldface numbers
    ${\bf 1}$,${\bf 2}$,${\bf 3}$, respectively. }
  \label{fig:P1P2Andd1d2d3Example}
\end{figure}

\begin{lem}
  Let $c_1$ and $c_2$ be configurations. $\mathcal{A}_{\relB}$ accepts
  $\Encode(c_1)\otimes \Encode(c_2)$ if and only if
  $c_2=\Pop{1}^m(c_1)$ such that $(c_1,c_2)\in {\relB}$. 
\end{lem}
\begin{proof}[Proof (sketch).]
  Assume that $\rho$ is an accepting run of
  $\mathcal{A}_{\relB}$  on $\Encode(c_1)\otimes\Encode(c_2)$.
  
  Every accepting run labels the root by $q_I$. Now an easy induction
  shows that there are nodes 
  $0\leq d_1 \leq d_2 \leq d_3$ such that the following holds. 
  \begin{iteMize}{$\bullet$}
  \item   $d_3$ is the rightmost leaf  of $\Encode(c_1)$.
  \item All nodes  $0\leq e \leq d_3$ are labelled by
    elements in 
    \begin{align*}
      M:=Q\times Q \times
      2^{Q\times Q}  \times \Sigma \times
      \{1,2\}\times \{S, P_1, P_2\}      
    \end{align*}
    such that $\pi_6(e)=
    \begin{cases}
      S &\text{if }0\leq e \leq d_1,\\
      P_1& \text{if } d_1<e \leq d_2, \\
      P_2& \text{if } d_2<e \leq d_3. \\
    \end{cases}
    $
  \end{iteMize}
  Furthermore, all nodes in 
  $\Encode(c_1) \otimes \Encode(c_2)$ to the left
  of this branch are labelled by $q_=$ which ensures that the two
  configurations agree on these nodes. 
  We distinguish the following cases:
  \begin{enumerate}[$d_1=d_2=d_3$]
  \item[$d_1=d_2=d_3$] In this case, no transition of the form
    \eqref{RelBT7} is used in the run. One easily concludes that
    $c_1=c_2$ and $(c_1,c_2)\in\relB$ is witnessed by the run of
    length $0$ connecting $c_1$ with $c_2$.
  \item[$d_1=d_2<d_3$] In this case, we prove by induction from $d_1$
    to $d_3$ that the automaton uses at $d_1$ a
    transition of the first form of \eqref{RelBT7}, between $d_1$
    and $d_3$ it uses transitions from \eqref{RelBT9}, and at $d_3$ it
    uses a transition from \eqref{RelBT8}. 
    Due to Lemma \ref{Lem:CharacterisationPop1Sequence} (first case)
    this implies 
    that $c_2=\Pop{1}^m(c_1)$ for some $m\in\N$ and
    $\LeftStack(d_1,c_1)$ is the stack of $c_2$. 
    Now, by induction from $d_3$ to $d_1$ one proves for each
    $d_1\leq e \leq d_3$ that $\pi_1(e)$ is the state of $c_1$ and
    $\pi_2(e)$ is a state such that there is a run witnessing 
    \mbox{$\left(c_1, (\pi_2(e), \LeftStack(e,c_1))\right)\in \relB$}. We
    conclude by another induction showing that $\pi_2(e)$ is the state
    of $c_2$.
  \item[$d_1<d_2<d_3$] Analogously to the previous case, we use Lemma
    \ref{Lem:CharacterisationPop1Sequence} (second case) to show that 
    \mbox{$c_2=\Pop{1}^m(c_1)$} for some $m\in\N$. 
    Induction from $d_3$ to $d_1$ shows that  for each 
    $d_1 \leq e \leq d_3$ there is some number 
    $k(e)\leq m$ such that 
    \begin{align*}
      \TOP{2}(\LeftStack(e, c_1))=\TOP{2}(\Pop{1}^{k(e)}(c_1))      
    \end{align*}
    and the
    label of $e$ is such that
    there is a run from $c_1$ to $(\pi_2(e), \Pop{1}^{k(e)}(c_1))$
    witnessing that this pair is in $\relB$. 
    Moreover, $k(d_1)=m$ and $\pi_2(d_1)$ is the state of $c_2$
    whence $(c_1,c_2)\in\relB$. 
  \end{enumerate}
  For the other direction let $c_1=(q_1,s_1),c_2=(q_2,s_2)$ be
  configurations and  $\rho$ 
  a run that witnesses  $(c_1,c_2)\in \relB$. 
  We only consider the case that 
  $\Encode(c_1)\otimes\Encode(c_2)$ is as described in 
 \eqref{ass2-Lem:CharacterisationPop1Sequence} of Lemma
  \ref{Lem:CharacterisationPop1Sequence}.
  The other case is similar.
  Let $b\in\Encode(c_2)$ be such that $b1$ is the rightmost leaf
  of $\Encode(c_2)$. Let $c$ be maximal in $\Encode(c_2)$ such that
  $c1\in \Encode(c_1)\setminus \Encode(c_2)$. Let $d$
  be the rightmost leaf of $\Encode(c_1)$. 
  Due to Lemma \ref{Lem:CharacterisationPop1Sequence}, $b < c < d$. 

  For all $x\leq d$, let $w_x:=\TOP{2}(\LeftStack(x,\Encode(c_1)))$.
  For $b\leq x \leq d$ let $i_x\in\domain(\rho)$ be minimal such that
  $\rho(i_x)=(q, \Pop{2}(s_1):w_x)$. Let
  $q_x\in Q$ be the state at $\rho(i_x)$.
  We define an accepting run $\lambda$ of $\mathcal{A}_{\relB}$ on
  $\Encode(c_1)\otimes\Encode(c_2)$ as follows.
  For all $\varepsilon< x \leq d$ let 
  \begin{align*}
    &\lambda(x):=(q_1,f_x,\ExRet(\Pop{1}(w_x)), \Sym(w_x),
    \Lvl(w_x), Y_x)\text{ where}\\
    &f_x=
    \begin{cases}
      q_2 &\text{if }x\leq b,\\
      q_x & \text{if } b<x\leq d,
    \end{cases}
    \text{ and}\\
    &Y_x=
    \begin{cases}
      S &\text{if } x\leq b,\\
      P_1&\text{if } b<x\leq c,\\
      P_2&\text{if } c<x\leq d.
    \end{cases}    
  \end{align*}
  Furthermore, we set $\lambda(b1):=(\Box,\varepsilon)$, 
  $\lambda(\varepsilon):=q_I$ and 
  $\lambda(x):=q_=$ for all other nodes
\mbox{  $x\in\Encode(c_1)\otimes\Encode(c_2)$. }
  
  Since $\rho$ decomposes as a sequence of loops, $\Pop{1}$ operations
  and collapse operations of level $1$,  it is
  straightforward to show that $\lambda$ is an accepting run of
  $\mathcal{A}_{\relB}$. 
\end{proof}

\section{Automaton for 
  Relation   \texorpdfstring{\relD}{R-Rightarrow}}
\label{Appendix:AD}
In the following definition, we use the same convention regarding
ranges of variables as in Appendix \ref{Appendix:AA}.

Before we define the automaton recognising the relation $\relD$
formally, we explain how a successful run of it will process a tree 
$\Encode(q_1,s_1)\otimes \Encode(q_2,s_2)$. 
The states of $\mathcal{A}_{\relD}$ come from the set $\{q_I, q_=, \bot\} \cup
M$ where 
\begin{align*}
  M:=Q\times Q \times \Sigma \times \{1,2\} \times
  (2^{Q\times Q}) \times (2^{Q\times Q}) \times \{R,L\} \times \{S,
  N\}.
\end{align*}
$q_I$ is the 
final state that is exclusively used to label the root. $q_=$ is the
state for all nodes in $\Encode(q_1,s_1)$ that do not
belong to the rightmost branch of this tree. This state is used to
check that $\Encode(q_1,s_1)$ and $\Encode(q_2,s_2)$ agree on this
part of the convolution.  The state $\bot$ is the
initial state
only used for marking the end of the tree, i.e., $\bot$ is the label
for the nodes in
$\left(\Encode(q_1,s_1)\otimes\Encode(q_2,s_2)\right)_+$.  
The rest of the nodes are labelled by elements from $M$. 

For some $\bar q\in M$ we write $\pi_i(\bar q)$ for the projection
to the $i$-th component.  In an accepting run, a node $d$ is
labelled by $\bar q\in M$ if the following is satisfied.
\begin{enumerate}[C1.]
\item \label{Cond:C1}
  $\pi_8(\bar q)= S$ iff $d$ is in the rightmost branch of
  $\Encode(q_1, s_1)$. $S$ stands for ``searching the rightmost leaf
  of $(q_1,s_1)$'' while $N$ stands for ``normal reachability''. 
  We ensure that $\pi_8(\bar q)=N$ iff $d$ is in 
  $\Encode(q_2, s_2) \setminus \Encode(q_1,s_1)$.
\item \label{Cond:C2}
  $\pi_7(\bar q) = R$ iff $d$ is in the rightmost branch of
  $\Encode(q_2,s_2)$ (by definition this is also the rightmost
  branch of $\Encode(q_1,s_1)\otimes \Encode(q_2,s_2)$. $R$ stands
  for ``rightmost branch'' while $L$ stands for ``left''.
\item $\pi_6(\bar q) =
  \ExLoop(\LeftStack(d,\Encode(q_2,s_2)))$. Since
  $s_1\in\Milestones(s_2)$, 
  $d\in \Encode(q_1,s_1)$ implies that
  $\pi_6(\bar q)= 
  \ExLoop(\LeftStack(d,\Encode(q_1,s_1)))$. 
\item $\pi_5(\bar q) =
  \ExRet(\LeftStack(d,\Encode(q_2,s_2)))$. Since
  $s_1\in\Milestones(s_2)$,     
  $d\in \Encode(q_1,s_1)$ implies that
   $\pi_5(\bar q)=
  \ExRet(\LeftStack(d,\Encode(q_1,s_1)))$.
\item $\pi_4(\bar q) = 
  \Lvl(\LeftStack(d, \Encode(q_2,s_2)))$.
\item \label{Cond:CNextToLast}
  $\pi_3(\bar q) = 
  \Sym(\LeftStack(d, \Encode(q_2,s_2)))$.
\item 
  Let $q_i:=\pi_1(\bar q)$ and $q_e:=\pi_2(\bar q)$.
  For $d\in
   \Encode(q_2,s_2)\setminus \Encode(q_1,s_1)$ 
  there is a run from $(q_i, \LeftStack(d, \Encode(q_2,s_2)))$ to 
  $(q_e, \InducedGenMilestone(d,\Encode(q_2,s_2)))$. 
  If $d$ is in the rightmost path of $\Encode(q_1,s_1)$,
  $q_i=q_1$ and there is a run from $(q_1,s_1)$ to $(q_e,
  \InducedGenMilestone(d, \Encode(q_2, s_2)))$. 
\end{enumerate}

\begin{defi}
  Let $\mathcal{S}=(Q, \Sigma,\Gamma, \Delta_{\mathcal{S}}, q_0)$ be
  some \CPS.
  Define the automaton $\mathcal{A}_{\relD}$ as follows. The
  set of states is contained in $\{q_I, q_=, \bot\} \cup M$ where
  \begin{align*}
    M:=Q\times Q \times \Sigma \times \{1,2\} \times
    (2^{Q\times Q}) \times (2^{Q\times Q}) \times \{R,L\} \times \{S,
    N\}.    
  \end{align*}
  $\bot$ is the
  initial state and $q_I$ is the only final state. 
  The transition relation $\Delta_{\mathcal{A}}$ of $\mathcal{A}$ contains the
  following transitions. 
  \begin{enumerate}[\phantom0 T1.]
  \item \label{TransitionAD:Init}
    $(q_I, (q_1, q_2), (q_1, q_2, \bot, 1, \ExRet(\bot),
    \ExLoop(\bot), R, S), \bot)\in\Delta_{\mathcal{A}}$ for all pairs
    $(q_1, q_2)\in Q^2$, and
  \item 
    $(q_=, (y,y), X, Y)$ for
    $X, Y \in\{q_=, \bot\}$ and $y\in
    (\Sigma\times\{1,2\})\cup\{\varepsilon\}$. 
  \end{enumerate}
  Fix some $\bar q:=( q_1, q_2, \sigma, l, \ExRet(w),
  \ExLoop(w), R, S)$.  Then we add the following transitions.
  \begin{enumerate}[\phantom0 T1.]
    \setcounter{enumi}{2}
  \item \label{T11} $(\bar q, (x,x), \bot, \bot)$ 
    if $q_1=q_2$;
  \item $(\bar q, (x,x), \bot, \bar q_1)$
    for
    \begin{align*}
    &(q_1, \sigma,  \gamma, q, \Clone{2})\in \Delta_{\mathcal{S}},\\
    &(q, q_1')\in \ExLoop(w),\text{ and}\\
    &\bar q_1=(q_1', q_2, \sigma, l, \ExRet(w), \ExLoop(w),R, N);      
    \end{align*}
  \item $(\bar q, (x,x), X, \bar q)$ for $X\in\{q_=,\bot\}$;
  \item $(\bar q, (x,x), \bar q_0, \bot)$ for 
    $\bar q_0:=(q_1, q_2, \tau, k, \ExRet(w\tau_k),
    \ExLoop(w\tau_k), R,S)$; 
  \item $(\bar q, (x,x), \bar q_0, \bar q_1)$  for 
    \begin{align*}
      &\bar q_0:=(q_1, q_2', \tau, k, \ExRet(w\tau_k),
      \ExLoop(w\tau_k), L,S),\\
      &(q_2', \tau, \gamma, q, \ColPop{k})\in \Delta_{\mathcal{S}},\\
      &(q, q_1')\in \ExLoop(w),\text{ and}\\
      &\bar q_1:=(q_1', q_2, \sigma, l, \ExRet(w), \ExLoop(w), R,N);
    \end{align*}
  \end{enumerate}    
  Fix some $\bar q:=( q_1, q_2, \sigma, l, \ExRet(w),
  \ExLoop(w), L, S)$.  Then we add the following transitions.
  \begin{enumerate}[\phantom0 T1.]
    \setcounter{enumi}{7}
  \item \label{T21} $(\bar q, (x,x), \bot, \bot)$ 
    for $(q_1, \sigma, \gamma, q,
    \Clone{2})\in\Delta_{\mathcal{S}}$ and $(q, q_2)\in \ExLoop(w)$;
  \item   $(\bar q, (x,x), \bar q_0, \bot)$ 
    for
    \begin{align*}
      &\bar q_0:=(q_1, q_2', \tau, k, \ExRet(w\tau_k),
      \ExLoop(w\tau_k), L, S),\\  
      &(q_2', \tau, q, \ColPop{k})\in\Delta_{\mathcal{S}}\text{ and}\\
      &(q, q_2)\in\ExLoop(w);      
    \end{align*}
  \item $(\bar q, (x,x), X, \bar q)$ for $X\in\{q_=,\bot\}$;
  \item $(\bar q, (x,x), \bot, \bar q_1)$ for 
    \begin{align*}
      &(q_1, \sigma, \gamma, q, \Clone{2})\in\Delta_{\mathcal{S}},\\
      &(q, q_1')\in \ExLoop(w)\text{  and}\\
      &\bar q_1:=(q_1', q_2, \sigma, l, \ExRet(w), \ExLoop(w), L, N);      
    \end{align*}
  \item $(\bar q, (x,x), \bar q_0, \bar q_1)$ 
    for 
    \begin{align*}
      &\bar q_0:=(q_1, q_2', \tau, k, \ExRet(w\tau_k),
      \ExLoop(w\tau_k), L, S),\\
      &(q_2', \tau, q, \ColPop{k})\in\Delta_{\mathcal{S}},\\
      &(q, q_1')\in \ExLoop(w)\text{ and}\\
      &\bar q_1:=(q_1', q_2, \sigma,l, \ExRet(w),
      \ExLoop(w), L, N).
    \end{align*}
  \end{enumerate}
  Fix some $\bar q:=( q_1, q_2, \sigma, l, \ExRet(w),
  \ExLoop(w), R, N)$.  Then we add the following transitions.
  \begin{enumerate}[\phantom0 T1.]
    \setcounter{enumi}{12}
  \item \label{T31}$(\bar q, (\Box, x), \bot, \bot)$ 
    if $q_1=q_2$;
  \item $(\bar q, (\Box, x), \bar q_0, \bot)$ 
    for  
    \begin{align*}
      &(q_1, \sigma, \gamma, q, \Push{\tau,k})\in\Delta_{\mathcal{S}},\\
      &(q, q_1')\in \ExLoop(w\tau_k)\text{ and}\\
      &\bar q_0:=(q_1', q_2, \tau, k, \ExRet(w\tau_k),
      \ExLoop(w\tau_k), R, N);      
    \end{align*}
  \item $(\bar q, (\Box, x), \bot, \bar q_1)$ 
    for
    \begin{align*}
      &(q_1, \sigma, \gamma, q, \Clone{2})\in\Delta_{\mathcal{S}},\\
      &(q, q_1')\in \ExLoop(w)\text{ and}\\
      &\bar q_1:=(q_1', q_2, \sigma, l, \ExRet(w),\ExLoop(w), R, N);
    \end{align*}
  \item 
    $(\bar q, (\Box, x), \bar q_0, \bar q_1)$ 
    for 
    \begin{align*}
    &(q_1, \sigma, \gamma, q, \Push{\tau,k})\in\Delta_{\mathcal{S}},\\ 
    &(q, q_1')\in \ExLoop(w\tau_k),\\ 
    &\bar q_0:=(q_1', q_2', \tau, k, \ExRet(w\tau_k),
    \ExLoop(w\tau_k), L, N),\\ 
    &(q_2', \tau,  q',\ColPop{k})\in\Delta_{\mathcal{S}},\\
    &(q', q_1'')\in \ExLoop(w)\text{ and}\\
    &\bar q_1:=(q_1'', q_2, \sigma,l, \ExRet(w),\ExLoop(w), R, N).
    \end{align*}
  \end{enumerate}
  Fix some $\bar q:=( q_1, q_2, \sigma, l, \ExRet(w),
  \ExLoop(w), L, N)$.  Then we add the following transitions.
  \begin{enumerate}[\phantom0 T1.]
    \setcounter{enumi}{16}
 \item \label{T41} $(\bar q, (\Box, x), \bot, \bot)$ 
   for
   $(q_1, \sigma, \gamma, q, \Clone{2})\in\Delta_{\mathcal{S}}$ and 
   $(q,q_2)\in\ExLoop(w)$;
  \item $(\bar q, (\Box, x), \bar q_0, \bot)$ 
     for  
     \begin{align*}
       &(q_1, \sigma, \gamma, q, \Push{\tau,k})\in\Delta_{\mathcal{S}},\\
       &(q, q_1')\in \ExLoop(w\tau_k),\\ 
       &\bar q_0:=(q_1', q_2', \tau, k,
       \ExRet(w\tau_k),\ExLoop(w\tau_k), L, N),\\ 
       &(q_2', \tau, q', \ColPop{k})\in\Delta_{\mathcal{S}}\text{ and}\\
       &(q', q_2)\in\ExLoop(w);        
     \end{align*}
  \item $(\bar q, (\Box, x), \bot, \bar q_1)$ 
    for  
    \begin{align*}
      &(q_1, \sigma, \gamma, q, \Clone{2})\in\Delta_{\mathcal{S}},\\
      &(q, q_1')\in \ExLoop(w),\text{ and}\\
      &\bar q_1:=(q_1', q_2, \sigma, l, \ExRet(w), \ExLoop(w), L, N);      
    \end{align*}
  \item $(\bar q, (\Box, x), \bar q_0, \bar q_1)$ 
    for 
    \begin{align*}
      &(q_1, \sigma, \gamma, q, \Push{\tau,k})\in\Delta_{\mathcal{S}},\\ 
      &(q, q_1')\in \ExLoop(w\tau_k),\\ 
      &\bar q_0:=(q_1',q_2',\tau,k,\ExRet(w\tau_k),\ExLoop(w\tau_k)),L,N),\\ 
      &(q_2', \tau, q', \ColPop{k})\in\Delta_{\mathcal{S}},\\
      &(q', q_1'')\in \ExLoop(w)\text{ and}\\
      &\bar q_1:=(q_1'', q_2, \sigma, l, \ExRet(w),\ExLoop(w), L, N).      
    \end{align*}
  \end{enumerate}\vspace{3 pt}
\end{defi}

\noindent The next lemma is a first step towards the proof that any accepting
run of $\mathcal{A}_{\relD}$ on a tree $\Encode(c_1)\otimes\Encode(c_2)$
witnesses the existence of some run from $c_1$ to $c_2$. 

\begin{lem}
  Let $\mathcal{S}$ be some \CPS.
  Let $(q_1, s_1), (q_2, s_2)$ be configurations and
  let $\rho$ be an accepting run on 
  $T:=\Encode(q_1, s_1)\otimes\Encode(q_2, s_2)$. 
  Then Conditions C\ref{Cond:C1}--C\ref{Cond:CNextToLast} of the
  beginning of this section
  hold and $s_1\in\Milestones(s_2)$. 
\end{lem}
The proof consists of straightforward inductions. 

\begin{lem}
  Let $\mathcal{S}$ be some \CPS.
  Let $(q_1, s_1), (q_2, s_2)$ be configurations and
  let $\rho$ be an accepting run on 
  $T:=\Encode(q_1, s_1)\otimes\Encode(q_2, s_2)$. 
  Let $T_1:= \domain(T)\setminus\domain(\Encode(q_1, s_1))$ and
  $T_2$ be the rightmost branch of $\Encode(q_1, s_1)$ without the root. 
  Furthermore, for all $d\in T_1$, let
  \mbox{$s_d:=  \LeftStack(d, \Encode(q_2,s_2))$} and for each
  $ d\in T_2$, let
  $s_d:= s_1$.
  For each $d\in T_1\cup  T_2$ we have
  $\rho(d)\in M$ and there is a run
  $\rho_{\mathcal{S}}$
  of $\mathcal{S}$ from
  $\left(\pi_1(\rho(d)), s_d\right)$ to 
  $\left(\pi_2(\rho(d)), \InducedGenMilestone(d, \Encode(q_2, s_2))\right)$
  such that for all $0<i\leq \length(\rho_\mathcal{S})$,
  $\rho_{\mathcal{S}}(i)\neq s_1$. 
\end{lem}
\begin{proof}
  The proof is by induction starting at the leaves. 
  The base cases are the following.
  \begin{iteMize}{$\bullet$}
  \item Assume that $s_1=s_2$. Due to C\ref{Cond:C1} and C\ref{Cond:C2},
    the rightmost leaf $d$
    of $T$ satisfies $(\pi_7(d), \pi_8(d))=(R,S)$. Thus, $\rho$
    applies at $d$ some transition of the form 
    T\ref{T11}. Due to the existence of this transition, we conclude
    that $\pi_1(\rho(d))=\pi_2(\rho(d)))$.
    Since $s_d=\InducedGenMilestone(d,\Encode(q_2, s_2))=s_1$, we
    conclude there is a loop of length $0$ from
    $(\pi_1(\rho(d)), s_d)$ to 
    $\left(\pi_2(\rho(d)), \InducedGenMilestone(d, \Encode(q_2, s_2))\right)$.
  \item Assume that $s_1\neq s_2$. 
    Let $d$ be the rightmost leaf of $T$, i.e., $d$ is the rightmost
    leaf  of $\Encode(q_2, s_2)$. Due to C\ref{Cond:C1} and C\ref{Cond:C2}, 
    $(\pi_7(d), \pi_8(d))=(R,N)$. Thus, $\rho$ applies a transition of
    the form T\ref{T31} whence $\pi_1(\rho(d))=\pi_2(\rho(d))$. As in
    the previous case we obtain
    $s_d=\InducedGenMilestone(d,\Encode(q_2, s_2))=s_2$ and a run of
    length $0$ connects 
    $(\pi_1(\rho(d)), s_d)$ with
    $(\pi_2(\rho(d)), \InducedGenMilestone(d, \Encode(q_2, s_2)))$
    because the two configurations agree.

    Now let
    $d\in T_1$
    be a leaf of $T$ that is not in the rightmost branch. 
    Due to C\ref{Cond:C1} and C\ref{Cond:C2}, $(\pi_7(d), \pi_8(d))=(L,N)$. 
    Thus, $\rho$ applies a transition of the form T\ref{T41}. 
    Note that $s_d=\LeftStack(d,\Encode(q_2,s_2))$ and
    $\InducedGenMilestone(d, \Encode(q_2, s_2))=\Clone{2}(s_d)$. 
    Due to the conditions on the existence of a transition of form
    T\ref{T41} one immediately concludes that there is a run from 
    $(\pi_1(\rho(d)), s_d)$ to 
    $\left(\pi_2(\rho(d)), \InducedGenMilestone(d, \Encode(q_2, s_2))\right)$.

    Now let $d$ be the rightmost  leaf in $T$ that is in  
    $\Encode(q_1, s_1)$.  
    Since $d$ is not in the rightmost branch of $T$ and due to
    C\ref{Cond:C1} and C\ref{Cond:C2}, $(\pi_7(d),
    \pi_8(d))=(L,S)$. Thus, $\rho$ applies a 
    transition of the form T\ref{T21}.  
    Note that $s_d=\LeftStack(d,\Encode(q_2,s_2))$ and
    $\InducedGenMilestone(d, \Encode(q_2, s_2))=\Clone{2}(s_d)$. 
    Due to the conditions on the existence of a transition of form
    T\ref{T21} one immediately concludes that there is a run from 
    $(\pi_1(\rho(d)), s_d)$ to 
    $\left(\pi_2(\rho(d)), \InducedGenMilestone(d, \Encode(q_2, s_2))\right)$.
  \end{iteMize}
  Note that all the runs obtained in the base cases do not visit 
  the stack $s_1$ except for the first configuration in the run
  associated to the rightmost leaf of $\Encode(q_1,s_1)$.

  Analogously to the base case, the inductive step consists of a
  lengthy but rather straightforward case distinction.  
  Instead of stating all cases, we mention
  the crucial ideas underlying the proof. 
  \begin{iteMize}{$\bullet$}
  \item For $d\in\{0,1\}^*$ and $i\in\{0,1\}$ such that
    $di$ is in the rightmost branch of $\Encode(q_1,s_1)$, then
    $s_d=s_{di}=s_1$. Thus, the run associated to $di$ serves as
    initial part of the run associated to $d$.
  \item For $d\in\Encode(q_2, s_2)\setminus \Encode(q_1, s_1)$ or $d$
    the rightmost leaf of $\Encode(q_1,s_1)$,
    let 
    $i\in\{0,1\}$  minimal such that $di\in\Encode(q_2, s_2)$. 
    Then $\LeftStack(d, (q_2, s_2))$ and 
    $\LeftStack(di, (q_2, s_2))$ differ in one stack operation
    $\op$.  The transition of $\mathcal{A}_{\relD}$ used at $d$
    ensures that there is a run from 
    $\left(\pi_1(\rho(d)), \LeftStack(d, (q_2, s_2)\right)$ to
    $\left(\pi_1(\rho(di)), \LeftStack(di, (q_2, s_2))\right)$ that
    performs this operation $\op$ followed by a loop. The
    composition of this run with the run associated to $di$ serves as
    initial part of the run associated to $d$.
  \item If $d$ is in the rightmost branch of $\Encode(q_2, s_2)$ and
    $i\in\{0,1\}$ is maximal such that $di\in \Encode(q_2,s_2)$, then 
    $\InducedGenMilestone(d, \Encode(q_2,s_2))=
    \InducedGenMilestone(di, \Encode(q_2,s_2))$ whence the run
    associated to $di$ serves as final part of the run associated to
    $d$.
  \item If $d1\in \Encode(q_2, s_2)$, then 
    $\InducedGenMilestone(d, \Encode(q_2,s_2))=
    \InducedGenMilestone(d1, \Encode(q_2,s_2))$ whence the run
    associated to $d1$ serves as final part of the run associated to
    $d$.
  \item If $d$ is not in the rightmost branch of $\Encode(q_2, s_2)$
    and $d1\notin \Encode(q_2, s_2)$, then we have
    $\InducedGenMilestone(d, \Encode(q_2,s_2))=
    \Pop{1}(\InducedGenMilestone(d0, \Encode(q_2,s_2)))$. 
    Furthermore, the transition of $\mathcal{A}_{\relD}$ used at $d$ ensures
    that there exists a run from 
    $(\pi_2(\rho(d0)), \InducedGenMilestone(d0, \Encode(q_2, s_2))$ to
    $(\pi_2(\rho(d)),$ $\InducedGenMilestone(d, \Encode(q_2, s_2)))$. This run
    serves as final part of the run associated to $d$.
  \item If $d0, d1\in \Encode(q_2, s_2)$, then 
    $\LeftStack(d1, (q_2, s_2))=
    \Pop{1}(\InducedGenMilestone(d0, \Encode(q_2,
    s_2)))$. Furthermore, the existence of the
    transition of $\mathcal{A}_{\relD}$ used at $d$ ensures that there is a
    run from 
    $\left(\pi_2(\rho(d0)), \InducedGenMilestone(d0, \Encode(q_2,
      s_2))\right)$ to 
    $\left(\pi_1(\rho(d1)), \LeftStack(d1, (q_2, s_2))\right)$. 
    This run is used to connect the initial part induced by $d0$ with
    the final part induced by $d1$ in order to obtain the run
    associated to $d$. \qedhere
  \end{iteMize}
\end{proof}

\begin{rem}
Due to the transitions of the form T\ref{TransitionAD:Init} any
accepting run of $\mathcal{A}_{\relD}$ on a tree
$\Encode(q_1,s_1)\otimes\Encode(q_2,s_2)$ 
satisfies $(\pi_1(\rho(0)), \pi_2(\rho(0)))=(q_1,q_2)$. Moreover,
recall that 
\mbox{$\InducedGenMilestone(0, \Encode(q_2,s_2))=s_2$}.
Thus, the lemma implies that
$\left( (q_1,s_1),(q_2,s_2)\right)\in \relD$ if there is an accepting
run of $\mathcal{A}_{\relD}$ on
$\Encode(q_1,s_1)\otimes\Encode(q_2,s_2)$.
\end{rem}

\begin{lem}
  Let $\mathcal{S}$ be some \CPS. Let $c_1:=(q_1, s_1)$, $c_2:=(q_2,
  s_2)$ be 
  configurations such that $s_1=\Pop{2}^k(s_2)$ for some $k\in\N$. Let
  $\rho_{\mathcal{S}}$ be a run from $(q_1, s_1)$ to $(q_2, s_2)$ 
  witnessing $(c_1,c_2)\in \relD$. 
  Then there is an accepting run of $\mathcal{A}$ on 
  $T:=\Encode(q_1, s_1)\otimes \Encode(q_2, s_2)$. 
\end{lem}
\begin{proof}
  We define the accepting run $\rho$ as follows.
  Set $\rho(\varepsilon):= q_I$, $\rho(d):=\bot$ for all 
  $d\in T_+$, $\rho(d):= q_=$ for all 
  $d\in \Encode(q_1, s_1)\setminus B$ where
  $B$ is the rightmost branch of $\Encode(q_1, s_1)$, and for all
  other $d\in\domain(T)$, set
  \begin{align*}
    \rho(d):=(q_i^d, q_e^d, \Sym(s'),  \Lvl(s'), \ExRet(s'),
    \ExLoop(s'),   X, Y)    
  \end{align*}
  where $s':=\LeftStack(d, \Encode(q_2,s_2))$ if $d\notin B$ and
  $s'=s_1$ if $d\in B$ and $q_i^d, q_e^d, X$ and $Y$ are
  defined as follows.  
  \begin{iteMize}{$\bullet$}
  \item $q_i^d=q_1$ if $d\in \Encode(q_1,s_1)$. 
    Otherwise let  
    $j\in\domain(\rho_{\mathcal{S}})$ be maximal such that
    $\rho_{\mathcal{S}}(j)=(q,\LeftStack(d,\Encode(q_2,s_2)))$ for
    some $q\in Q$. Set $q_i^d:=q$.
  \item Let $j\in\domain(\rho_{\mathcal{S}})$ be maximal such that 
    $\rho_{\mathcal{S}}(j)=(q,\InducedGenMilestone(d,\Encode(q_2,s_2)))$
    for some $q\in Q$. Set  $q_e^d:=q$. 
  \item Set $X=R$ if $d$ is in the rightmost branch of $T$ and 
    set $X=L$ otherwise.
  \item Set $Y=S$ if $d$ is in the rightmost branch
    of $\Encode(q_1, s_1)$ and set $Y=N$ otherwise. 
  \end{iteMize}
  A straightforward, but tedious induction shows that $\rho$ is
  accepting on $T$.   It relies on the decomposition result for runs
  witnessing $(c_1,c_2)\in \relD$ from Corollary \ref{Cor:OrderEmbedding}.
\end{proof}

\section{Modifications for the 
  Proof of Proposition \ref{Prop:ReachLregular} 
  (cf.~page \pageref{Proof:ReachLregular})}
\label{Appendix:ReachL}
We replace the automaton $\mathcal{A}_{\relA}$ in the construction
of $\Reach$ on the product of the pushdown system with the automaton
for the regular language $L$ with
the following version $\mathcal{A}_{{\relA}'}$.
Let $Q_L$ be the states of a finite automaton recognising $L$, $i_0\in
Q_L$ be its initial state and
$F\subseteq Q_L$ be its final states. 
We replace transitions of the form
$(q_I, (q_1,q_2),  (S, q_1, q_2, \bot, 1, \ExRet(\bot_2)),   \bot)$
by transitions of the form
$(q_I, (q_1,q_2), (S, (q_1,i_0), (q_2,\hat q), \bot, 1, \ExRet(\bot_2)),
\bot)$ for $q_1,q_2\in Q$ and $\hat q\in Q_L$. 
Constructing $\varphi$ with $\mathcal{A}_{{\relA}'}$ instead of
$\mathcal{A}_{\relA}$ ensures that $T_1$ encodes some configuration
$(q,s)$ where the state $q$ is in $Q$, but it checks for runs
starting in $((q,i_0), s)$. 

Furthermore, we replace the automaton $\mathcal{A}_{\relD}$ 
with the version $\mathcal{A}_{{\relD}'}$ where we replace the transitions
of the form 
$(q_I, (q_1, q_2), (q_1, q_2, \bot, 1, \ExRet(\bot), \ExLoop(\bot), R,
S), \bot)\in\Delta_{\mathcal{A}}$ 
with the transitions
$(q_I, (q_1, q_2), ((q_1, \hat q), (q_2, q_f), \bot, 1, \ExRet(\bot),
\ExLoop(\bot), R,  S), \bot)\in\Delta_{\mathcal{A}}$ 
for $q_1,q_2\in Q$, $\hat q\in Q_L$ and $q_f\in F$. 
Constructing $\varphi$ with $\mathcal{A}_{{\relD}'}$ instead of
$\mathcal{A}_{\relD}$ ensures that $T_2$ encodes some configuration
$(q,s)$ where the state $q$ is in $Q$, but it checks for runs
ending in $((q,q_f), s)$ for some final state $q_f$ of
$\mathcal{A}_L$.

\end{document}